\documentclass[11pt,reqno]{amsart}
\usepackage[letterpaper, portrait, margin=1in]{geometry}
\makeatletter
\g@addto@macro{\endabstract}{\@setabstract}

\usepackage{amsmath}
\usepackage{latexsym}
\usepackage{amsfonts}
\usepackage{amssymb}
\usepackage{lipsum}

\usepackage{bbm,dsfont}
\usepackage{dsfont}
\usepackage{graphicx}
\usepackage{hyperref}

\usepackage{verbatim}

\newtheorem{proposition}{Proposition}
\newtheorem{proposition?}{Proposition?}
\newtheorem{theorem}{Theorem}
\newtheorem{lemma}{Lemma}
\newtheorem{corollary}{Corollary}

\newtheorem{remark}{Remark}

\newtheorem{example}{Example}
\newtheorem{definition}{Definition}
\newcommand{\authorfootnotes}{\renewcommand\thefootnote{\@fnsymbol\c@footnote}}%
\makeatother

\usepackage{color}

\newcommand{\blue}[1]{\textcolor{blue}{#1}}

\usepackage[normalem]{ulem} 

\newcommand{\po}{{\sc pom }}

\newcommand{\mods}[1]{\left \vert #1 \right \vert ^2}
\newcommand*\colvec[3][]{
    \left(\begin{smallmatrix}\ifx\relax#1\relax\else#1\\\fi#2\\#3\end{smallmatrix}\right)
}

\newcommand{\hi}{\mathcal{H}} 
\newcommand{\his}{\mathcal{H}_{\mathcal{S}}}
\newcommand{\hir}{\mathcal{H}_{\mathcal{R}}}
\newcommand{\hia}{\mathcal{H}_{\mathcal{A}}}

\newcommand{\hs}{\mathcal{H}_{\mathcal{S}}}
\newcommand{\hit}{\mathcal{H}_{\mathcal{T}}}
\newcommand{\Y}{\yen}

\newcommand{\lh}{\mathcal{L(H)}} 
\newcommand{\lhs}{\mathcal{L}(\mathcal{H}_{\mathcal{S}})} 
\newcommand{\lht}{\mathcal{L}(\mathcal{H}_{\mathcal{T}})} 

\newcommand{\lhr}{\mathcal{L}(\hir)} 

\newcommand{\ip}[2]{\left\langle\,#1\,{|}\,#2\,\right\rangle} 
\newcommand{\ket}[1]{|#1\rangle} 
\newcommand{\bra}[1]{\langle#1|} 
\newcommand{\state}[1]{|#1\rangle} 
\newcommand{\dstate}[1]{\langle#1|} 
\newcommand{\kb}[2]{|#1\rangle\langle#2|} 
\newcommand{\no}[1]{\left\|#1\right\|} 
\newcommand{\nos}[1]{\left\|#1\right\|^2} 
\newcommand{\tr}[1]{\textrm{tr}\left[#1\right]} 
\newcommand{\id}{\mathbbm{1}} 
\newcommand{\nul}{O} 
\newcommand{\fii}{\varphi}
\newcommand{\Esf}{\mathsf{E}}

\newcommand{\Fsf}{\mathsf{F}}

\newcommand{\Sy}{\mathcal{S}}
\newcommand{\Ap}{\mathcal{A}}
\newcommand{\Bp}{\mathcal{B}}

\newcommand{\R}{\mathcal{R}}
\newcommand{\T}{\mathcal{T}}

\newcommand{\var}{\textrm{Var}} 

\newcommand{\E}{\mathsf{E}}
\newcommand{\F}{\mathsf{F}}
\renewcommand{\P}{\mathsf{P}}

\newcommand{\Rmb}{\mathbb{R}}

\begin{document}


\title{Symmetry, Reference Frames, and Relational Quantities in Quantum Mechanics}
\maketitle
\begin{center}

  \normalsize
  \authorfootnotes
  Leon Loveridge\footnote{l.d.loveridge@uu.nl}\textsuperscript{1}, 
Takayuki Miyadera\footnote{miyadera@nucleng.kyoto-u.ac.jp}\textsuperscript{2} and
  Paul Busch\footnote{paul.busch@york.ac.uk}\textsuperscript{3},

  \textsuperscript{1}Department of Mathematics and Descartes Centre for the History and Philosophy of Science and the Humanities, Utrecht University, 3584 CC Utrecht, The Netherlands   \par
  \textsuperscript{2}Department of Nuclear Engineering, Kyoto University, Nishikyo-ku, Kyoto, Japan 615-8540 \par
  \textsuperscript{3}Department of Mathematics, University of York, Heslington, York, UK. YO10 5DD\par \bigskip

\end{center}

\date{\today}

\begin{abstract}
We propose that observables in quantum theory  are properly understood as representatives of
symmetry-invariant quantities relating one system to another, the latter to be called a reference system. We provide a rigorous mathematical language to introduce and study quantum reference systems, showing that the orthodox ``absolute" quantities are good representatives of observable relative quantities if the reference state is suitably localised. We use this relational formalism to critique the literature on the relationship between reference frames and superselection rules, settling a long-standing debate on the subject.

\end{abstract}

\section{Introduction}

In classical physics, symmetry, reference frames and the relativity of physical quantities are intimately connected. The position of a material object is defined as relative to a given frame, 
and the relative position of object to frame is a shift-invariant quantity. Galiliean directions/angles, velocities and time of events are all relative, and invariant only once the frame-dependence
has been accounted for. The relativity of these quantities is encoded in the Galilei group, 
and the observable quantities are those which are invariant under its action. Einstein's theory engendered a deeper relativity---the length of material bodies and 
time between spatially separated events are also frame-dependent quantities---and observables
must be sought in accordance with their invariance under the action of the Poincar\'{e} group.

In quantum mechanics the analogues of those quantities mentioned above (e.g., position, angle)
must also be understood as being relative to a reference frame. As in the normal presentation of the classical theory, the reference frame-dependence is implicit. However, in the quantum case, there arises an ambiguity regarding the definition of a reference frame: if it is classical,
this raises the spectre of the lack of universality of quantum mechanics along with technical difficulties surrounding hybrid classical-quantum systems; if quantum,
such a frame is subject to difficulties of definition and interpretation arising from 
indeterminacy, incompatibility, entanglement, and other quantum properties (see, e.g., \cite{ed1,ak1,ak2} for early discussions of some of the important issues).

In previous work \cite{mlb, lbm}, following classical intuition we have posited that observable quantum quantities are invariant under relevant symmetry transformations,
and examined
the properties of quantum reference frames (viewed as physical systems) which allow for the usual description, in which the reference frame is implicit, to be recovered. We constructed a map $\Y$ which brings out 
the relative nature of quantities normally presented in ``absolute" form in conventional treatments,
which allows for a detailed study of the relativity of states and observables in quantum mechanics
and the crucial role played by reference localisation.

The objectives for this paper are: 1) To provide a mathematically rigorous and conceptually clear framework with which to discuss quantum reference frames, making precise existing work
on the subject (e.g., \cite{brs}) and providing proofs of the main claims in \cite{lbm}; 2) to construct examples,
showing how symmetry dictates that the usual text book formulation of quantum theory
describes the relation between a quantum system and an appropriately localised reference system; 3) to provide further conceptual context for the quantitative trade-off relations proven in \cite{mlb}; 4) to provide explicit and clear explanation of what it means for states/observables to be defined relative to an external reference frame, and show how such an external description is compatible with quantum mechanics as a universal theory; 5) to introduce the concepts of absolute coherence and mutual coherence, showing the latter to be required for good approximation of relative quantities by absolute ones, and demonstrating it to be the crucial property for interference phenomena to manifest in the presence of symmetry; 6) to address the questions of dynamics and measurement under symmetry, offering an interpretation of the Wigner-Araki-Yanase theorem based on relational quantities; 7) to analyse simplified models similar to those appearing in the literature purporting to produce superpositions typically thought ``forbidden" due to superselection rules, and provide a critical analysis of large amplitude limits in this context guided by two interpretational principles due to Earman and Butterfield, leading directly to 8) to provide a historical account of two differing views on the nature of superselection rules (\cite{www, www70} ``versus" \cite{lub1, as, brs}), their fundamental status in quantum theory and precisely what restrictions arise in the presence of such a rule, showing how our framework brings a unity to the opposing standpoints; 9) to remove ambiguities and inconsistencies appearing in all previous works on the subject of the connection between superselection rules and reference frames;  10) to offer a fresh perspective, based on the concept of mutual coherence, on the nature and reality of quantum optical coherence, settling a long-standing debate
on the subject of whether laser light is ``truly" coherent. See also \cite{dia} for an important contribution on this topic. We provide general arguments and many worked examples to show precisely how the framework presented works in practice, and which 
simplify a number of models appearing in the literature.

Our paper constitutes further effort in a long line of enquiries (e.g., (\cite{brs,bor1,rov1,bene1,mermrel1,AR,sp1}) aimed at capturing the relationalism at the heart of the quantum mechanical world view. The fundamental role of symmetry has not impressed itself strongly upon previous
consideration of the relative nature of the quantum description, and we view this work (along with
\cite{mlb,lbm}) as opening new lines of enquiry in this direction. Our work is inspired by \cite{brs} and visits similar themes, and is complementary to 
recent work on resource theories (e.g., \cite{brs,ms1,ajr1,pian1,pian2}), which focus primarily on practical questions surrounding, for example, high-precision quantum metrology.

We now provide standard mathematical background material, and will work in units where
$\hbar = 1$.

\section{Notation and Some Definitions}\label{sec:nad}

\subsection{Observables and States}
Associated to each physical system is a
separable complex Hilbert space $\hi$.
We let $\lh$ denote the ($C^*$/von Neumann) algebra of all bounded linear operators in $\mathcal{H}$.
\begin{definition} Let $(\Omega, \mathcal{F})$ denote the measurable space consisting of a $\sigma$-algebra $\mathcal{F}$ of subsets of some set $\Omega$.
A normalised {\em positive operator valued measure} \emph{(}{\sc pom}\emph{)} $\Esf$ on $(\Omega, \mathcal{F})$ is a mapping $\Esf: \mathcal{F} \to \lh$  for which
\begin{enumerate}
\item $\Esf (\Omega) = \id$,
\item $\Esf (X) \geq 0 $ for all $X \in \mathcal{F}$,
\item $\Esf \left(\bigcup{X_i} \right) = \sum \Esf(X_i)$ for disjoint sequences $X_i \subset \mathcal{F}$ (sum converging weakly).
\end{enumerate}
\end{definition}
\noindent(Here $\le,\ge$ denote the standard operator ordering.)

Normalised {\sc pom}s represent \emph{observables} (subject to extra constraints in the presence of symmetry, discussed below). Throughout this paper, the pair $(\Omega, \mathcal{F})$ will normally correspond to $\left( \mathbb{R}^n, \mathcal{B}(\mathbb{R}^n) \right)$ (or possibly subsets/subalgebras) with $\mathcal{B}(\cdot)$ denoting the Borel sets. The operators $\Esf(X)$
are called {\em effects} (occasionally also {\em {\sc pom} elements} or {\em effect operators}); they satisfy $\mathds{O} \leq \Esf(X) \leq \id$. The unit operator interval $\left[ \mathds{O}, \id \right]$ comprises the set of all effects $\mathcal{E}(\hi)$. $\mathcal{E}(\hi)$ is convex as a subset of the real linear space of self-adjoint operators in $\lh$, and the collection of extremal elements is the set of projections, characterised as the idempotent effects.
If all elements of a {\sc pom} $\Esf$ are idempotent, then $\Esf$ is called a projection valued measure ({\sc pvm}), and if   $\Esf$ is defined on $\Rmb$, it  defines a unique self-adjoint operator $A:=\int  \Esf (d \lambda)$ with spectral measure $\Esf ^{A}\equiv\Esf$. An observable defined by a self-adjoint operator, or equivalently, a {\sc pvm}, will be called \emph{sharp}, and all others \emph{unsharp}.

\begin{definition}\label{def:norm}
A positive linear map $\omega : \lh \to \mathcal{A}$ \emph{(}where $\mathcal{A}$ is a von Neumann algebra\emph{)} is called
\emph{normal} if for any increasing net $(A_{\alpha}) \subset \lh$ with $\sup{\{A_{\alpha}\}}=A$, $\omega(A) = \sup{\{\omega (A_{\alpha})\}}$.
\end{definition}
Normality is equivalent to $\sigma$-weak continuity. We will denote the trace class of $\mathcal{H}$ by $\mathcal{L}_1(\hi)$ and the trace functional by $\tr{\cdot}$.
\emph{Normal states} are then obtained by setting $\mathcal{A}=\mathbb{C}$ in Definition \ref{def:norm}; any normal state is of the form $A \mapsto\tr{\rho A}\equiv\langle A\rangle_\rho$, where $\rho \in \mathcal{L}_1(\hi)$ is a positive operator and $\tr{\rho}=1$.
 The set of normal states,
denoted $\mathcal{S}(\hi)$, is (identified with) a $\sigma$-convex subset of the real vector space $\mathcal{L}_1(\hi)_{sa}$ of self-adjoint elements of $\mathcal{L}_1(\hi)$. Henceforth all states are assumed to be normal, and we freely move between algebraic (linear functional) and spatial (density operator) notions of states.
The extreme points of $\mathcal{S}(\mathcal{H})$, corresponding to the pure states, are given by the rank one projections, which will be denoted $P_{\varphi} \equiv \ket{\varphi}\bra{\varphi}$, where $\varphi\in\hi$, $\|\varphi\|=1$. We will usually identify pure normal states with unit vectors in $\mathcal{H}$.
States generate expectation-valued functionals $\mathcal{L}(\hi)_{\text{sa}} \to \mathbb{R}$ on $\mathcal{L}(\hi)_{\text{sa}}$ --- the self-adjoint part of $\lh$ --- and when restricted to $\mathcal{E}(\hi)$ can be viewed as generalised  probability measures $\mathcal{E}(\hi) \to [0,1]$.
For a given {\sc pom} $\E:\mathcal{F} \to \lh$ and $\rho \in \mathcal{L}_1(\hi)$ we will write 
$X \mapsto p^E_{\rho}(X)$ for the probability measure $X \mapsto \tr{\E(X)\rho}$ and if $\Esf = \Esf^A$ we use the shorthand $X \mapsto p^A_{\rho}(X)$ to represent the measure $X \mapsto \tr{\E^A(X)\rho}$.

\subsection{Covariant {\sc pom}s and Localisability}\label{sec:cpl}

Covariant {\sc pom}s will feature as reference quantities in the sequel, and their
localisation properties will play an important role. We review these basic notions here.

\subsubsection{Systems of Covariance, Norm-1 Property}
\begin{definition}
Let $U$ denote a unitary representation of a locally compact group
$G$, and let $\Fsf : \mathcal{F} \to \lh$ be a POM whose outcome space $\Omega$ is a $G$-space. Then $(U,\Fsf,\mathcal{H})$
is a system of covariance for $G$ if
\begin{equation}\label{eq:gencov}
\Fsf(g.X) = U(g)\Fsf(X)U(g)^* \text{~for all~} g \in G, X \in \mathcal{F}.
\end{equation}
\end{definition}
\noindent $\Fsf$ is called a \emph{covariant} {\sc pom} for $U$. The triple $(U,\Fsf,\mathcal{H})$
is called a \emph{system of imprimitivity} if $\Fsf$ is projection-valued. We often consider the case $\Omega = G$ and $G$ abelian.

\begin{remark}\rm
Systems of covariance/imprimitivity may also be defined for projective representations.
\end{remark}

We give a definition relating to the \emph{localisability} of {\sc pom}s---the so-called
norm-1 property (see, e.g., \cite{norm1}):
\begin{definition} \label{norm1}
A {\sc pom} $\Esf: \mathcal{B}(G) \to \lh$ is said to satisfy the {\em norm-1 property} if  
$\no{\Esf(X)}=1$ for all $X$ for which $\Esf(X) \neq 0$. 
\end{definition}
\noindent The following is an immediate consequence.
\begin{lemma}\label{lem:n1} If $\Esf$ satisfies the norm-1 property, then for any $X$ for which $\Esf (X) \neq 0$,
there exists a sequence of unit vectors $(\varphi _k) \subset \hi$ for which $\lim_{k \to \infty} \ip{\varphi_k} {\Esf (X)\varphi _k} = 1$. 
\end{lemma}
This entails that such a \po gives rise to probability distributions which are (approximately) {\em localisable} in every set $X$ for which $\Esf(X)\ne 0$. In comparison, for a projection valued measure $\mathsf{P}$, for any $X$ with $\mathsf{P}(X) \neq 0$, there is a unit vector
$\varphi \in \mathcal{H}$ for which $\ip{\varphi}{ \mathsf{P} (X)\varphi } = 1$ (any unit vector in the range of  $\mathsf{P} (X)$ will have this property). Hence {\sc pom}s with the norm-1 property have, in a limiting sense, the localisability properties possessed by all {\sc pvm}s.

\begin{remark}\rm
In the case of a covariant {\sc pom}, we do not need to check all the 
subsets $X$ to confirm the norm-1 property, as shown in the following lemma. 
\begin{lemma}
Let $\E$ be a covariant {\sc pom} with $\Omega = G$.
The following are equivalent. 
\begin{itemize}
\item[(i)]$\E$ has the norm-1 property.  
\item[(ii)]$\Vert E(X) \Vert =1$ for all $X \in \mathcal{F}$ 
with $e\in X$ and $\E(X)\neq 0$.
\item[(iii)]For all $X\in \mathcal{F}$ with $e \in X$ and $\E(X) \neq 0$, 
there exists a sequence of unit vectors 
$(\varphi_k)$ such that $\lim_k \langle \varphi_k| \E(X) \varphi_k\rangle =1$.
\end{itemize}
\end{lemma}
\begin{proof}
Assume (ii). Then for an arbitrary $Y \in \mathcal{F}$ with 
$\E(Y) \neq 0$, there exists $g\in G$ such that 
$e\in g.Y \in \mathcal{F}$ holds. Thus (i) follows. 
The other relations are trivial.\footnote{The restriction on $X \in \mathcal{F}$ for $\E(X) \neq 0$ is 
not needed in the case of $S^1$ and $\mathbb{R}^d$.
Let $\E$ be a covariant {\sc pom} with a compact $\Omega = G$.
Then the following are equivalent. 
\begin{itemize}
\item[(i)]$\E$ has norm-1 property.  
\item[(ii)]$\Vert E(X) \Vert =1$ for all $X \in \mathcal{F}$ 
with $e\in X$.
\item[(iii)]For all $X\in \mathcal{F}$ with $e \in X$, 
there exists a sequence of unit vectors 
$(\varphi_k)$ such that $\lim_k \langle \varphi_k| \E(X) \varphi_k\rangle =1$.
\end{itemize}
\begin{proof}
Assume (ii). Suppose that there is a neighbourhood $X$ of $e$ 
satisfying $\E(X) =0$. $G$ is covered by $\{g.X\}$, which can be 
reduced to a finite cover $\{g_n. X\}$. 
$\E(G) \leq \sum_n \E(g_n. X) =0$ gives a contradiction. 
\end{proof}
}
\end{proof}
\end{remark}

\subsubsection{Positions, Momenta, and Covariant Phase Space \sc{pom}s}
For the case $G=\mathbb{R}$, the position operator $Q$ with spectral measure $\E^Q : \mathcal{B}(\mathbb{R}) \to \mathcal{L}(L^2(\mathbb{R}))$ acting by multiplication, and
the strongly continuous representation $U(x)=e^{ixP}$ ($P$ momentum) of $\mathbb{R}$ gives rise to a system
of imprimitivity $(U, \E^Q, L^2(\mathbb{R}))$ under the covariance 
\begin{equation}
\E^Q(X-x) = U(x)\E^Q(X)U(x)^*.
\end{equation}
The momentum operator $P$
with spectral measure $\E^P: \mathcal{B}(\mathbb{R}) \to \mathcal{L}(L^2(\mathbb{R}))$ satisfies, with $V(p)=e^{ipQ}$, the following covariance relation with respect to boosts:
\begin{equation}
\E^P(Y-p) = V(p)\E^P(Y)V(p)^*,
\end{equation}
yielding the system of imprimitivity $(V,\E^P, L^2(\mathbb{R}))$.

Unsharp versions of position (for instance, smeared positions (e.g., \cite{OQP})) are also covariant; indeed it is such a covariance requirement
that defines the class of unsharp positions (analogously for unsharp momenta). Let $\mu$ be a probability (``confidence") measure on $\mathbb{R}$. A smeared position observable
$\Esf ^{\mu}$ is defined as
\begin{equation}
\Esf ^{\mu} (X) := (\mu * \E)(X) = \int_R \Esf (X + q) d\mu (q)
\end{equation}
where ``$*$'' denotes convolution of measures. Under the assumption of absolute continuity we can write $\mu(X) = \int_X e(x) dx$ and
$\Esf ^{\mu}(X) \equiv \Esf ^{e} (X) = (\chi_X * e)(Q)$. Such a quantity is called a smeared position observable with confidence function $e$. It is covariant, and in the limit that $e$ becomes a delta function, or equivalently, the associated $\mu$ becomes a point measure, the sharp position is returned.

An example of a necessarily unsharp covariant quantity is provided by a covariant phase-space
{\sc pom} $M:\mathcal{B}(\mathbb{R}^2) \to \mathcal{L}(L^2(\mathbb{R}))$ which is both shift and boost covariant, i.e.,
\begin{equation}\label{eq:cpo}
W(q,p)M(Z)W(q,p)^*=M(Z+(q,p)),
\end{equation}
where $W(q,p):=e^{(-i/2 )qp} e^{-iqP} e^{ipQ}$ are the Weyl operators. $M$ contains unsharp positions and momenta as marginals.




We now turn to the case of $G = S_1$ which plays a major role in the rest of the paper.

\subsubsection{Covariant Phases}\label{ssec:covp}

We identify $S^1$ with $[0,2 \pi]$ or occasionally $[- \pi, \pi]$ (identifying also the endpoints of these intervals); $\theta\mapsto U(\theta)=e^{i N\theta}$ is a strongly 
continuous unitary representation, for self-adjoint $N$, of $S_1$ in $L^2(S^1)$ and
 $\Fsf:\mathcal{B}(S_1) \to \mathcal{L}(L^2(S^1))$ is called a  covariant phase \po if 
\begin{equation}\label{eq:phasedef}
e^{i\theta N }\mathsf{F}(X)e^{-i \theta N } =\Fsf(X\dotplus\theta), \quad \theta\in [0,2\pi),\ X\in\mathcal{B}([0,2 \pi))
\end{equation}
(where $\dotplus$ denotes addition modulo $2 \pi$).
There is a constraint on the spectrum of the unique self-adjoint generator $N$. Using the spectral representation of $N$, $N=\int x \E^N(dx)$, we have $\int e^{i2\pi x}\E^N(dx)=\id$ so that the spectrum of $N$ must consist of integers. Recall that the generator $N$ associated with a phase shift group is called a {\em number operator}. We consider three typical cases of generators $N$ and covariant phase {\sc pom}s associated with them.

\begin{example}\rm
Consider the canonical pair of an angular momentum component and the associated angle variable of a particle in three dimensions. In this case, $N=L_z$ (say), where $L_z$ generates rotations about the $z$ axis,  and as a covariant phase \po one can take the spectral measure of the self-adjoint azimuthal angle operator, $\Esf=\Esf^\Phi$, where $\Phi \psi(r,\theta,\phi)=\phi \psi(r,\theta,\phi)$. Note that the spectrum of $N$ is $\mathbb{Z}$.
\end{example}

\begin{example}\rm
The second example is motivated by the number operator counting the eigenvalues of the harmonic oscillator Hamiltonian. The associated phase {\sc pom}s are covariant under rotations in phase space. These cannot be {\sc pvm}s due to the fact that $N$ is bounded from below \cite{hol1, ljp1}; the preceding example then figures naturally as the minimal Naimark extension of the canonical phase (defined presently). Thus, let $\{e_n \}$ be an orthonormal basis in $\hi\simeq \ell^2$ and $N := \sum_{n=0}^\infty n P[e_n] \equiv \sum_{n=0}^\infty n P_n$ be a number operator. Any covariant phase {\sc pom} conjugate to $N$ is known to be of the form 
 \begin{equation}\label{phase}
\Fsf (X) = \sum_{n,m=0}^{\infty} c_{nm} \frac{1}{2 \pi}  \int_{X} e^{i(n-m) \theta} d \theta \kb{n}{m}
\end{equation}
where $(c_{nm})$ is a so-called {\em phase matrix} --- a positive matrix for which $c_{nn}=1$ for all $n \in \mathbb{N}$. The \emph{canonical phase} $\Fsf^{\text{can}}$ is singled out by the condition $c_{n,m} = 1$ for all $n,m \in \mathbb{N}\cup\{0\}$. $\Fsf ^{\rm{can}}$ is characterised by various optimality properties (see \cite{pekju}), in particular it satisfies the norm-1 property.
\end{example}

\begin{example}\rm
As the third example we consider covariant phase {\sc pom}s  in finite dimensional Hilbert spaces; 
one such instance is the spin phase (e.g., \cite{OQP}). The norm-1 property and the associated localisability are lost when we move to the finite dimensional setting, as shown in Lemma \ref{loclem} below.

Let $\hi\simeq \mathbb{C}^d$ and consider the operator $N \in \lh$ defined by $N = \sum _{n = 0}^{d-1}nP_n$. An example of a covariant phase  {\sc pom} is given by
\begin{equation}\label{finphase}
\Fsf (X) = \sum_{n,m = 0} ^{d-1}\frac{1}{2 \pi}\int_X e^{i(n-m) \theta} \kb{m}{n} d \theta.
\end{equation}
\end{example}
For a set $X\in\mathcal{B}([0,2\pi))$, we denote its Lebesgue measure by $|X|$.

\begin{lemma}(Localisation Lemma)\label{loclem}
Consider a covariant phase {\sc pom} $\Fsf$ in a $d$-dimensional Hilbert space $\hi$. For any $X \in\mathcal{B}([0,2\pi))$ $(X \neq [0,2 \pi) )$ and for any state $\rho$, it holds that 
\begin{align*}
{\rm tr}[{\rho \Fsf(X)}]\leq d|X|/2\pi.
\end{align*}
\end{lemma}

\begin{proof}
The inequality follows immediately from $\tr{\rho \Fsf(X)}\le\tr{\Fsf(X)}$ and the fact that due to the
covariance condition \eqref{eq:phasedef} the phase distribution is uniform in the number states, i.e., for each $\theta$, 
\begin{align*}
\ip{n}{ \Fsf(X)|n}  = \ip{n}{ \Fsf(X \dotplus\theta)|n},
\end{align*}
therefore $\ip{n}{ \Fsf(X)|n} ={|X|}/{(2\pi)}$, and so $\tr{\Fsf(X)}=d|X|/(2\pi)$.
\end{proof}
We note that all covariant phase {\sc pom}s are absolutely continuous with respect to the Lebesgue measure.  The localisation lemma puts stringent bounds on the magnitude of the localisation probability if $|X|$ is small. Conversely, in order to get high localisation probability (close to 1) for small intervals $X$ in a finite dimensional system, one needs to choose the dimension $d$ to be large. We will use covariant quantities to construct reference frames in the relativisation model, where large (typically infinite-dimensional) Hilbert spaces are required
for the reference system to to be good, in a sense to be discussed.

\section{Symmetry}

\subsection{Observables as Invariant Quantities}

A quantum system $\Sy$ is constrained to behave in accordance with the symmetries of the spacetime
it inhabits and to the concomitant conservation laws that arise. An upshot of such a constraint is
that certain quantities require two systems for their definition or, more colloquially, require a reference frame. Henceforth, \emph{absolute} quantities will be understood
as those whose formal representation does not explicitly rely on such a reference, which is therefore viewed 
as \emph{external} \cite{brs}. Such absolute quantities should not be taken to be observable:\footnote{To be proper, we should write ``absolute" in quotation marks to emphasise that such quantities are not represented in reality, but are mere notational short-hands. We avoid this only for aesthetic reasons.} the absolute position of a system is not meaningful, but both the relative positions of parts of $\Sy$ as a compound system
and the position of $\Sy$ (or, e.g., its centre of mass) relative to some other system $\R$ is. Relative position is a shift-invariant quantity. We proceed with the hypothesis that what can be measured is invariant with respect to the relevant transformation group, with particular emphasis on the group of phase shifts. 

Thus, a (locally compact) symmetry group $G$ acts in the Hilbert space $\his$ of the system $\Sy$
via a (strongly continuous, projective) unitary representation $U$. In non-relativistic quantum
mechanics $G$ (the spacetime symmetry group) is the Galilei group, and the stipulation of symmetry is that any {\sc pom} $\Esf$ of $\Sy$
to be deemed observable must satisfy $U(g)\Esf(X)U(g)^* = \Esf(X)$ for all $g \in G$ and $X\in \mathcal{B}(\Omega)$ ($\Omega$ is any appropriate $G$-space).
In this paper we simplify the problem, treating only unitary representations, and focus on shifts in one dimension ($G=\mathbb{R}$), and rotations ($G=S^1$). The latter case has a spacetime realisation as rotations about an axis, and an ``internal" realisation as shifts in phase of, say, a laser beam.

\subsection{Number and Phase}

Consider a (possibly unbounded) number operator $N=\sum_n n P_n$ acting in $\his$, generating a strongly
continuous unitary representation $U_{\Sy}$ of $S^1$ in $\his$ via the unitary operators $U_{\Sy}(\theta):=e^{iN_{\Sy} \theta}$, giving rise by conjugation to an action on $\lhs$, i.e., $A \mapsto e^{iN_{\Sy} \theta} A e^{-iN_{\Sy} \theta}$, and on states $\rho \mapsto e^{-iN_{\Sy} \theta} \rho e^{iN_{\Sy} \theta}$. 

Consider the mapping
$\tau_{\Sy}:\lhs \to \lhs$ defined by

\begin{equation}
\tau_{\Sy}(A)=\sum P_n A P_n,
\end{equation}
with its predual $\tau_{\Sy}{_*} : \mathcal{T}_1(\his) \to \mathcal{T}_1(\his)$ taking the same form:
\begin{equation}
\tau_{\Sy}{_*}(\rho)=\sum P_n \rho P_n.
\end{equation}
$\tau_{\Sy}{_*}$ is familiar from various contexts. In the quantum theory of measurement,
it is the L\"{u}ders map arising from a non-selective measurement of $N_{\Sy}$; in an optical setting it is called a dephasing channel. It also appears in decoherence theory. $\tau_{\Sy}{_*}$ is trace-preserving, and hence bounded and trace-norm continuous. It can be shown that for any pure state 
$P[\phi]$, $\tau_{\Sy}{_*}(P[\phi])$ is the mixture of $N_{\Sy}$-eigenstates which minimises the Hilbert-Schmidt distance from $P[\phi]$; we omit the proof.

\begin{proposition}\label{prop:invh}
Let $\mu$ denote the (normalised) Haar measure on $S^1$. 
For self-adjoint $A \in \lhs$ the following are equivalent:
\begin{enumerate}
\item $[A,P_n]=0$ for all $n$ (and thus $[A,N_{\Sy}]=0$ for bounded $N_{\Sy}$).
\item  $U(\theta)AU(\theta^*)=A$ for all $\theta$.
\item $\int U(\theta) A U(\theta)^* d\mu (\theta) = A$.
\item $\tau_{\Sy}(A) = A$.
\end{enumerate}
For $\rho \in \mathcal{L}_1(\his)$, the following are also equivalent:
\begin{enumerate}
\item[(5)] $[\rho, P_n]=0$  for all $n$ (and thus $[\rho,N_{\Sy}]=0$ for bounded $N_{\Sy}$).
\item[(6)]  $U(\theta)\rho U(\theta^*)=\rho$ for all $\theta$.
\item[(7)] $\int U(\theta)^* \rho U(\theta) d\mu (\theta) = \rho$.
\item[(8)] $\tau_{\Sy}{_*}(\rho) = \rho$.
\end{enumerate}
\end{proposition}

\begin{proposition}\label{prop:inv2}
The following hold (prime denoting commutant).
\begin{enumerate}
\item For any $A \in \{P_n\}^{\prime}$, $p_{\rho}^{A}(X) = p^A_{\tau_{\Sy}{_*}(\rho)}(X)$.
\item For any $\rho \in \{P_n\}^{\prime}$, $p_{\rho}^{A}(X) = p^{\tau_{\Sy}(A)}_{\rho}(X)$.
\end{enumerate}
\end{proposition}

We omit the proof of (1) which is straightforward, and note that (2) follows from the duality
\begin{equation}\label{eq:spi}
\tr{\rho \tau_{\Sy}(\E^A(X))} = \tr{\tau_{\Sy}{_*}(\rho) \E^A(X)} = \tr{\tau_{\Sy}{_*}(\rho) \tau_{\Sy}(\E^A(X))};
\end{equation}
the final equality is not part of the proof, but shows that demanding states or observables to be invariant is equivalent to demanding the invariance of both.
The proposition holds also if appropriately rephrased for unsharp $\E$ in place of $A$. This shows that no invariant quantity (i.e., no {\em boda fide} observable) of $\Sy$ can distinguish between $\rho$ and its invariant ``counterpart", $\tau_{\Sy}{_*}(\rho)$. Operationally, then, the stipulation of invariance of observables partitions the state space into equivalence classes of indistinguishable states under the obvious equivalence relation. In the dual picture, keeping to invariant states means absolute and invariant quantities cannot be distinguished. Thus, stipulating that either
states or observables must be invariant constitutes a restriction to ordinary quantum theory,
and only if both states {\em and} observables are unrestricted do we have the usual textbook description.

\subsection{Position and Momentum}

The shift group on $\mathbb{R}$ is unitarily implemented in $L^2(\mathbb{R})$ by the operators
$U(x)=e^{ixP}$, with $P$ the momentum operator in one space dimension. The spectral measure $\E^Q$ of position $Q$
is singled out (among spectral measures) by the condition $\E^Q(X-x) = U(x)\E^Q(X)U(x)^*$. Unsharp positions also satisfy
such a covariance criterion, and thus, as non-invariant quantities, absolute positions (sharp or unsharp) do not represent observable quantities, reflecting the lack of absolute space. 

However, it may be possible to distinguish separate parts of a given quantum system, $\Sy$ and $\R$, and it may be possible to speak of the position of $\Sy$ relative to $\R$, and therefore to measure
the shift-invariant quantity $Q_{\Sy}\otimes \id - \id \otimes Q_{\R}$ or, indeed, any other shift-invariant quantity of $\Sy + \R$. Similar considerations apply to boosts.

{\remark{\rm
The non-compactness of (the shift group on) $\mathbb{R}$ rules out the existence of a normalisable Haar
measure playing the role of $\mu$ in Proposition \ref{prop:invh} (as also pointed out recently by Smith et al. \cite{smith}).}}

Therefore, the absolute position $Q_{\Sy}$
should be understood as representing the relative position $Q_{\Sy} - Q_{\R}$, in the situation that the $Q_{\R}$ system may be suppressed, or ``externalised" \cite{brs} from the description.
We now turn to a general analysis of the possibility of such an externalisation for arbitrary groups and relative quantities, before turning once more to typical examples.


\section{Relativisation}

In this section we introduce a \emph{relativisation} mapping $\Y$, and prove various mathematical properties satisfied by it. We discuss the physical interpretation of $\Y$ as the making 
explicit of a reference system, and in the following section show that under high localisation
of the reference system with respect to an appropriate covariant quantity used to define $\Y$,
 the description of the system alone in terms of absolute, non-invariant quantities
can provide a statistically good account of the relative quantities. Conversely, 
it is shown that in the case of reference system delocalisation, the description afforded by system quantities is necessarily invariant, giving generally poor representation of relative observables.
The $\Y$ map generalises (by considering a {\sc pom} for the reference and more general groups) and makes mathematically precise (by avoiding improper states, and giving rigorous meaning to the integral) the ``$\$$" map of \cite{brs}. We also introduce the predual, $\Y_*$, which de-relativises states, replacing the
erroneous use of $\$$ also on states in \cite{brs}.

\subsection{Definition and Properties of the Map $\Y$}

\begin{definition}
Let $\hit=\his\otimes\hir$ with $\his,\hir$ finite or infinite dimensional Hilbert spaces hosting strongly continuous unitary representations $U_{\Sy}$ and $U_{\R}$ respectively of a locally compact metrisable group $G$ and let $\Fsf$ be a covariant \po acting in $\hir$. 
Then $\Y:\mathcal{L}(\his) \to \mathcal{L}(\hit)$ is defined by 
\begin{equation}\label{eq:ydef}
\Y(A)= \int_G U_{\Sy}(g)AU_{\Sy}(g)^* \otimes \Fsf (dg).
\end{equation}
\end{definition}
\noindent $\Y$ will be called a {\em relativisation map}, various properties of which will be given in
proposition \ref{prop2}. $\Y$ also acts on {\sc pom}s by $(\Y \circ \Esf)(X):=\Y(\Esf(X))$.

We must first give the definition of this integral. If $\hir$ is finite-dimensional and $G$ is compact and metrisable, 
there exists a unique positive $T$ such that $\Fsf(X) = \kappa \int_X U_{\R}(g)TU_{\R}(g)^* dg$ where $\kappa$ is chosen so that $\int_G U_{\R}(g)TU_{\R}(g)^* dg = \id$ \cite{davies}. Then \cite{mlb}, the integral
\eqref{eq:ydef} may be defined by 

\begin{equation}\label{eq:ydef2}
\Y(A)= \kappa \int_G U_{\Sy}(g)AU_{\Sy}(g)^* \otimes U_{\R}(g)TU_{\R}(g)^* dg.
\end{equation}

For the case that $\his$ and $\hir$ are of infinite dimension, more work is required (see also \cite{mlak}, section 5.4). 
Let $G$ be a locally compact second countable group. 
We first construct $\Y$ for a subset $\mathcal{A} \subset \mathcal{L}(\his)$
on which the action $\alpha_g$ is norm continuous, noting that this subset is 
weakly dense; see the discussion below.

For $e \in G$, for any $\epsilon >0$ there exists a neighbourhood $U$ such that 
$\Vert \alpha_g(A) - A\Vert \leq \epsilon \Vert A\Vert$. 
By translating this $U$ we obtain a covering $G= \cup_{g \in G} g U$. 
The Lindel\"{o}f property of $G$ allows us to obtain a countable cover 
$G= \cup_i U_i$ 
out of it. By taking their intersections we obtain a disjoint 
countable cover (a mesh) $G= \cup_n V_n$ 
such that for any $g, g' \in V_n$ it holds that
$\Vert\alpha_g(A) - \alpha_{g'}(A) \Vert 
\leq \epsilon \Vert A\Vert$. 
Let $\{\epsilon_{N}\}$ be a decreasing sequence converging to $0$. 
By employing the above construction we can construct a mesh $G=
\cup_n V^{N}_n$ for each 
$N$ so that $V^N_n = \cup_{m \in H^N_M} V^{M}_m$ holds for 
each $N\leq M$ with some proper $H^N_M \subset \mathbf{N}$. 
(That is, a mesh of $M$ is strictly finer than that of $N$.) 
We choose $g^N_n \in V^N_n$ for each $n$ (and $N$). 
Now we assume a covariant {\sc pom} $\E(\cdot)$ on the reference side to be 
 projection-valued. This suffices since any {\sc pom} $\F$ can be dilated to a {\sc pvm} by 
Naimark extension.  

We introduce for each $N$, the mapping 
\begin{eqnarray*}
\Y_N(A) := \sum_{n} \alpha_{g^N_n}(A) \otimes \E(V^N_n).
\end{eqnarray*}
It is easy to see that this is bounded. 
In fact, for an arbitrary normalised vector $|\psi\rangle \in \his \otimes \hir$, 
it holds that
\begin{eqnarray*}
\langle \psi | \Y_N(A)^* \Y_N(A)|\psi \rangle 
= \sum_n \mbox{tr}[\rho^N_n \alpha_{g^N_n}(A^*A) ]
p^N_n \leq \sum_n p^N_n \Vert A\Vert^2 = \Vert A\Vert^2, 
\end{eqnarray*}
where $p^N_n := \langle \psi | \id \otimes \E(V^N_n) |\psi\rangle$ 
and $\rho^N_n$ is a density operator uniquely determined by 
$\mbox{tr}[\rho^N_n X] = \langle \psi |X\otimes \E(V^N_n) |\psi \rangle$. 
\\
Now we show that the sequence $\{\Y_N(A)\}$ is a Cauchy sequence. 
As an arbitrary $A$ can be decomposed into two self-adjoint operators, 
it suffices to show the property for self-adjoint $A$. 
For $\epsilon>0$, we show that $\Vert \Y_N(A) - \Y_M(A) \Vert
\leq \epsilon \Vert A\Vert$ for $\epsilon_N, \epsilon_M \leq \epsilon$. 
Let $M>N$. For an arbitrary normalised $|\psi\rangle \in \his\otimes \hir$, 
we have 
\begin{eqnarray*}
|\langle \psi | \Y_M (A) - \Y_N(A) |\psi\rangle |
&= &\left|
\sum_n p^N_n \mbox{tr}[\rho^N_n \alpha_{g^N_n}(A)]
- \sum_m p^M_m \mbox{tr}[\rho^M_m \alpha_{g^M_m}(A)]
\right|
\\
&\leq &
\sum_n \left| 
p^N_n \mbox{tr}[\rho^N_n \alpha_{g^N_n}(A)]
- \sum_{m \in H^N_M} \mbox{tr}[\rho^M_m \alpha_{g^M_m}(A)]\right|
\\
&\leq &
\sum_n \sum_{m\in H^N_M}
p^M_m
\left|
\mbox{tr}[\rho^M_m (\alpha_{g^N_n}(A) - \alpha_{g^M_m}(A)]
\right|
\leq \epsilon \Vert A\Vert,
\end{eqnarray*}
where we used $\sum_{m\in H^N_M} p^M_m \rho^M_m = p^N_n \rho^N_n$.
Thus for $A=A^*$, we find $\Vert \Y_M(A) - \Y_N(A)\Vert \leq 
\epsilon \Vert A\Vert$.  
Thus we can define $\Y(A)$ by $\Y(A):=\lim_N \Y_N(A)$.  
\\
Note that this definition does not depend on 
the choice of covers. In fact one can see that for 
two covers $\{V^N_n\}$ and $\{\hat{V}^N_n\}$ for $\epsilon_N$, 
their intersections $\{V^N_n \cap V^N_m\}$ is also a cover. 
It is easy to see that the difference between the $\Y_N(A)$ constructed
with $\{V^N_n\}$ ($\{\hat{V}^N_n\}$) and $\{V^N_n \cap \hat{V}^N_m\}$ 
is smaller than $\epsilon \Vert A\Vert$. 

Now we discuss the density of $\mathcal{A}$ in $\mathcal{L}(\his)$.  
We assume that the action is weakly continuous---a natural assumption from a physical point of view. 
In addition the action is assumed to be implemented by 
a unitary operator $U(g)$. 
Then one can see that the representation 
$U(g)$ is strongly continuous. We can define, for a smooth function $f$ 
whose support is compact in $G$, $A(f) := \int \mu(dg) U(g)AU(g)^* 
f(g)$. For such $A(f)$ the action $\alpha_g$ is norm continuous. 
We may introduce $\mathcal{A}$ as a subalgebra generated 
by such elements. 
If we take $f$ to be localised around the unit of $G$, 
$A(f)$ gets close to $A$ with respect to the weak topology. 
Thus $\mathcal{A}$ is dense in $\mathcal{L}(\his)$. 

As a final remark, we show that this $\Y(A)$ can be defined on the 
whole $\mathcal{L}(\his)$ for Abelian $G$ implemented by a true 
unitary representation. 
As each $U(g)$ commutes with each $U(g^{\prime})$, their generators can be 
diagonalised simultaneously. For simplicity we treat here only $G=\mathbb{R}$ 
and write its generator as $K = \int_{\mathbb{R}} k \P(dk)$. 
Now we introduce $P_E= \int_{|k|\leq E} \P(dk)$. Then 
one can see that for any $A \in \mathcal{L}(\his)$, 
$P_E A P_E$ is a ``smooth" element (i.e., 
$\alpha_g(P_E AP_E)$ is norm continuous)
and 
$\Y(P_E A P_E)$ is defined. 
In addition its norm is bounded as 
$\Vert \Y(P_E A P_E) \Vert \leq \Vert A \Vert$. 
Now for arbitrary vectors $|\psi\rangle, |\phi \rangle 
\in \his \otimes \hir$, we 
define a sesquilinear form $Q(|\phi\rangle, |\psi\rangle )$  
by $\lim_{E\to \infty} 
\langle \phi | \Y(P_E A P_E) |\psi\rangle$. 
We first confirm that this is well-defined. 
For any $\epsilon>0$, there exists $E_0$ such that 
$\Vert (\id - P_{E_0})|\psi\rangle \Vert, \Vert (\id - P_{E_0})|\phi \rangle \Vert 
\leq \epsilon$. 
Then for any $E\geq E' \geq E_0$, 
we have
\begin{eqnarray*}
\left|\langle \psi | \Y(P_EA P_E) |\phi \rangle 
- \langle \psi | \Y(P_{E'} A P_{E'}) | \phi \rangle \right|
&=& \left| \langle \psi | \Y(P_E A P_E) |\phi \rangle 
- \langle \psi | P_{E'} \Y (P_E A P_E ) P_{E'} | \phi \rangle \right|
\\
&\leq& (2\epsilon + \epsilon^2 )\Vert A \Vert. 
\end{eqnarray*}
Thus this sequence is Cauchy. 
On the other hand, each quantity is 
bounded by $\Vert A\Vert \Vert |\psi \rangle \Vert 
\Vert |\phi \rangle \Vert$. 
Thus it converges.  
%
It is also easy to see that 
this sesquilinear form is bounded as 
$Q(|\psi \rangle, |\phi \rangle ) \leq \Vert A\Vert 
\Vert |\psi \rangle \Vert \Vert |\phi \rangle \Vert$. 
Thus there exists an operator $\Y(A)$ satisfying 
$\langle \psi |\Y(A) |\phi \rangle = Q(|\psi \rangle, 
|\phi \rangle)$. Moreover, it is easy to see that 
such defined $\Y(A)$ is bounded as $\Vert \Y(A) \Vert 
\leq \Vert A\Vert$.

\if{
\blue{For reference: 
$\Y(A)$ for $A \in \mathcal{L}(\his)$
is defined in the case of $G=S^1$ as follows. 
(For simplicity we assume that $N_{\R}$ is nondegenerate and $U_{\Sy}(\theta)$ 
implements weakly continuous action of $S^1$.)
We can write $|\psi_1\rangle , |\psi_2\rangle \in \his \otimes \hir$ as
$|\psi_1\rangle = \sum_n |\xi^n_1\rangle \otimes |n\rangle$
and $|\psi_2 \rangle = \sum_m  
|\xi^m_2 \rangle \otimes |m\rangle$. 
They satisfy
\begin{eqnarray*}
\Vert |\psi_1\rangle \Vert = \sqrt{\sum_n \Vert \xi^n_1\Vert}
\\
\Vert |\psi_2\rangle\Vert = \sqrt{\sum_m \Vert \xi^m_2\Vert}.
\end{eqnarray*}
Then a sesquilnear form $\langle \psi_1 | \Y(A) |\psi_2\rangle$
is formally introduced as, 
\begin{eqnarray*}
\langle \psi_1 | \Y(A) |\psi_2\rangle 
:= \sum_{nm} \int \langle \xi^n_1  |U_{\Sy}(\theta) AU_{\Sy}(\theta)^*|\xi^m_2 \rangle 
\langle n | \F(d\theta) |m\rangle.
\end{eqnarray*} 
The integral is defined, for instance, by usual Riemann sum. 
In fact, complex measure $\langle n| \F(\Delta)|m\rangle$ is 
bounded as, 
\begin{eqnarray*}
|\langle n |\F(\Delta) |m\rangle|
\leq \langle n | \F(\Delta)|n \rangle^{1/2}
\langle m | \F(\Delta)|m \rangle^{1/2}
= \frac{|\Delta|}{2\pi}, 
\end{eqnarray*}
and $ \langle \xi^n_1  |U_{\Sy}(\theta) AU_{\Sy}(\theta)^*|\xi^m_2 \rangle $
is continuous. 
In addition, each integral is bounded as
\begin{eqnarray*}
\left| \int \langle \xi^n_1  |U_{\Sy}(\theta) AU_{\Sy}(\theta)^*|\xi^m_2 \rangle 
\langle n | \F(d\theta) |m\rangle \right|
\leq \Vert A\Vert \Vert \xi^n_1\Vert \Vert \xi^m_2\Vert. 
\end{eqnarray*}
}
}
\fi


\if{
\blue{For reference: 
$\Y(A)$ for $A \in \mathcal{L}(\his)$
is defined in the case of $G=S^1$ as follows. 
(For simplicity we assume that $N_{\R}$ is nondegenerate and $U_{\Sy}(\theta)$ 
implements weakly continuous action of $S^1$.)
We can write $|\psi_1\rangle , |\psi_2\rangle \in \his \otimes \hir$ as
$|\psi_1\rangle = \sum_n |\xi^n_1\rangle \otimes |n\rangle$
and $|\psi_2 \rangle = \sum_m  
|\xi^m_2 \rangle \otimes |m\rangle$. 
They satisfy
\begin{eqnarray*}
\Vert |\psi_1\rangle \Vert = \sqrt{\sum_n \Vert \xi^n_1\Vert}
\\
\Vert |\psi_2\rangle\Vert = \sqrt{\sum_m \Vert \xi^m_2\Vert}.
\end{eqnarray*}
Then a sesquilnear form $\langle \psi_1 | \Y(A) |\psi_2\rangle$
is formally introduced as, 
\begin{eqnarray*}
\langle \psi_1 | \Y(A) |\psi_2\rangle 
:= \sum_{nm} \int \langle \xi^n_1  |U_{\Sy}(\theta) AU_{\Sy}(\theta)^*|\xi^m_2 \rangle 
\langle n | \F(d\theta) |m\rangle.
\end{eqnarray*} 
The integral is defined, for instance, by usual Riemann sum. 
In fact, complex measure $\langle n| \F(\Delta)|m\rangle$ is 
bounded as, 
\begin{eqnarray*}
|\langle n |\F(\Delta) |m\rangle|
\leq \langle n | \F(\Delta)|n \rangle^{1/2}
\langle m | \F(\Delta)|m \rangle^{1/2}
= \frac{|\Delta|}{2\pi}, 
\end{eqnarray*}
and $ \langle \xi^n_1  |U_{\Sy}(\theta) AU_{\Sy}(\theta)^*|\xi^m_2 \rangle $
is continuous. 
In addition, each integral is bounded as
\begin{eqnarray*}
\left| \int \langle \xi^n_1  |U_{\Sy}(\theta) AU_{\Sy}(\theta)^*|\xi^m_2 \rangle 
\langle n | \F(d\theta) |m\rangle \right|
\leq \Vert A\Vert \Vert \xi^n_1\Vert \Vert \xi^m_2\Vert. 
\end{eqnarray*}
}
}
\fi

\if
{\blue{It seems that we have three options for rigorously constructing $\Y$ in the infinite dimensional case.
\begin{enumerate}
\item Base it on a construction of Holevo \cite[Prop.~1]{holevochoi}, assuming $G$ is metrisable and complete (both reasonable!). For any finite measure $\mu$ on $G$ and bounded $A$ the map (Banach space valued function)
$g \mapsto U(g)AU(g)^*$ is Bochner integrable since $\no{U(g)AU(g)^*}\leq \no{A}$, and so 
$\int_G \no{U(g)AU(g)^*} d \mu < \infty$ which, according to wikipedia, is necessary and sufficient for Bochner integrability (I need to find a proper reference). \blue{Then, it seems to me, that we can carry out Holevo's construction (which I'm afraid I don't yet fully understand) replacing all occurrences of trace norm with operator norm ($\no{A} \leq \no{A}_{\rm tr}$). The only problem seems to be that he uses separability of the convex set of density matrices with respect to the trace norm. For an infinite Hilbert space, I don't think we have separability for the uniform topology on $B(\mathcal{H})$ ($B(H)$ is separable with respect to SOT and WOT), so I am still unsure how to proceed.}
\item Use duality. The map $G \to \mathcal{S}(\his)$; $g \mapsto U^*(g)\rho U(g)$ is Bochner integrable and the integral $\int_G U^*(g)\rho U(g) \tr{\sigma \Fsf(dg)}$ exists, I think. It is a Bochner integral. But this is nothing but the predual of $\Y$, as we have it, so perhaps we can define $\Y$ as the dual of this map?
\item Define with respect to a {\sc pvm}, and project.
\end{enumerate}
}}
\fi


\begin{proposition}\label{prop2}
$\Y : \lhs \to \mathcal{L}(\hit)$ has the following properties.
\begin{enumerate}
\item \label{it:lua} $\Y$ is linear, unital \emph{(}$\Y(\id _{\hs}) = \id _{\mathcal{H}_T}$\emph{)}, and preserves adjoints \emph{(}$ \Y(A^*) = \Y (A)^*$ for any $A \in \lhs$\emph{)}.
\item $\Y$ \label{it:phb} is positive and hence bounded. 
\item $\Y$ is completely positive.
\item $\Y$ is the dual  of a bounded linear map $\Y_*: \mathcal{L}_1 (\hit) \to \mathcal{L}_1 (\his)$ defined
by ${\rm tr} \left[ R \Y (A) \right]={\rm tr} \left[\Y_* (R) A \right]$ for all 
$A \in \mathcal{L}(\his),~R\in \mathcal{L}_1 (\hit)$. In particular, $\Y$ is normal.
\item \label{item:5} If $E$ is an effect, so is $\Y(E)$. With $\Esf: \mathcal{B}(G) \to \lhs$, $\Esf \mapsto \Y \circ \Esf\equiv \Esf^{(\Y)}$ defines a map from $\mathcal{L(H_S)}$-valued {\sc POM}s to 
 $\mathcal{L}(\hit)$-valued POMs.
\item If $\Fsf$ is projection-valued, $\Y$ is multiplicative, i.e., $\Y (AB) = \Y (A) \Y (B)$ for all $A,B \in \lhs$,---and is thus an  algebraic  $^*$-homomorphism.
\item With $\mathcal{U}(g) = U_{\Sy}(g)\otimes U_{\R}(g)$,
\begin{equation}
\mathcal{U}(g)\Y (A) \mathcal{U}(g)^* = \Y (A)
\quad\text{for all }A \in \mathcal{L}(\mathcal{H}_S),\quad g \in G.
\end{equation}
 If $U_{\Sy}(g)AU_{\Sy}(g)^*=A$ for all $g \in G$, then
$\Y(A) = A \otimes \id$.  
\end{enumerate}
\end{proposition}

\begin{proof}
\begin{enumerate}
\item  These properties follow immediately from the definition.
\item \label{it:pos} To prove positivity, consider $\ip{\psi}{\Y (A) \psi}$ for $\psi \in \hit$ and assume that $A$ is positive, therefore $A=B^2$ for a (unique) $B \geq 0$. Let $\{\varphi _i \otimes \phi _j \}$ be an orthonormal
basis in $\mathcal{H}_{\Sy} \otimes \mathcal{H}_{\mathcal{R}}$, and $\psi = \sum_{i,j}c_{ij}\varphi _i \otimes \phi _j$. Let $\gamma_{g}(A) \equiv U_{\Sy}(g)AU_{\Sy}(g)^*$; we then have 
\begin{equation}\label{yenpos}
\ip{\psi}{\Y (A) \psi} = \sum_{i,j,k,l}\bar{c}_{ij} c_{kl} \int_G \ip{\varphi _i}{\gamma_{g}(A) \varphi _k}\ip{\phi_j}{ \Fsf (dg) \phi_l}.
\end{equation}
Writing $\gamma_{g} (A) = U_{\Sy}(g)B^2  U_{\Sy}(g)^*$ and $B^2 = \sum_m B \ket{\varphi_m} \bra{\varphi_m}B$, the expression for $\ip{\psi}{\Y (A) \psi}$ becomes
\begin{equation}\label{yenpos2}
\ip{\psi}{\Y (A) \psi} = \sum_m \int_G \ip{\xi _{m}(g)}{ \Fsf (dg) \xi_m (g)},
\end{equation}
where we have defined $\xi _{m}(g) := \sum_{k,l}c_{kl}\ip{\varphi_m}{BU_{\Sy}(g)^* \varphi _k} \phi _l$. The right hand side of the expression \eqref{yenpos2} is manifestly positive. Any positive (linear) map between $C^*$-algebras with unit is automatically bounded---see \cite{con1}, Prop. 33.4. Now we note that the right hand side of \eqref{yenpos2} can be written
\begin{equation}
\tr { \int_G \sum_m \ket {\xi _m(g)}\bra{\xi _{m}(g)} \Fsf  (dg)}.
\end{equation}


\item In order to show that $\Y \otimes \id_n : \lhs \otimes M_n(\mathbb{C}) \to \lhs \otimes \lhr \otimes M_n(\mathbb{C})$ is positive we introduce an orthonormal basis $\{ \eta _k\} \subset \mathbb{C}^n$. Then the proof runs along essentially the same lines as (\ref{it:pos}) above, and one finds
 that for $\Psi = \sum_{i,j,k}c_{ijk}\varphi_i \otimes \phi_j \otimes \eta_k$, and repeating
 the argument, letting $A=B^2 = \sum_pB \ket{\varphi _p}\bra{\varphi_p}B$, 
 $\ip{\Psi}{\Y(A) \otimes \id \Psi}$ is given as
 \begin{equation}
 \ip{\Psi}{\Y(A) \otimes \id \Psi} = \sum_{i,j,k,l,m,n}\int_G \ip{\varphi_i}{U(g)B^2U(g)^* \varphi_l}\ip{\phi_j}{\Fsf (dg)\phi_m}\ip{\eta_k}{\eta_n}
 \end{equation}
 which may be written as 
 \begin{equation}
 \ip{\Psi}{\Y(A) \otimes \id \Psi} = \sum_p \int_G \ip{\zeta_p(g)}{\Fsf  (dg) \zeta_p(g)},
 \end{equation}
 where $\zeta_g(p) := \sum_{l,m,n} \ip{\varphi_p}{BU(g)^*\varphi_l} \phi_m \otimes \eta_n$.
 Thus, by the same argument as in (\ref{it:pos}), $\Y$ is completely positive.
 \item The normality of $\Y$ follows from $\Y$ being the dual of the positive linear map $\Y_* :\mathcal{S}(\hit) \to \mathcal{S}(\his)$ \cite[Lemma 2.2.]{davies}.
\item The effect property $0 \leq\Y(E) \leq 1$ follows immediately from \ref{it:lua} and \ref{it:phb}. That $\Y$ has the $\sigma$-additivity property of a {\sc pom} follows from the normality of $\Y$. 
\item  Let $\Fsf$ be projection valued. For $A,B \in \lhs$, let $f_{\varphi, \varphi ^{\prime}}(g, g^{\prime})$  denote the bounded complex function
$\ip{\varphi}{\gamma_g(A)\gamma_{g^{\prime}}(B)\varphi ^{\prime}}$: 
\begin{equation}\label{eq:ypm}
\ip{\varphi \otimes \phi}{\Y(A)\Y(B)\varphi ^{\prime} \otimes \phi^{\prime}}=\int \int f_{\varphi, \varphi ^{\prime}}(g,g^{\prime}) \ip{\phi}{\Fsf(d g) \Fsf (d g ^{\prime}) \phi ^{\prime}}.
\end{equation}
The product measure defined by $X\times Y \mapsto \F(X)\F(Y)$ is zero whenever $X \cap Y = \emptyset$, and hence the right hand side of \eqref{eq:ypm} reduces to
\begin{equation}
\int \ip{\varphi}{U_{\Sy}(g) (AB) U_{\Sy}(g)^* \varphi ^{\prime}}\ip{\phi}{\Fsf (d g) \phi^{\prime}},
\end{equation}
which is the expression for $\ip{\varphi \otimes \phi}{\Y(AB) \varphi^{\prime}\otimes \phi ^{\prime}}$.
\item 
We compute:
\begin{align*}\label{eqn:yensym}
\mathcal{U}(g)\Y (A)\mathcal{U}(g)^* &= \int_G {U}_{\Sy}(gg^{\prime})A {U}_{\Sy}(gg^{\prime})^* \Fsf (d(gg^{\prime})) = \Y(A)
\end{align*} 
Therefore to each bounded self-adjoint operator of the system the map $\Y$ assigns a bounded self-adjoint 
operator (see below) $\Y (A)$ acting in $\hit$ which is invariant under 
the action of $\mathcal{U}$.
\end{enumerate} 
\end{proof}

{\remark\rm  From (\ref{item:5}) we observe that $\Y$ not only relativises self-adjoint operators, but also their spectral measures, and more generally any {\sc pom}. An example of the latter, which we will encounter in subsection \ref{subsubsec:rpo}, is the relativisation of a covariant phase {\sc pom}, resulting in a relative phase observable.}

\subsection{Examples}
In this subsection we give examples of familiar relative quantities obtained under $\Y$ to demonstrate that $\Y$ functions as expected. Its main utility, however, lies in the fact that 
it relativises arbitrary quantities.

\subsubsection{Position and Momentum}
Consider the spectral measure $\E^{Q_{\Sy}}$ of the position $Q_{\Sy}$, $\E^{Q_{\R}}$ of $Q_{\R}$, and unitary shifts $U_{\Sy}(x)=e^{ixP_{\Sy}}$ and $U_{\R}(x)=e^{ixP_{\R}}$. Then,
\begin{equation}
(\Y \circ \E^{Q_{\Sy}})(X) = \int_{\mathbb{R}}e^{ixP_{\Sy}}\E^{Q_{\Sy}}(X)e^{-ixP_{\Sy}}\otimes \E^{Q_{\R}}(dx),
\end{equation}
which may be written as 
\begin{equation}
(\Y \circ \E^{Q_{\Sy}})(X) = \int_{\mathbb{R}} \int_{\mathbb{R}} \chi_X(x^{\prime}-x)\E^{Q_{\Sy}}(dx^{\prime}) \otimes \E^{Q_{\R}}(dx).
\end{equation}
This is easily recognised as the spectral measure of the relative position 
\begin{equation}
Q_{\Sy} - Q_{\R} = \int_{\mathbb{R}} \int_{\mathbb{R}} (x-x^{\prime})\E^{Q_{\Sy}}(dx^{\prime}) \otimes \E^{Q_{\R}}(dx).
\end{equation}
Therefore, one may formally write $\Y(Q_{\Sy}) = Q_{\Sy}-Q_{\R}$.

Under the given relativisation, the spectral measure $\E^{P_{\Sy}}$ of the momentum $P_{\Sy}$ takes the simple form $(\Y\circ \E^{P_{\Sy}})(X) = \E^{P_{\Sy}}(X)\otimes \id$, and again we write $\Y(P_{\Sy}) = P_{\Sy} \otimes \id$. 

Relativisation of momentum under boosts follows an identical argument; we find that (using the same symbol $\Y$ for relativising with respect to boosts)

\begin{equation}
(\Y \circ \E^{P_{\Sy}})(Y) = \int_{\mathbb{R}} e^{iyQ_{\Sy}}\E^{P_{\Sy}}(X)^{-iyQ_{\Sy}}\otimes \\E^{P_{\R}}(dy),
\end{equation}
yielding 
\begin{equation}
(\Y \circ \E^{P_{\Sy}})(Y)=\int_{\mathbb{R}} \int_{\mathbb{R}} \chi_Y(y^{\prime}-y)\E^{P_{\Sy}}(dy^{\prime}) \otimes \E^{P_{\R}}(dy),
\end{equation}
which is the spectral measure of $P_{\Sy}-P_{\R}$. Under this relativisation, $Q \mapsto Q \otimes \id$. 
\begin{remark}\rm
With respect to the boost part of the Galilei group, the momentum relativisation assumed $\Sy$ and $\R$ are of equal mass. For systems with different mass, i.e., $m_{\Sy}$ and $m_{\R}$, $\Y$ must be appropriately redefined. 
\end{remark}

We note also the possibility of unsharp relativisations, that is, allowing for one or both of the spectral measures $\E^{Q_{\Sy}}$ and $\E^{Q_{\R}}$ to be replaced by unsharp or smeared (covariant) positions. The same applies for unsharp momenta.

As an example, which we revisit in section \ref{sec:lad}, we may take a smeared position $(\chi_X *e)(Q_{\R}) \equiv \E^{Q_\R}_{e}(X)$ for the reference, yielding

\begin{equation}
(\Y \circ \E)(X) = \int_{\mathbb{R}} e^{iPx}\Esf(X)e^{-iPx} \otimes \E^{Q_\R}_{e}(dx).
\end{equation}


\if
Consider now the map $\Y:\lhs \to \lh$ defined by

\begin{equation}
\Y(A) = \frac{1}{2 \pi \hbar}\int W_{\Sy}(q,p) A W_{\Sy}(q,p)^* \otimes W_{\R}(q,p) T W_{\R}(q,p)^* dq dp,
\end{equation}
which now represents a quantity {\it invariant} with respect to phase space translations. 
The uncertainty relation for joint localisability implies that covariant phase space relativisations are necessarily unsharp.

\fi

\subsubsection{Angle}

Here we consider two systems with $\Theta_{\Sy}, \Theta_{\R}$ being their 
angle operators conjugate to the $z$-components of their angular momenta. For the covariant 
{\sc pvm} of $\R$ we thus choose $\Fsf=\Esf^{\Theta_{\R}}$. Then we obtain:
\begin{align}
& \Y(\Theta_{\Sy}) = \int _0 ^{2 \pi} U(\theta') \Theta_{\Sy} U(\theta ^{\prime})^* \otimes \Esf^{\Theta_{\R}}d (\theta^{\prime})\\
& = \int _0 ^{2 \pi}U(\theta^{\prime}) \left[\int _{0}^{2 \pi} \theta  \Esf^{\Theta_{\Sy}}(d \theta)\right] U(\theta^{\prime})^{*} \otimes \Esf^{\Theta_{\R}}(d \theta ^{\prime}).
\end{align}
Exploiting the covariance of the spectral measure $\Esf^{\Theta_{\Sy}}$ and performing the substitution $\theta+\theta^{\prime} \equiv \theta^{{\prime} {\prime}}$ we find the above equal to
\begin{align}
&\int _{0}^{2 \pi} \left[\int_{0} ^{2 \pi} (\theta ^{\prime \prime} - \theta^{\prime})  \Esf^{\Theta_{\Sy}} (d{\theta ^{\prime \prime}}) \right] \otimes \Esf^{\Theta_{\R}}(d{\theta ^{\prime}})\\
& = \int _{0}^{2 \pi} \left[\Theta_{\Sy} - \theta^{\prime} \id \right] \otimes \Esf^{\Theta_{\R}}(d{\theta^{\prime}})\\
& = \Theta_{\Sy} \otimes \id - \id \otimes \Theta_{\R}.
\end{align}
Therefore, $\Y(\Theta_{\Sy}) =  \Theta_{\Sy} - \Theta_{\R}$.

\subsubsection{Phase}\label{subsubsec:rpo}

We may use $\Y$ to construct a relative phase observable as given in \cite{jpthesis, hlp1}. Let $\Fsf$ be a covariant phase {\sc pom} as defined in \eqref{phase}, and denote by $\Fsf^{\R}$ a covariant phase for $\R$.
Then $ \Y \circ\Fsf$ is given by:
\begin{equation}\label{eq:relphs}
\Y \bigl[\Fsf(X)\bigr] =\int _0 ^{2 \pi} \Fsf (X \dotplus \theta) \otimes \Fsf^{\R} (d\theta), 
\end{equation}
and
\begin{equation}\Y [\Fsf(X)] = \frac{1}{(2 \pi)^2} \sum_{n,m,k,l} \widetilde{c}_{n,m,k,l} \int _{0}^{2\pi}
 d \theta \int _{X \dotplus \theta} e^{i(n-m) \theta ^{\prime}} \ket{n} \bra{m} \otimes  \ket{k} \bra{l} e^{i(k-l) \theta } d \theta ^{\prime},
\end{equation}
where $ \widetilde{c}_{n,m,k,l} \equiv c_{n,m} c^{\prime}_{k,l}$.
Writing $ \ket{n,k} \equiv \ket{n} \otimes \ket{k}$, we have
\begin{equation}
\Y \bigl[\Fsf(X)\bigr] = \frac{1}{2 \pi} \sum_{n,m,k,l} \widetilde{c}_{n,m,k,l}\delta_{n-m,k - l}  \int _X  \ket{n,k} \bra{m,l} e^{i(n-m) \theta } d \theta  ,
\end{equation}
which is a relative phase observable. 

\section{Restriction}

\subsection{Basic Properties}
Consider now a fixed state $\omega$ of $\R$ and the isometric embedding $\mathcal{V}_{\omega} :\mathcal{L}_1(\his) \to \mathcal{L}_1(\hit)$ defined by $\rho \mapsto \rho \otimes \omega$. This has a dual (\emph{restriction}) map $\Gamma_{\omega}:\lht \to \lhs$, which on tensor product operators $A \otimes B$ takes the form 
\begin{equation*}
\Gamma_{\omega}(A \otimes B) = A \tr{\omega B}.
\end{equation*} 
\begin{proposition}
 $\Gamma_{\omega}$ possesses the following properties.
\begin{enumerate}
\item $\Gamma_{\omega}$ is linear, unital, adjoint-preserving.
\item $\Gamma_{\omega}$ is completely positive (and therefore positive).
\item $\Gamma_{\omega}$ is normal.
\item $\Gamma_{\omega}$ is a (normal) conditional expectation in the sense of von Neumann algebras.
\end{enumerate}
\end{proposition}
We recall (e.g., \cite{davies}) that given $\lh$ and a von Neumann subalgebra $\mathcal{W} \subset \lh$ a normal conditional expectation $\mathcal{E}:\lh \to \mathcal{W}$ is a positive, adjoint-preserving normal map satisfying the additional properties i) $\mathcal{E}(X)=X$ if and only if 
$X \in \mathcal{W}$ and ii) $\mathcal{E}(X_1YX_2)=X_1\mathcal{E}(Y)X_2$ for any $X_1,~ X_2 \in \mathcal{W}$ and $Y \in \lh$.
\begin{proof} (1) - (3), see \cite[Ch. 9]{davies}; for the final part we view $\lhs$ as a subalgebra of $\lht$ by identifying $A \in \lhs$ with $A \otimes \id \in \lht$. Then, 
\begin{equation}
\Gamma_{\omega}\left((X_1\otimes \id)(A \otimes B)(X_2 \otimes \id) \right) 
= X_1AX_2 \omega (B) = X_1 \Gamma_{\omega} (A \otimes B)X_2 \otimes \id.
\end{equation}
We extend by linearity to finite sums $\sum_{i,j}A_i \otimes B_j \in \lht$, and to infinite sums by continuity.
\end{proof}

Thus we will sometimes refer to $\Gamma_{\omega}$ as a restriction {\em channel}. $\Gamma_{\omega}$ restricts {\sc pom}s of $\Sy + \R$ to those of $\Sy$,
and is used to translate back from the relative picture to the absolute one, contingent upon the state $\omega$ of $\R$. For a pure product state, for example, for a given self-adjoint $R \in \lht$ and fixed unit $\phi \in \hir$ the expression $\ip{\cdot \otimes \phi}{R \cdot \otimes \phi}$
determines a bounded, real valued quadratic form $\his \times \his \to \mathbb{C}$ and therefore a 
unique bounded self-adjoint operator $R_{\phi} =\Gamma_{\phi}(R) \in \lhs$.

The restriction map is related to the trace as follows. By the duality $\lhs \cong \mathcal{L}_1(\his)^*$ the map 
$\tr{R \cdot \otimes \omega}:\mathcal{L}_1(\his) \to \mathbb{C}$ determines a unique
bounded $A_{\omega} \in \lhs$, which is self-adjoint when $R \in \lht$ is and $A_{\omega} = \Gamma_{\omega}(R)$.

\subsection{Further Properties}\label{subsec:fup}
The restriction maps $\Gamma_{\omega}$ have further properties of interest,
which we collect here. For the purpose of characterising the relationship between the choice of state $\omega$ and the quantities thus obtained under $\Gamma_{\omega}$ it is convenient to introduce covariant channels.

\begin{definition}
Let $U$ and $V$ be unitary representations of $G$ in Hilbert spaces $\hi$ and $\mathcal{K}$ respectively, and let $\Lambda: \lh \to \mathcal{L(K)}$. $\Lambda$ is called covariant
if $\Lambda (U(g)AU(g)^*) = V(g)\Lambda(A)V(g)^*$ for all $A \in \lh$ and $g \in G$. 
\end{definition}
Covariance for maps acting on the trace class takes an obvious analogous form. The next lemma demonstrates that the restriction map $\Gamma_{\omega}$ applied to an invariant quantity is invariant if $\omega$ is invariant.

{\lemma \label{lem:cova} Let $\Gamma_{\omega}:\lht \to \lhs$ be a restriction channel for some state $\omega$.  $\Gamma_{\omega}$ is covariant if and only if $\omega$ is invariant. If $\Gamma_{\omega}$ is covariant, then $\Gamma_{\omega}(R)$ is invariant if $R$ is.}

\begin{proof}
Writing $\mathcal{U}(g) = U_{\Sy}(g)\otimes U_{\R}(g)$, the first part follows immediately from the covariance condition
\begin{equation}
\sum_{i,j}U_{\Sy}(g)A_iU_{\Sy}(g)^*\omega (U_{\R}(g)B_jU_{\R}(g)^*))=\sum_{i,j}U_{\Sy}(g)A_iU_{\Sy}(g)^*\omega (B_j),
\end{equation} 
to hold for an arbitrary bounded operator $\sum_{i,j}A_i \otimes B_j \in \lh$ (or a limit of such terms) and all $g \in G$.
For the second part, if $R = \mathcal{U}(g)R\mathcal{U}(g)^*$ then clearly $U_{\Sy}(g)\Gamma_{\omega}(R)U_{\Sy}(g)^* = \Gamma_{\omega}(R)$. 
\end{proof}

Thus, for invariant (observable) $R \in \lht$, the only possible way to achieve 
a non-invariant restriction $\Gamma_{\omega}(R)$ is by choosing a non-invariant $\omega$. 
An invariant $\omega$ therefore yields, in the case of number/phase, restricted quantities satisfying (1)-(4) of Proposition \ref{prop:invh}.

A simple calculation shows that the partial trace map ${\rm tr}_{\mathcal{H}_{\R}}:\mathcal{L}_1(\his \otimes \hir) \to \mathcal{L}_1(\his)$ is covariant with respect to $U_{\Sy}$ and $\mathcal{U}$ (and analogously for the partial trace over $\his$). The following is a trivial consequence:
\begin{proposition}\label{prop:pti}
For an arbitrary state $\Omega \in \mathcal{L}_1(\hit)$, ${\rm tr}_{\mathcal{H}_{\R}}(\Omega)$ is invariant under
$U_{\Sy}$ if $\Omega$ is invariant under $\mathcal{U}$. 
\end{proposition}
Hence we have the following:
\begin{corollary}\label{cor:ptti}
${\rm tr}_{\mathcal{H}_{\R}}(\tau_{\T *} (\rho_{\Sy} \otimes \rho_{\R}))$ is invariant under $U_{\Sy}$.
\end{corollary}

\subsection{Restrictions after $\Y$}

The restriction map $\Gamma_{\omega}$ may be composed with the relativisation map $\Y$
to give, for arbitrary $G$,

\begin{equation}
(\Gamma_{\omega}\circ\Y) (A) = \int_{G} U_{\Sy}(g)AU_{\Sy}(g)^* d\mu^{\F}_{\omega}(g),
\end{equation}
where the measure $\mu^{\F}_{\omega}:= \omega \circ \F$ (or, for a density operator $\rho$ corresponding to $\omega$, $\mu^{\F}_{\rho}(X) = \tr{\F(X)\rho}$). As we shall see in the next section, the measure $\mu^{\F}_{\omega}$ dictates the  
proximity of $A$ and $(\Gamma_{\omega}\circ \Y)(A)$.

\section{Localisation and Delocalisation}\label{sec:lad}

\subsection{High Localisation}

Recall (Lemma \ref{lem:n1}) that if a {\sc pom} $\Fsf$ satisfies the norm-1 property then for each $X$ for which $\Fsf (X) \neq 0$  there exists a sequence of unit vectors $(\phi_n) \subset \hir$ such that $\lim_{n \to \infty} \ip{\phi_n}{\Fsf (X) \phi_n} = 1$. This ``localising sequence" $(\phi_n)$ allows for  the expression of expectation values of relative observables to be given in terms of those of absolute quantities to arbitrary precision:
 
\begin{theorem}\label{prop:subrec}
Let $\Fsf$ have the norm-1 property and let $G$ be either $S_1$ {\rm(}which we identify with the interval $[-\pi, \pi)${\rm)} or $\mathbb{R}$ written additively with identity $0$. If $(\phi_n) \subset \hir$ is a sequence of unit vectors which becomes well localised at $g=0$, then for each $A \in \lhs$ and all $\varphi \in \his$
\begin{equation}
\lim_{n \to \infty} \ip{\varphi \otimes \phi _n}{ \Y(A) \varphi \otimes \phi_n} =\ip{ \varphi}{ A \varphi}
\end{equation}
\end{theorem} 

\begin{proof}
Assuming without loss of generality that $\|\varphi\|=1$, we write 
\begin{equation}
\ip{\varphi \otimes \phi_n}{\Y(A) \varphi \otimes \phi_n} =\int_G \ip{\varphi}{U(g) A U(g)^* \varphi}\ip{\phi_n}{\Fsf (d g) \phi_n} = \ip{\varphi}{A \varphi} +c_n
\end{equation}
where $c_n$ is the ``error" for each $n$ which we show goes to zero as $n$ becomes large.
\begin{equation}
\left|c_n \right| = \left|\int_G \ip{\varphi}{\left(U(g) A U(g)^* -A \right) \varphi}\ip{\phi_n}{\Fsf (dg) \phi_n}  \right|.
\end{equation}
Let $\Delta_n = (-1/2n,1/2n)$; then 
\begin{eqnarray}
\left|c_n \right| \leq \left|\int_{\Delta_n}  \ip{\varphi}{\left( U(g) A U(g)^*- A \right) \varphi}\ip{\phi_n}{\Fsf (d g) \phi_n} \right| \\+ \left| \int_{G \backslash \Delta_n} \ip{\varphi}{\left( U(g) A U(g)^* -A \right) \varphi}\ip{\phi_n}{\Fsf (d g) \phi_n}  \right|.
\end{eqnarray}

Now the $\phi_n$ are chosen as follows: for each $n$, there is a $\phi_n$ for which $|\ip{\phi_n}{\Fsf (\Delta_n) \phi_n} - 1| < 1/n$ and $|\ip{\phi_n}{\Fsf (G \backslash\Delta_n) \phi_n}| < 1/n$.
Therefore the second term is bounded above by $2 \no{A} \int_{G \backslash \Delta_n} \ip{\phi_n}{\Fsf (d g) \phi_n}$ which vanishes in the limit. For the first term, writing $f_{\varphi}^A := g \mapsto \ip{\varphi}{(U(g)AU(g)^* - A)\varphi}$ , we estimate 
\begin{equation}\label{eq:limco}
\int_{\Delta_n}|f_{\varphi}^A (g)|\ip{\phi_n}{\F(dg)\phi_n} \leq \int_{\Delta_n} \sup_{g \in \Delta_n} |f_{\varphi}^A (g)|\ip{\phi_n}{\F(dg)\phi_n} \leq \sup_{g \in \Delta_n}|f_{\varphi}^A(g)|\ip{\phi_n}{\F(\Delta_n) \phi_n}.
\end{equation}
From the continuity for self-adjoint $A$ of the real function $f_{\varphi}^A$
it follows that $f_{\varphi}^A(g) \to 0$ as $g \to 0$, and therefore the right hand side 
of \eqref{eq:limco} goes to zero in the $g \to 0$ limit. This extends to arbitrary bounded $A$, as can be seen by decomposing $A$ into real and imaginary parts.
\end{proof}
Therefore by choosing a localising sequence $(\phi_n) \subset \hir$ we can make $\ip{\varphi \otimes \phi_n}{\Y(A) \varphi \otimes \phi_n}$ as close to $\ip{\varphi}{A\varphi}$ as we like.
This result rests crucially on the assumption that the chosen $\Fsf$ satisfies the norm-$1$ property. The main result may thus be rephrased (for the localising sequence $(\phi_n)$) in terms of $\Y$ and $\Gamma$ as follows:
\begin{equation}\label{eq:hll}
\lim_{n \to \infty} (\Gamma_{\phi_n} \circ \Y)(A) = A
\end{equation}
in the weak topology on $\lhs$, that is, in the topology of pointwise convergence of expectation values.

\subsubsection{Examples}
\begin{example}{\bf Qubit algebra.}\label{qinf} \rm
Consider the space $\mathcal{L}(\mathbb{C}^2)$ and a basis of Pauli operators with identity: $\{ \id, \sigma _1, \sigma _2, \sigma _{3}\}$. 
Let  $N_S:=\frac 12(\id +\sigma_3)$ (which has spectrum $\{ 0,1 \}$ and corresponding eigenvectors denoted $\ket{0}, \ket{1}$). 
We addend an infinite dimensional reference system, $\hir$, with ``number basis'' $\{\ket n:\,n\in\mathbb{N}\}$, thus defining  $N_{\mathcal{R}}:=\sum_{n=0}^\infty n\kb nn$ . Then we may use $\Fsf \equiv \Fsf ^{\text{can}}$ on $\hir$ (see eq. \eqref{phase}) to construct $\Y$, and we find that
\begin{align}
\Y(\id) &= \id\otimes \id,\\
\Y(\sigma_{3}) &= \sigma_3 \otimes \id,\\
\Y (\sigma _1) &= \sum_{m\geq0} \left( \ket{0} \bra{1} \otimes \ket{m+1} \bra{m} +  \ket{1} \bra{0} \otimes  \ket{{m}} \bra{{m+1}}\right), \label{eqn:xrel}\\
\Y(\sigma _2) &=  i\sum_{m\geq0}\left(-\ket{0} \bra{1} \otimes \ket{m+1} \bra{m} +  \ket{1} \bra{0} \otimes  \ket{{m}} \bra{{m+1}}\right)\label{eqn:yrel}.
\end{align}
The possibility of good phase localisation of states with respect to $\Fsf ^{\rm can}$
allows the entire qubit algebra $\mathcal{L}(\mathbb{C}^2)$ to be recovered in the following way. Let $A \in \mathcal{L}(\mathbb{C}^2)$ be an arbitrary self-adjoint element, and $\varphi \in \mathbb{C}^2$ an arbitrary unit vector. Define $\ket{\phi_n} =\frac{1}{\sqrt{n+1}} \sum_{j= 0}^n \ket{j}$, which represents an approximately localised phase centred at zero.
 
\begin{remark}\rm The property that $\{\ket{\phi_n}\}$  represents an approximate phase eigenstate at phase value $\theta=0$ for $\Fsf^{\rm can}$ means the following: for every $\delta>0$, the probability of localisation in the interval $[-\delta/2,+\delta/2]$ approaches 1 as $n\to\infty$. Thus, 
\[
\lim_{n\to\infty}\left\langle\phi_n\big|\Fsf^{\rm can}\bigl(\bigl[-\tfrac\delta2,\tfrac\delta2\bigr]\bigr)\phi_n\right\rangle=1\quad\text{for any }\delta\in(0,2\pi).
\]
We sketch a proof of this property. In fact, we will find that the speed of convergence can be specified more precisely: we can allow $\delta$ to tend to zero as $\delta=\delta_n:=(n+1)^{(-1+\epsilon)/2}$ for any $\epsilon\in(0,1)$. We put $\Delta_n:=[-\tfrac{\delta_n}2,\tfrac{\delta_n}2]$ and $X_n=[-\pi,\pi]\setminus\Delta_n$.

Thus, we show that the probability $p_n:=\left\langle\phi_n\big|\Fsf^{\rm can}\bigl([-\pi,\pi]\setminus\Delta_n\bigr)\phi_n\right\rangle\to 0$ as $n\to\infty$. We have
\[\begin{split}
p_n&=\frac1{2\pi(n+1)}\int_{X_n}\left(\sum_{k=0}^ne^{ik\theta}\sum_{\ell=0}^ne^{-i\ell\theta}\right)d\theta
=\frac1{2\pi(n+1)}\int_{X_n}\frac{1-\cos\bigl((n+1)\theta\bigr)}{1-\cos\theta}d\theta\\
&\le\frac1{2\pi(n+1)}\frac1{1-\cos(\delta_n/2)}\int_{X_n}\bigl[1-\cos\bigl((n+1)\theta\bigr)\bigr]d\theta\\
&=\frac1{n+1}\frac1{1-\cos(\delta_n/2)}\left[1-\frac{\delta_n}{2\pi}-2\left.\frac{\sin\bigl((n+1)\theta\bigr)}{2\pi(n+1)}\right|_{\delta_j/2}^\pi\right]\\
&=\frac1{n+1}\frac1{1-\cos(\delta_n/2)}\left[1-\frac{\delta_n}{2\pi}+\frac{\sin\bigl((n+1)\delta_n/2\bigr)}{\pi(n+1)}\right]=\frac1{n+1}\frac1{1-\cos(\delta_n/2)}\left[1-\frac{\delta_n}{2\pi}
+\frac{\delta_n}{2\pi}\frac{\sin\,u}{u}\right]\\
&\le\frac{1}{n+1}\frac1{1-\cos(\delta_n/2)} \\
&\le \frac{1}{n+1}\frac1{\frac{\delta_n^2}8-\frac{\delta_n^4}{24\cdot 16}}
=\frac{8(n+1)^{-\epsilon}}{1-{\delta_n^2}/48}\ \to\ 0\ \text{as }n\to\infty.
\end{split}
\] $\square$

\end{remark}

With $\varphi = c_0 \ket{0} + c_1 \ket{1}$ (normalised)  and $A= a_0 {\id} + \mathbf{a}\cdot \boldsymbol{\sigma} = a_0\id+a_1\sigma_1+a_2\sigma_2+a_3\sigma_3$,
we have 
\[
\ip{\varphi} {A \varphi}= a_0 + 2a_1 {\rm Re} (\bar{c}_0c_1) + 2i a_2 {\rm Im} (c_0 \bar{c}_1) + a_3 (\mods{c_0} - \mods{c_1}).
\]
Evaluating 
\begin{align*}
 \ip{\varphi \otimes \phi_n}{\Y(\id)\varphi \otimes \phi_n}&= 1,\\
 \ip{\varphi \otimes \phi_n}{\Y(\sigma _1)\varphi \otimes \phi_n} &= \frac{n}{n+1}2\,{\rm Re} (\bar{c}_0c_1)=\frac{n}{n+1}\ip{\varphi}{\sigma_1\varphi},\\
 \ip{\varphi \otimes \phi_n}{\Y(\sigma _2)\varphi \otimes \phi_n}&= \frac{n}{n+1}2i\,{\rm Im}(c_0 \bar{c}_1)=\frac{n}{n+1}\ip{\varphi}{\sigma_2\varphi},\\
\ip{\varphi \otimes \phi_n}{\Y(\sigma _3)\varphi \otimes \phi_n} &= \ip{\varphi}{\sigma_3 \varphi},
\end{align*}
we see that as $n$ becomes large, we indeed reproduce for any unit vectors $\varphi \in \mathcal{H}_S$ the expectation values of the basis operators $\id,\sigma_1,\sigma_2,\sigma_3$, and therefore for all $A \in \mathcal{L}(\mathbb{C}^2)$:
\begin{equation}
\lim_{n \to \infty}\ip{\varphi \otimes \phi_n}{\Y(A)\varphi \otimes \phi_n}= \ip{\varphi}{A \varphi}. 
\end{equation}

In conclusion, we see that by making the reference system explicit and taking the limit of a highly phase-localised state, the statistics of any absolute qubit effect offers an accurate representation of the relative qubit effect
in $\mathcal{H}_{\Sy} \otimes \mathcal{H}_{\R}$.
\end{example}

\begin{example}{\bf Finite cyclic group.}\label{ex:fcg} \rm 
We may construct $\Y$ so that $\Y(A)$ is invariant with respect to a unitary representation $U_{\Sy} \otimes U_{\R}$ of some finite cyclic group $G$. 
Hence, let $G$ be a group of cyclic permutations of a finite index set $I$, which can therefore be identified with $G$. 
We consider a Hilbert space  $\mathcal{H}_{\R}$ that allows a direct sum decomposition into subspaces of equal dimension, 
$\mathcal{H}_{\R} = \bigoplus_{i\in I} \mathcal{H}_{\R,i}$, and define a unitary representation $U_{\R}:G \to \mathcal{L}(\mathcal{H}_{\R})$ such that for any $g\in G$, $U_\R(g)$ maps a given orthonormal basis of each $\mathcal{H}_{\R,i}$ to a given orthonormal basis of $\mathcal{H}_{\R,g.i}$ (where $i\mapsto g.i$ denotes an action of $G$ on $I$).  
Let $\{P_i\}$ (or technically the map $i \mapsto P_i$) denote the {\sc pvm} composed of the projections onto $\mathcal{H}_{\R_i}$.
Then $\left( U_{\R}, \{P_i\}, \{ \mathcal{H}_{\R}\} \right)$ is a system of imprimitivity for $G$, with the covariance 
$U_{\R}(g) P_i U_{\R}(g) ^* = P_{g.i}$.

With $U_{\Sy}: G \to \mathcal{L(H_S)}$ any representation of $G$ in $\mathcal{H_S}$ and with $\hit = \mathcal{H_S} \otimes \mathcal{H_R}$, define $\Y : \mathcal{L(H_S)} \to \lht$ by:
\begin{equation}
\Y(A) = \sum _{g \in G} U_{\Sy}(g)AU_{\Sy}(g)^* \otimes P_g ,
\end{equation}
which, with $\mathcal{U} := U_{\Sy} \otimes U_{\R}$, satisfies $\mathcal{U}(g) \Y (A) \mathcal{U}(g)^*=\Y(A)$. Furthermore $\Y$ is a $^*$-homomorphism (since the covariant {\sc pom} $\{P_g\}$ generating $\Y$ is projection valued).
Then there exists a state $\phi \in \mathcal{H}_{\R}$ for which 
\begin{equation}
\ip{\varphi}{A \varphi}=\ip{\varphi \otimes \phi}{\Y(A) \varphi \otimes \phi}
\end{equation}
for all $\varphi \in \mathcal{H_S}$. Indeed, we have
\begin{equation}\label{ydis}
\ip{\varphi \otimes \phi}{\Y(A) \varphi  \otimes \phi} = \sum_{g \in G}\ip{\varphi}{U_{\Sy}(g)AU_{\Sy}(g)^* \varphi} \ip{\phi}{P_g  \phi},
\end{equation}
so that $\phi$ may be chosen to be any unit vector $\phi \in \mathcal{H}_{\R,e}$, where $e$ is the identity element of $G$: in this case
$P_g \phi = \delta_{g,e} \phi$, and \eqref{ydis} collapses to 
\begin{equation}
\ip{\varphi}{U(e)A U(e)^* \varphi}  = \ip{\varphi}{A \varphi} \text{for all~} \varphi,
\end{equation}
i.e., $\Gamma_{\phi}(\Y(A))=A$.
Therefore, by choosing a state localised at the identity of $G$, $\phi \in \mathcal{H}_{\R,e}$, all expectation values of any self-adjoint 
$A \in \mathcal{L}(\mathcal{H_S})$ are precisely those of the relativised $\Y(A) \in \lht$. 
\end{example}

\begin{example}{\bf Unsharp Position.}\label{ex:upo} \rm
The quality of the reference system, understood as the localisability of the covariant quantity, dictates the quality of approximation of relative quantities by absolute ones (see \cite{mlb} for a detailed investigation of this phenomenon.) An intuitive example of this behaviour is 
given by fixing a smeared position $\E^{Q_\R}_{e}$ for the reference, yielding, for sharp
$\E^{Q_{\Sy}}$ of $\Sy$
\begin{equation}
(\Y \circ \E^{Q_{\Sy}})(X) = \int_{\mathbb{R}} e^{iPx}\E^{Q_{\Sy}}(X)e^{-iPx} \otimes \E^{Q_\R}_{e}(dx).
\end{equation}
After restriction, we therefore wish to find the {\sc pom} $\tilde{\E}^{Q_{\Sy}}$ defined by 
\begin{equation}
\ip{\varphi}{\tilde{\E}^{Q_{\Sy}}(X)\varphi}= \ip{\varphi}{(\Gamma_{\phi} \circ \Y)(\E^{Q_{\Sy}}) (X) \varphi} \text{~for all}~ \varphi, ~X
\end{equation}
and to compare this with $\E^{Q_{\Sy}}$.
Moving to the position representation the right hand side of this expression may be written:
\begin{equation}
\int \int \int \chi_{X + y + z}(x)\mods{\varphi(x)}e(y)\mods{\phi(z)}dxdydz,
\end{equation}
which we write as 
\begin{equation}
\int dx \mods{\varphi(x)} F_X(x) = \ip{\varphi}{\tilde{\E}_{Q_{\Sy}}(X)}.
\end{equation}
After some manipulations we find that 

\begin{equation}
F_X(x) = \chi_X * (e * \mods{\phi})(x),
\end{equation}
and that therefore 
\begin{equation}
\tilde{\E}_{Q_{\Sy}}(X) = \chi_X * (e * \mods{\phi})(Q_{\Sy}).
\end{equation}

The spread of the function $\tilde{e} = e * \mods{\phi}$ dictates the (in)accuracy of $\tilde{\E}^{Q_{\Sy}}$
as it approximates ${\E}^{Q_{\Sy}}$. This spread can be quantified in different ways. Using the variance measure, we find that 
$\var(\tilde{e}) = \var(e) + \var\bigl(\mods{\phi}\bigr)$. 
\begin{definition}\label{def:ovw}
For $0 \leq \epsilon < 1$, the {\em overall width} $W_{\epsilon}(p)$ at confidence level $1-\epsilon$ of a probability measure $p$ 
is defined by  
\begin{equation}\label{eq:ovew}
W(p; 1-\epsilon):=\inf_{I}\{|I|: p(I) \geq 1-\epsilon\};
\end{equation}
here the infimum of the lengths, $|I|$, is taken over all intervals $I$ in $\mathbb{R}$.
\end{definition}
 We may also use the overall width $W_{\epsilon}(\tilde{e})$ applied to the density $\tilde{e}$ and the fact that the overall width of a convolution is bounded below by the width of the function with the greatest width, i.e., $W_{\epsilon}(\tilde{e}) \geq \max\{W_{\epsilon}(e), W_{\epsilon}(\mods{\phi}) \}$.

Therefore, the quality (localisability of the smeared position) of the reference system dictates the quality of the representing absolute quantity. Here the inaccuracy inherent in the reference system features as a lower bound on the inaccuracy of the absolute quantity. Even with perfect localisation of the reference with respect to the sharp position, 
there is a residual inaccuracy in the position of the system arising from the unsharpness of
the covariant reference position. To get perfect accuracy, we need the preparation to be highly localised at $0$, and the smearing distribution to be highly localised around $0$.
\end{example}





\subsection{Phase Delocalisation}

At the other extreme (to that of high localisation) we may also consider very poorly localised reference states, including the worst case scenario of complete delocalisation, possible only for compact groups. For concreteness, we focus on the phase case.

Consider covariant phase {\sc pom} $\F$ and invariant state $\omega$. Such a state is completely
delocalised with respect to $\F$, i.e., $\mu ^{\F}_{\omega} \equiv (\omega \circ \F)(X) = \frac{|X|}{2\pi}$. This is the Haar measure on $S^1$ under identification of $S^1$ with $[0,2\pi]$. Thus we may formally write $\mu ^{\F}_{\omega}(d \theta) =  d \theta/2\pi$.

The composition of relativisation and restriction for a delocalised state $\omega$ then has the following effect on $\lhs$ and $\mathcal{L}_1(\his)$ respectively:

\begin{equation}\label{eq:yobs}
(\Gamma_{\omega} \circ \Y)(A) = \frac{1}{2 \pi}\int_{S^1} U(\theta)AU(\theta)^* d \theta \equiv \tau_{\Sy}(A);
\end{equation}
with predual/Schr\"{o}dinger picture
\begin{equation}\label{eq:ystate}
(\Gamma_{\omega} \circ \Y)_*(\rho)=\frac{1}{2 \pi}\int_{S^1} U(\theta)^*\rho U(\theta) d \theta \equiv \tau{{_\Sy}_*}(\rho).
\end{equation}
The above equations hold if (for example) we choose for $\omega$ the density operator $\tau_{\R *}(\rho_{\R})$
for any $\rho_{\R}$ of $\R$.

Thus, from the perspective of $\Y$, completely delocalised reference states give rise to restricted quantities/states which are phase-shift invariant/commute with $N_{\Sy}$. We note that the above relations (eq. \eqref{eq:yobs} and \eqref{eq:ystate}) also generalise, i.e., 
\begin{proposition}
Let $\Lambda$ denote a general 
relative (invariant) self-adjoint operator acting in $\hit$. If $\omega$ is invariant, then
$\Gamma_{\omega}$ is covariant. Then there exists a self-adjoint $A \in \lhs$ for which  $$\Gamma_{\omega}(\Lambda) = \frac{1}{2 \pi}\int_{S^1} U(\theta)AU(\theta)^* d \theta \equiv \tau_{\Sy}(A).$$
\end{proposition}
The proof is a simple corollary of Lemma \ref{lem:cova} and Proposition \ref{prop:invh}. We include it separately to emphasise the general property of invariance of restricted quantities obtained from general invariant quantities and invariant reference states.

\subsection{Discussion}

We have now seen that in the case of perfect reference phase localisation, 
absolute quantities of $\Sy$, along with non-invariant states, provide an adequate theoretical
and empirical account of the statistics produced by invariant quantities of $\Sy + \R$. In this case, the reference can be externalised, and excluded from the description. Though such a localised state is certainly a quantum state, there is a sense in which it may be viewed as classical---if the reference were provided by an abelian algebra (say, $C_0(G)$, embedded in $\mathcal{L}(L^2(G))$), classical pure states correspond to points in $G$ which, are of course, localised (moreover, as shown in \cite{mlb}, ``good" reference frames must be large, pointing to some form of classicality). These observations go some way towards justifying the informal
use of external/classical reference frames in working with non-invariant states of $\Sy$, which has become common-place in the literature.

The $\tau_{\Sy}$ mapping (and $\tau_{\Sy *}$) manifests in (at least) two distinct ways. Initially,
we observed that the assumption that observables of $\Sy$ are invariant implies that $\rho_{\Sy}$ and $\tau_{\Sy *}(\rho_{\Sy})$ cannot be distinguished, yielding an equivalence class of indistinguishable states (indeed, the notion of state could be redefined as this class). Now we see that $\tau_{\Sy *}$ produces a state description for $\Sy$ applicable when $\R$ is prepared in a completely phase-indefinite state (for example, an eigenstate of $N_{\R}$.) Indeed, $\tau_{\Sy *}$
(or \eqref{eq:ystate}) is the ``twirling" operation used in, e.g., 
\cite{brs}, to yield a state description in which some observer of $\Sy$ ``lacks a phase reference". There, it is argued that this is the description one would use if some experimenter 
wished to describe the state of $\Sy$, but had no knowledge of the value of the (classical) phase reference that the state of $\Sy$ implicitly refers to.

In our formulation, this averaging arises as part of the physical description 
of an experiment in which the reference phase is completely phase-indefinite. Number states are of 
this type, and the phase-indefiniteness is a quantum restriction arising from number-phase preparation uncertainty relations. Therefore, we are able to give an alternative interpretation
of to what the ``lack of a phase reference" may be understood to refer: it is the situation
in which a \emph{quantum} phase reference is completely phase-indeterminate. There is no requirement of epistemic arguments regarding information possessed by experimenters, nor any need to refer to classical phase references at all.

\subsection{The findings of \cite{mlb}: General Considerations}

We briefly review the findings of \cite{mlb}, which presents a study of ``intermediate" situations
for the reference frame, in between the very high and very low localisation covered in the 
present paper. \cite{mlb} provides quantitative and operational size-versus-inaccuracy trade-off relations highlighting the necessity of large apparatus for good agreement between some arbitrary effect $A$ and $\Gamma_{\omega}(E)$ for invariant effect $E$. 

We now state the main results more precisely. In the following, the Hilbert spaces involved are assumed finite dimensional. The operator norm $\no{\cdot}$ on $\lhs$ induces the metric $D(A,B):= \no{A-B} = \sup_{\sigma}|\sigma (A) - \sigma (B)|$, which when restricted to the set of effects $\mathcal{E}(\his)$ gives an operational measure of the discrepancy between the effects $A$ and $B$. 

In the following, the quantity $W^0_{\epsilon}(\mu_{\omega_\R}^\F)$ refers to the overall width (cf.~Definition \ref{def:ovw}) of the probability measure $\mu_{\omega_\R}^\F(X)\equiv \omega_{\R}\circ \F(X)$
around $0$, i.e., 
\begin{eqnarray*}
W^0_{\epsilon}(\mu_{\omega_\R}^\F)
:=  \inf \left\{w \big |\ \mu_{\omega_{\R}}^{\F}\bigl( \mathcal{I}(0,w)\bigr)\geq 1-\epsilon\right\},
\end{eqnarray*}
where $\mathcal{I}(0,w)$ denotes the closed interval of width $w \leq 2 \pi$ centred at $\theta \in [-\pi. \pi)$.
Of course, $W^0_{\epsilon}(\mu_{\omega_\R}^\F) 
\geq W_{\epsilon}(\mu_{\omega_\R}^\F)$. This $\inf$ can be replaced by $\min$ (see, e.g., \cite{pape17}, chapter 12).

In the case in which relative quantities of $\Sy + \R$ are obtained through $\Y$, strong localisation of $\omega$ around $\theta = 0$ gives good approximation between absolute and relativised effects: 
\begin{proposition}\label{prop1}
Let $\Y$ be a relativisation map and $\Gamma_{\omega}$ a restriction map.
For an arbitrary effect $A$ and $0\leq \epsilon <1$, it holds that
\begin{eqnarray*}
D\bigl(A, \Gamma_{\omega}(\Y (A))\bigr)
\leq \bigl\Vert [N_{\Sy}, A]\bigr\Vert  
\left( \tfrac12{W^0_{\epsilon}(\mu_{\omega}^{\F})}(1-\epsilon) 
+ \pi \epsilon \right).
\end{eqnarray*}
\end{proposition}

Bad localisation gives bad approximation:
\begin{proposition}\label{prop2}
For $A= \frac{1}{2}\bigl(|0\rangle \langle 0| + |1\rangle \langle 1| 
  +|0\rangle \langle 1| + |1\rangle \langle 0|\bigr)$, it holds that 
\begin{eqnarray*}
D\bigl(A, \Gamma_{\omega_{\R}}( \Y(A))\bigr)
\geq \frac{\epsilon}{2}\Bigl( 1- \cos \left (\tfrac12{W^0_{\epsilon}(\mu_{\omega_\R}^{\F})}\right) \Bigr).
\end{eqnarray*} 
\end{proposition}
And finally, 

\begin{theorem}\label{thm:owb}
Let $A$ be an effect defined by $A=\frac{1}{2}(|0\rangle \langle 0 | 
+ |1\rangle \langle 1| + |0\rangle \langle 1 | + |1\rangle \langle 0|)$.  
For $\omega_{\R}$ satisfying $\Delta_{\omega_{\R}} N_{\R} < \frac{1}{6}$,
\begin{eqnarray*}
D(A, \Gamma_{\omega_{\R}}(\Y(A))) > \frac{1}{32}.
\end{eqnarray*}
For $\omega_{\R}$ satisfying 
$\Delta_{\omega_\R} N_{\R} \geq \frac{1}{6}$, 
it holds 
\begin{eqnarray*}
D(A, \Gamma_{\omega_{\R}}(\Y(A))) 
\geq \frac{1}{32}\left(
1- \cos\left( 
\frac{\pi}{12\Delta_{\omega_{\R}}N_{\R}}\right)
\right). 
\end{eqnarray*}
\end{theorem}

One may go beyond the case of invariant quantities being obtained using $\Y$, and consider general invariant quantities. The following holds for a general (a finite-dimensional, connected) Lie group $G$, acting via projective unitary representations $U_{\Sy}$ and $U_{\R}$ in $\his$ and $\hir$ respectively, with self-adjoint generators $N_{\Sy}$ and $N_{\R}$.

\begin{theorem}\label{th:tradeoff} Recall that $V(A)=\Vert A-A^2\Vert$, $D(A,B) = \no{A-B}$,
$N_{\Sy}$ and $N_{\R}$ are number operators on $\his$ and $\hir$ respectively, and $\omega_{\R} \in \mathcal{S}(\hir)$. Then the following inequality holds:
\begin{align*}
\bigl\Vert [A, N_\Sy]\bigr\Vert 
\leq 
2 D\bigl(\Gamma(E), A\bigr) \Vert N_\Sy\Vert
+2 \left(\omega_{\R}(N_{\R}^2)-\omega_{\R}(N_{\R})^2\right)^{1/2}
\,\bigl(2D(\Gamma(E),A)+V(A)\bigr)^{1/2}.
\end{align*}
\end{theorem}
Therefore, good approximation between arbitrary absolute effects (of $\Sy$) and relative effects
(of $\Sy + \R$), a large spread in the reference's number operator is required, and sufficient for this is good phase localisation.

\subsection{Absolute Coherence}
The stipulation that observable quantities are invariant under the given symmetry action
bears strongly upon the possibility of operationally discerning between coherent superpositions (of eigenstates of the generator, in our case, number) and incoherent mixtures. It will therefore be useful to have a simple working definition and quantification of the \emph{absolute coherence} of states with respect to a number operator (see \cite{bcb1} for other measures and observables). 

\begin{definition}
Let $N$ be a number operator acting in {\rm(}generic Hilbert space{\rm)} $\hi$ and $\rho$ a density matrix. Then $\rho$ is \emph{absolutely coherent} with respect to $N$ if $\tau_*(\rho) \neq \rho$, and {\rm(}absolutely{\rm)} incoherent otherwise.
\end{definition}
 This suggests a measure of absolute coherence:
\begin{definition}
The absolute coherence $\mathcal{C}(\rho) :=\frac{1}{2}\no{\rho - \tau_*(\rho)}_1$.
\end{definition}
Clearly, then, a state $\rho$ is absolutely coherent if and only if it is not invariant under $\rho \mapsto e^{iN\theta}\rho e^{-iN\theta}$.
Let $\{\ket{n}\}$ be a (possibly infinite) orthonormal basis of $\hi$ consisting of eigenvectors of $N$. For an arbitrary state $\varphi = \sum_n c_n \ket{n}$, no self-adjoint operator $A$ commuting with $N$, i.e., no observable quantity, can distinguish between
$P[\varphi]$ and $\tau_{*}(P[\varphi]) = \sum{\mods{c_n}}\ket{n}\bra{n}$, i.e., between a state with absolute coherence and a state without it (Proposition \ref{prop:inv2}).

Since localised states with respect to phase conjugate to $N$ are not invariant, they are necessarily absolutely coherent. We now show that highly localised states have large absolute coherence, for the case of $S^1$ and finite dimensional Hilbert spaces. 

\begin{proposition}
Consider a covariant {\sc pom} $\E$ of $S^1$ on a finite
dimensional Hilbert space $\his$. 
For $\rho$ with overall width $W_{\epsilon}(\mu^{\E}_{\rho})$ 
and for an arbitrary $\epsilon$, 
it holds that 
\begin{eqnarray*}
\mathcal{C}(\rho) \geq 1 -2\epsilon -  \frac{3W_{\epsilon}(\mu^{\E}_{\rho})}{2\pi}
(1-2 \epsilon).
\end{eqnarray*} 
\end{proposition}
\begin{proof}
For simplicity we assume $\Delta = [-W/2, W/2]$ satisfies
$\mbox{tr}(\rho \E(\Delta))= 1-\epsilon$, where $W= W_{\epsilon}(\mu^{\E}_{\rho})$. 
Since the claim is trivial for $3W/2 \geq \pi$, we assume 
$3W/2 <\pi$.  
Since $\mathcal{C}(\rho)= \frac{1}{2}\no{\rho - \tau_*(\rho)}_1 
= \sup_{E: \nul \leq \Vert E\Vert \leq \id}
| \mbox{tr}(\rho E) - \mbox{tr}(\rho \tau(E))|$ 
and $ \nul \leq E(\Delta) \leq \id$ hold, 
we have  
\begin{eqnarray*}
\mathcal{C}(\rho)
&\geq& \mbox{tr}( \rho \E(\Delta))- \frac{1}{2\pi} \int d\theta \mbox{tr}( \rho E(\Delta + \theta)).
\end{eqnarray*}
The first term in the right-hand side is $1-\epsilon$. 
The second term in the right-hand side is now estimated.
For any $\theta$, by the definition of the overall width, 
$\mbox{tr}(\rho E(\Delta +\theta)) \leq 1- \epsilon$ holds. 
For $\theta \in S^1 \setminus [-3W/2,3W/2]$,   
since $E(\Delta +\theta) \cap E(\Delta) = \emptyset$ holds,
we have $\mbox{tr}(\rho E(\Delta +\theta) )\leq \epsilon$. 
 Thus we obtain 
\begin{eqnarray*}
\frac{1}{2\pi} \int d\theta \mbox{tr}( \rho E(\Delta + \theta))
&=& \frac{1}{2\pi} \int^{3W/2}_{-3W/2} d\theta \mbox{tr}( \rho E(\Delta + \theta))
+ \frac{1}{2\pi} \int_{S^1 \setminus [-3W/2, 3W/2]}
 d\theta \mbox{tr}( \rho E(\Delta + \theta))
\\
&\leq& \frac{3W}{2\pi}(1-\epsilon) + \frac{2\pi -3W}{2\pi} \epsilon = \frac{3W}{2\pi}(1-2\epsilon) +\epsilon. 
\end{eqnarray*}
\end{proof}

\subsection{Summary and Analysis of a Potential Objection}\label{subsec:sapo}
Returning to the situation in which we identify a system $\Sy$ and a reference $\R$
with number observables $N_{\Sy}$ and $N_{\R}$ respectively, an apparent circularity arises.
Briefly summarising the story so far, we have argued that observable quantities are (defined
as) those which are invariant under given symmetries. Given a quantum object, this entails
that, in the phase-shift-invariance case, the states $\rho$ and $\tau_{\T_*}(\rho)$ are observationally equivalent and occupy the same equivalence class. From this point of view, absolute coherence is 
not a necessary feature of any description of the quantum object.

However, we have argued that in certain circumstances the given object may be separated into two parts: system $\Sy$ and reference $\R$, and that invariant quantities of $\Sy + \R$ can, in the case of 
$\R$ having a phase quantity possessing the norm-1 property, be arbitrarily well approximated 
by absolute quantities of $\Sy$, given a highly localised state of $\R$. Absolute quantities
are, in particular, sensitive to the difference between an absolutely coherent state $\rho$ and
its invariant, absolutely incoherent counterpart $\tau_{{\Sy}_*}(\rho)$.

Therefore, a description of $\Sy + \R$ in terms of only restricted/absolute quantities of $\Sy$ (not commuting with $N_{\Sy}$), along with states with absolute coherence is possible, \emph{given} states localised with respect to the absolute phase of $\R$, which requires that such states have absolute coherence with respect to $N_{\R}$. This poses a difficulty, since it appears that we claim the description in terms of $\Sy$ alone is a relational one, depending implicitly on $\R$, but have offered no such account for $\R$. Moreover, the \emph{appearance} of absolute coherence (of states of) $\Sy$ appears to depend on the \emph{actuality} of absolute coherence of (states of) $\R$.

What we therefore seek to develop in the next section is a ``fully relational" picture in which
$\Sy$ and $\R$ are treated on an equal footing. What 
emerges is that coherence is a truly relational notion in quantum mechanics, requiring two systems for its definition. From this, through development of the new concept of \emph{mutual coherence}, we are able to give an analysis
of interference experiments in terms of mutual coherence, and provide novel
perspectives on the ``reality" of optical coherence and the subtle issue of superselection rules 
and their relationship with quantum reference frames.

\section{Fully Relational Picture}

In order to obviate the objection raised in the previous section, we now rephrase our findings in
what we describe as a \emph{fully relational picture}, that is, in presenting our main results without taking recourse to absolute coherence, absolute localisation, absolute quantities, {\it etc}. In short, we may present the main theorem of the paper so far---Theorem \ref{prop:subrec}---in a fully invariant manner for $\Sy + \R$. This does not change the mathematical content of the theorem,
but highlights that only invariant states of $\Sy + \R$ are required for good approximation of relational quantities by absolute ones, from which we may conclude that non-invariant states
of $\Sy$ are representative of invariant ones of $\Sy + \R$, in direct analogy to the case of observables. It also motivates the concept of \emph{mutual coherence}, to be presented in the next section.

\subsection{States}

We return once again to the situation wherein $\R$ has phase quantity satisfying the norm-1 property. First, with $(\phi_i) \subset \hir$ a localising sequence (around $0$), we may write equation \eqref{eq:hll} as 
\begin{equation}\label{eq:hlwl}
\lim_{i \to \infty}\tr{\rho \otimes P[\phi_i] \Y(A)} = \tr{\rho A},
\end{equation}
holding for all $\rho \in \mathcal{L}_1(\his)$ and $A \in \lhs$. Thus, since $\Y(A)$ is invariant
we find that 
 \begin{equation}
\lim_{i \to \infty}\tr{\tau_{\mathcal{T}_*}(\rho \otimes P[\phi_i]) \Y(A)} = \tr{\rho A}
\end{equation}
for all $\rho \in \mathcal{L}_1(\his)$ and $A \in \lhs$. Hence, the limit and the resulting approximation may be carried out using only invariant/absolutely incoherent states of $\Sy + \R$. 

Just as absolute quantities of $\Sy$ may be used to represent invariant ones of $\Sy + \R$,
with good approximation coming with good localisation,
a state $\rho$ of $\Sy$ with absolute coherence may be used to represent invariant/absolutely incoherent states
of the form $\tau_{\mathcal{T}_*}({\rho \otimes P[\phi]})$ of $\Sy + \R$, again with good approximation
coming with high phase localisation of $\phi$. 

However, we recall that there are difficulties with ascribing physical significance 
to the absolutely localised state $\phi$ with respect to the absolute phase {\sc pom} appearing in the definition of $\Y$, namely, absolute properties such as localisation and coherence, and absolute quantities such as phase should be understood as relative to a reference, with the reference being only implicit. Therefore, we should seek a consistent formulation in which absolute properties of $\R$ are not required---the description should be entirely relational. The above discussion is a step in this direction: the states $\{\tau_{\mathcal{T}_*}({\rho \otimes P[\phi]})\}$ for some localised $\phi$ no longer ``contain" the localised $\phi$ in the sense that the partial trace over system or reference yields invariant/delocalised/absolutely incoherent states, and therefore a localised state cannot be attributed to $\R$. 

 We now introduce the concept of mutual coherence, which we view as the fully relational version of ordinary coherence.

\section{Coherence Revisited: Mutual Coherence}
The need for a relational understanding of coherence has been clearly enunciated in the literature (e.g., \cite{as, brs, dbrs, dia}). However, little formalism is provided to deal precisely with such a relational notion, and there is no framework capable of making sense of an external classical frame (which appears in \cite{brs, dbrs, dia}). We will analyse this in more detail in Section \ref{sec:cont} ; here we introduce the concept of mutual coherence, which we view as the relational counterpart (and a generalisation of)
of absolute coherence (usually referred to as coherence in standard treatments).

We treat the number/phase case, recalling that we view {\sc pom}s which are invariant under relevant symmetry transformations as the truly observable quantities, and use the term ``invariant quantity" for such objects.

\begin{lemma}\label{lemmaSRT}
It holds that 
\begin{eqnarray}
(\tau_{\Sy}\otimes id)\circ \tau_{\T}
= (id \otimes \tau_{\R}) \circ \tau_{\T}
=(\tau_{\Sy}\otimes \tau_{\R}) \circ \tau_{\T}.
\end{eqnarray}
\end{lemma}
\begin{proof}
We denote eigenvalue decompositions by 
$N_{\Sy}= \sum_n n P^{\Sy}_n$, $N_{\R}= \sum_m m P^{\R}_m$
and $N_{\T}= N_{\Sy}+N_{\R}= \sum_N N P_N$. 
Then,
$P_N = \sum_{n+m= N} P^{\Sy}_n \otimes P^{\R}_m$
and 
\begin{eqnarray*}
\tau_{\T}(A) &=& 
\sum_N P_N A P_N
\\ 
&=& \sum_N \sum_{n_1+m_1=N}
\sum_{n_2+ m_2= N}
(P^{\Sy}_{n_1}\otimes P^{\R}_{m_1})A (P^{\Sy}_{n_2} \otimes P^{\R}_{m_2}). 
\end{eqnarray*}
Since $\tau_{\Sy}(B) = \sum_n P^{\Sy}_n B P^{\Sy}_n$
and $\tau_{R}(C) = \sum_m P^{\R}_m C P^{\R}_m$, a simple calculation shows that
\begin{eqnarray*}
(\tau_{\Sy}\otimes id)\circ \tau_{\T}(A)
&=& \sum_N \sum_{n+m= N} 
(P^{\Sy}_n \otimes P^{\R}_m) A (P^{\Sy}_n \otimes P^{\R}_m)
\\
&=&
(id \otimes \tau_{\R})\circ \tau_{\T}(A)
=(\tau_{\Sy} \otimes \tau_{\R}) \circ \tau_{\T}(A). 
\end{eqnarray*}
\end{proof}

\begin{corollary}\label{cor:coh}
The following two conditions are equivalent.
\begin{itemize}
\item[(i)]\label{item:1}
There exists an invariant quantity $\E$ (thus $\tau_{\T}(\E(X)) = \E(X)$ for all $X$) and an $X$ such that 
\begin{eqnarray*}
\mbox{\rm tr}[(\tau_{\Sy *}(\rho_{\Sy}) \otimes \rho_{\R})\E(X)]
\neq \mbox{\rm tr}[(\rho_{\Sy}\otimes \rho_{\R}) \E(X)].
\end{eqnarray*}  
\item[(ii)]\label{item:2}
There exists an invariant quantity $\E$ ($\tau_{\T}(\E(X)) = \E(X)$ for all $X$) and an $X$ such that 
\begin{eqnarray*}
\mbox{\rm tr}[(\rho_{\Sy} \otimes \tau_{\R *}(\rho_{\R}))\E(X)]
\neq \mbox{\rm tr}[(\rho_{\Sy}\otimes \rho_{\R}) \E(X)].  
\end{eqnarray*}
\end{itemize}
\end{corollary}
\begin{proof}
Assume (i) holds. 
Then for $\E(X)$ satisfying the condition (i), 
\begin{eqnarray*}
\mbox{tr}[(\rho_{\Sy} \otimes \rho_{\R}) \E(X)] 
&\neq&
\mbox{tr}[((\tau_{\Sy * }(\rho_S) \otimes \rho_{\R})\E(X)] \qquad
\\
&=& 
\mbox{tr}[(\rho_{\Sy} \otimes \rho_{\R})(\tau_{\Sy} \otimes id) 
\circ \tau_T(\E(X))]
\\
&=&
\mbox{tr}[(\rho_{\Sy} \otimes \rho_{\R}) (id \otimes \tau_{\R}) \circ \tau_{\T}(\E(X))]
\\
&=&
\mbox{tr}[(\rho_{\Sy} \otimes \tau_{\R *}(\rho_{\R})) \E(X)]. 
\end{eqnarray*}
Thus (ii) follows and vice versa.
\end{proof} 
Moreover, one can observe that 
for condition (i) to hold both $\tau_{\Sy *}(\rho_{\Sy}) \neq \rho_{\Sy}$ 
and $\tau_{\R *} (\rho_{\R}) \neq \rho_{\R}$ must be satisfied. 

Therefore, since $(\tau_{\Sy} \otimes id) \circ \tau_{\T}
= (id \otimes \tau_{\R}) \circ \tau_{\T}$, one can conclude that 
a system state $\rho_{\Sy}$ is coherent relative to the reference state
$\rho_{\R}$ if and only if $\rho_{\R}$ is coherent relative to $\rho_{\Sy}$.
Thus, coherence has a truly relational character. This motivates the following definition:

\begin{definition}\label{def:mcp}
A pair of states $(\rho_{\Sy}, \rho_{\R})$ is called \emph{mutually coherent} if either of the conditions (i) or (ii) of Corollary \ref{cor:coh} holds.
\end{definition}

This may be generalised to an arbitrary (possibly non-separable) state $\Theta \in \mathcal{S}(\hit)$.
\begin{definition}
A state 
$\Theta$ of $\Sy + \R$ is said to be mutually coherent (with respect to $\Sy$ and $\R$) if
\begin{itemize}
\item[(i)'] there exists an invariant observable $\E$ and an $X$ such that 
\begin{eqnarray*}
\mbox{\rm tr}[(\tau_{\Sy*}\otimes id)(\Theta)\E(X)]
\neq \mbox{\rm tr}[\Theta \E(X)]
\end{eqnarray*} or (equivalently)
\item[(ii)'] there exists an invariant observable $\E$ and an $X$ such that 
\begin{eqnarray*}
\mbox{\rm tr}[(id\otimes\tau_{\R*})(\Theta)\E(X)]
\neq \mbox{\rm tr}[\Theta \E(X)].  
\end{eqnarray*}  
\end{itemize}
\end{definition}

A quantitative measure $\mathcal{M}(\Theta)$ of mutual coherence of $\Theta$ may be provided by the quantity (where the supremum is taken over invariant effects)
\begin{eqnarray*}
\mathcal{M}(\Theta):= \sup_{E} \bigl| \mbox{tr}\bigl[((\tau{_{\Sy*}}\otimes id)(\Theta) - \Theta)E\bigr]\bigr|
\\
= \sup_{E} \bigl| \mbox{tr}\bigl[ ((id\otimes\tau{_{\R*}})(\Theta) - \Theta)E\bigr]\bigr| 
\\
=\sup_{E}
\bigl|\mbox{tr}\bigl[((\tau{_{\Sy*}}\otimes \tau{_{\R*}})(\Theta ) - \Theta)E\bigr]\bigr|. 
\end{eqnarray*} 
The above equalities follow easily from Lemma \ref{lemmaSRT}.
For $\Theta = \rho_{\Sy} \otimes \rho_{\R}$ we may write this 
quantitative measure as 
$\mathcal{M}(\rho_\Sy, \rho_\R)$.

We note that this measure of mutual coherence is invariant with respect to the unitary representations $U_{\Sy} \otimes id$ and $id \otimes U_{\R}$ (and therefore also under $U_{\Sy} \otimes U_{\R}$). The following propositions show that if either (system or reference) state is invariant (absolutely incoherent), the mutual coherence vanishes, and at the other extreme (high reference localisation), the mutual coherence is well approximated by the absolute coherence (of the system state).

\begin{proposition}\label{prop:mut}
\begin{eqnarray*}
\mathcal{M}(\rho_\Sy, \rho_\R) \leq \min\{ \mathcal{C}(\rho_\Sy), 
\mathcal{C}(\rho_\R)\}
\end{eqnarray*}
\end{proposition}
\begin{proof}
\begin{eqnarray*}
\mathcal{M}(\rho_\Sy, \rho_\R) 
\leq \frac{1}{2} 
\Vert \tau_{\Sy*}(\rho_\Sy) \otimes \rho_{\R}-\rho_{\Sy} \otimes \rho_{\R}\Vert_{1} 
=\frac{1}{2}\Vert \tau_{\Sy*}(\rho_{\Sy}) - \rho_{\Sy}\Vert_1.
\end{eqnarray*}
\end{proof}
In particular, for an invariant state $\rho_\R$, 
the mutual coherence $\mathcal{M}(\rho_{\Sy}, \rho_{\R})$ vanishes.

\begin{proposition}
For a highly phase-localised state $\rho_{\R}$ of $\R$, $\mathcal{M}(\rho_{\Sy}, \rho_{\R})$
is approximately $\mathcal{C}(\rho_{\Sy})$ (the absolute coherence of $\rho_{\Sy}$).
\end{proposition}
\begin{proof}
(We give a proof for finite dimensional Hilbert spaces.)
For highly localised $\rho_\R$, we have shown (Proposition \ref{prop1}) that 
for an arbitrary effect $E$ of $\Sy$, 
$\Gamma_{\rho_{\R}}(\Y(E))$ well approximates 
$E$, as
\begin{equation}\label{eq:use}
D(E, \Gamma_{\rho_{\R}}(\Y(E)))
\leq \Vert[N_{\Sy}, E]\Vert 
\left(
\frac{1}{2}W^0_{\epsilon}(\mu^{\F}_{\rho_\R})(1-\epsilon) 
+ \pi \epsilon \right).
\end{equation}
From the definition of $\mathcal{M}$ we observe that
\begin{align*}
\mathcal{M}(\rho_{\Sy}, \rho_{\R}) 
&\geq \sup_{E}
| \mbox{tr}((\tau_{\Sy_*}(\rho_{\Sy})\otimes \rho_{\R} 
- \rho_{\Sy}\otimes \rho_{\R})\Y(E)) |
\\
&=
\sup_{E} 
| \mbox{tr}((\tau_{\Sy_*}(\rho_{\Sy}) - \rho_{\Sy})\Gamma_{\rho_{\R}}(\Y(E)))|. 
\end{align*}
For a fixed effect $E$, 
we have, using \eqref{eq:use}:
\begin{align*}
|\mbox{tr}((\tau_{\Sy_*}(\rho_{\Sy})
-\rho_{\Sy})\Gamma_{\rho_{\R}}(\Y(E)))|
&\geq
|\mbox{tr}((\tau_{\Sy_*}(\rho_{\Sy})
-\rho_{\Sy})E)|
- |\mbox{tr}((\tau_{\Sy_*}(\rho_{\Sy})-\rho_{\Sy})
(\Gamma_{\rho_{\R}}(\Y(E))- E))|
\\
&\geq
|\mbox{tr}((\tau_{\Sy_*}(\rho_{\Sy})
-\rho_{\Sy})E)|- 2 D(E, \Gamma_{\rho_{\R}} (\Y(E)))
\\
&\geq
|\mbox{tr}((\tau_{\Sy_*}(\rho_{\Sy})
-\rho_{\Sy})E)|
-2\Vert N_{\Sy}\Vert 
\left(
\frac{1}{2}W^0_{\epsilon}(\mu^{\F}_{\rho_\R})(1-\epsilon) 
+ \pi \epsilon \right).
\end{align*}
Since $E$ is arbitrary, we have 
\begin{equation}\label{eq:rcba}
\mathcal{M}(\rho_{\Sy}, \rho_{\R}) 
\geq \mathcal{C}(\rho_{\Sy})
- 2\Vert N_{\Sy}\Vert 
\left(
\frac{1}{2}W^0_{\epsilon}(\mu^{\F}_{\rho_\R})(1-\epsilon) 
+ \pi \epsilon \right).
\end{equation}

We recall Proposition \ref{prop:mut}, which states that $\mathcal{M}(\rho_\Sy, \rho_\R) \leq \min\{ \mathcal{C}(\rho_\Sy), \mathcal{C}(\rho_\R)\}$. In the high localisation regime for $\rho_{\R}$, the second expression on the right side of \eqref{eq:rcba} becomes small and we may assume that $\mathcal{C}(\rho_\Sy) \leq \mathcal{C}(\rho_\R)$. Therefore, we have the approximate equality
$\mathcal{M}(\rho_\Sy, \rho_\R) \approx \mathcal{C}(\rho_\Sy)$, with the quality of approximation
becoming arbitrarily good as $\rho_{\R}$ becomes highly localised.
\end{proof}

In other words,  
 the mutual coherence takes on the appearance of absolute coherence in the high reference localisation limit. 
 
 We shall soon discuss the role of mutual coherence in interference phenomena and superselection rules. First, we note the following observations relating to approximation of 
relational observables by absolute quantities for some $\Y(A)$ constructed using a phase {\sc pom} possessing the norm-1 property.

Suppose that we have some non-invariant state $\rho_{\Sy} \neq \tau_{\Sy *}(\rho_{\Sy})$. Then for arbitrary $A$, $\tr{\rho_{\Sy}A}$ and $\tr{\rho_{\Sy} \otimes \rho_{\R}\Y(A)}$ can be made equal only if $(\rho_{\Sy}, \rho_{\R})$ is mutually coherent. The reason is clear: Suppose that $(\rho_{\Sy}, \rho_{\R})$ is not mutually coherent. Then, by Definition \ref{def:mcp}, for any invariant $R \in \lht$ it must be that $\tr{\rho_{\Sy} \otimes \tau_{\R*}(\rho_{\R})R} = \tr{\rho_{\Sy} \otimes \rho_{\R} R}$. Then, $\tr{\rho_{\Sy} \otimes \tau_{\R *}(\rho_{\R})R}=
\tr{\rho_{\Sy} \Gamma_{\tau_{\R *}(\rho_{\R})}(R)}$ (see Subsection \ref{subsec:fup}) and (by Lemma \ref{lem:cova}) $\Gamma_{\tau_{\R *}(\rho_{\R})}(R)$ is invariant. But, due to the invariance of $\Gamma_{\tau_{\R *}(\rho_{\R})}(R)$,  $\tr{\rho_{\Sy}\Gamma_{\tau_{\R *}(\rho_{\R})}(R)} = \tr{\tau_{\Sy *}(\rho_{\Sy})\Gamma_{\tau_{\R *}(\rho_{\R})}(R)}$. This latter quantity can never equal $\tr{\rho_{\Sy}A}$ for non-invariant $A$ and non-invariant $\rho_{\Sy}$. 

In fact, this also establishes the more general result that for any invariant $R\in \lht$ and non-invariant $\rho_{\Sy}$,  $\tr{\rho_{\Sy} \otimes \rho_{\R}R} = \tr{\rho_{\Sy}A}$
for arbitrary $A$
only if $(\rho_{\Sy}, \rho_{\R})$ is mutually coherent. Theorem \ref{thm:owb} demonstrates for a specific non-invariant effect $A \in \mathcal{E}(\his)$ that for a $\rho_{\R}$ with poor localisation ($\Delta_{\rho_{\R}}N_{\R} < 1/6$), the discrepancy $D(A,\Gamma_{\rho_{\R}}(\Y(A))) > 1/32$, and therefore that $\tr{\rho_{\Sy}A}$ and $\tr{\rho_{\Sy} \otimes \rho_{\R}\Y(A)}$ cannot even be close in this case.

\if
Also,
\begin{lemma}
Let $\varphi$ be an absolutely coherent pure state of $\Sy$ and $A \in \lhs$ non-invariant. If there exists invariant $R \in \lht$ and pure $\Theta \in \mathcal{L}_1(\hit)$ for which
$\ip{\varphi}{A \varphi}= \ip{\Theta}{R \Theta}$, then $\Theta$ is mutually coherent.
\end{lemma}
\begin{proof}
Let $\Theta$ be non-mutually coherent. Then for any invariant $R \in \lht$, 
\end{proof}
\fi

\if
This gives an interesting physical observation. 
Let us suppose that Alice and Bob each has 
her (his) own classical reference frame. 
But they do not share the frames. That is, 
Alice does not know how Bob's reference frame is 
related to hers, and vice versa. 
In this case, Alice can measure noninvariant observable 
$A$ with respect to her own classical reference frame. 
Suppose that Alice tells this outcome (expectation value) to Bob. 
This gives no information to Bob. That is, even knowing Alice's 
expectation value, 
Bob cannot tell his future measurement result at all. 
On the other hand, if Alice conclude that 
a state $\rho_S$ is coherent with respect to 
her own classical reference frame and tells it to Bob, 
this will let Bob confirm that the same state 
is also coherent with respect to his own classical reference 
frame. Thus even without sharing the reference frames, 
Alice and Bob possessing own classical reference frame 
can share information on coherence. 
\fi

\section{Measurement}
The enquiry thus far has been of a kinematical nature. We now consider the important role 
played by dynamical evolution of states and ensuing measurements, considered in light of the relational perspective presented. 
The main theorem regarding the role of symmetry in quantum measurements is the Wigner-Araki-Yanase (WAY) theorem \cite{wig1,buschtrans,ay1}, which addresses measurements in the presence of additive conserved quantities of system-plus-reference. After presenting the essentials of the quantum theory of measurement required for our analysis, we present a ``strong" form of the WAY theorem, assuming the system on its own has a conserved quantity, followed by two readings of the WAY theorem: the orthodox reading, as presented in \cite{LB2011}, and the relational viewpoint.

\subsection{Measurement Theory: Brief Overview}

We briefly describe the quantum theory of measurement  of relevance to this work. For simplicity
we present these concepts without the impositions of symmetry. 

Let $\mathcal{H_S}$ be the Hilbert space representing a quantum system $\mathcal{S}$ under investigation, $\hia$ that representing a measuring apparatus, with the combined system then given by $\mathcal{H}=\his \otimes \hia$. A unitary mapping $U: \mathcal{H} \to \mathcal{H}$ models a measurement interaction, serving to correlate the states of the
system to those of the apparatus during an interaction period $T$. 
The specification of a self-adjoint ``pointer observable" $Z$ on 
$\mathcal{H_A}$, a fixed state $\phi \in \mathcal{H_A}$  (which for convenience is assumed to be pure) and the scaling function $f$ (which maps the values of the pointer
to those of the measured observable) then fix
the \emph{measurement  scheme} 
$\mathcal{M} \equiv \langle \mathcal{H_A},U,\phi ,Z,f\rangle$ for observable $\Esf$ of $\Sy$. With $\Psi _{T }=U(\varphi \otimes \phi )\in 
\mathcal{H}$, $\mathcal{M}$ must satisfy the
\emph{probability reproducibility condition}:
\begin{equation}\label{prc}
\left\langle \Psi _{T }|\id \otimes \Esf^{Z}\bigl(f^{-1}(X)\bigr)\Psi _{T}
\right\rangle \equiv \left\langle \varphi |\Esf(X)\varphi \right\rangle ,
\end{equation}
where $ \Esf^{Z}\bigl(f^{-1}(X)\bigr)$ are spectral projections of $Z$, and \eqref{prc} holds for all $\varphi$ and $X$. In words,
\eqref{prc} stipulates that the outcome
distribution for $\Esf$ in any state $\varphi $ may be recovered from the pointer
statistics in the final state $\Psi _{T}$. Conversely, given a
measurement scheme as described above, this relation determines the measured
observable $\Esf$. 

A measurement (scheme) is said to be \emph{repeatable} if, upon immediate repetition
of the measurement, the same outcome is achieved with certainty. This may be written:
\begin{equation}\label{rep}
\left\langle \Psi _{T} |\Esf(X)\otimes \Esf^{Z}\bigl(f^{-1}(X)\bigr)\Psi _{T}
\right\rangle =\left\langle \varphi |\Esf(X)\varphi \right\rangle.
\end{equation} 
We note that \eqref{prc} does not entail \eqref{rep}, and therefore the question of repeatability must be treated independently 
of that of probability reproducibility. 

\subsection{Conservation Laws: Strong and Weak WAY Theorems}\label{subsec:sww}

We present here the standard version and interpretation of the theorem of Wigner, Araki and Yanase
(WAY), as presented in \cite{LB2011}, giving both the no-go part, prohibiting sharp, repeatable measurements of an observable (in the ordinary sense) which does not commute with (the system part of) an additive conserved quantity, and the positive part demonstrating conditions under which good approximation can be achieved. We begin with a version of the WAY theorem subject to a stronger assumption than is typical---that subsystem quantities are conserved---which we therefore 
refer to as the \emph{strong WAY theorem}.

\subsubsection{Strong Conservation}

The stipulation that observability entails invariance follows as a theorem in the quantum theory of measurement from a constraint on $\mathcal{M}$, namely the conservation of some quantity of $\his$ (see also \cite{pape17}, ch. 21).

Consider the strongly continuous unitary group described by the operators $U_{\Sy}(t) \equiv e^{itL_{\Sy}}$, with $t \in \mathbb{R}$ or 
$t \in [0, 2 \pi]$ and $L_{\Sy}$ a self-adjoint operator acting in $\his$. Then the following holds.\footnote{A similar statement to Proposition \ref{prop:ssr1} was proven on the basis of stronger assumptions by S. Tanimura \cite{Tanimura2011}.}

\begin{proposition}\label{prop:ssr1}
Suppose that for any measurement scheme $\mathcal{M}$ for $\Esf$, $L_{\Sy}$ is conserved, i.e., $[U,U_{\Sy}(t) \otimes \id]=0$. Then
\begin{equation} \label{eqn:ssr}
U_{\Sy}(t) \Esf(X) U_{\Sy}(t)^* = \Esf (X)
\end{equation}
for all value sets $X \in \mathcal{F}$ and all $t$.
\end{proposition}
\begin{proof}
It is sufficient to consider measurement schemes for which the pointer function $f$ in \eqref{prc} is
the identity map. Equation \eqref{prc} gives 
\[
\ip{U\fii\otimes\phi}{\id\otimes{\mathsf E}^Z(X)U\fii\otimes\phi}=\ip{\fii}{\Esf(X)\fii}
\]
for all object initial states $\fii$ and all $X$. Replacing $\fii$ with $U_{\Sy}(t)\fii$ does not change the left hand side,
and therefore the right hand side is also unchanged, immediately giving (\ref{eqn:ssr}). 
\end{proof}



\begin{remark}\rm
If $\Sy$ is considered to be an {\em elementary} system,  the operators $U_{\Sy}(t)$ comprise an irreducible representation, in which case the only effects satisfying \eqref{eqn:ssr} are those of the form $\Esf (X) = c_X \id$ ($0 \leq c_X \leq 1$), i.e., 
the trivial effects. If, however, $\Sy$ is more complex, comprising several elementary systems for instance, $\Sy$ can be separated into an ``object" system $\Sy _{\mathcal{O}}$ and a ``reference" system $\Sy_{\mathcal{R}}$. Absolute quantities of $\Sy _{\mathcal{O}}$ then function only
as representatives of observables of $\Sy$ as a whole, and depending on the composition
of $\Sy_{\mathcal{R}}$ may or may not accurately represent observables; as we have seen,
$\Sy_{\mathcal{R}}$ under certain localisation requirements allows for the absolute quantities 
of $\Sy_{\mathcal{O}}$ to be good representations of observables. In particular, 
$\Sy_{\mathcal{R}}$ may be viewed as a measuring apparatus, as is the case in the WAY theorem, which we initially present in its conventional form, and subsequently reinterpret in a relative vein.
\end{remark}




\subsubsection{Weak Conservation: Wigner-Araki-Yanase Theorem}

In this instance a conservation law is applied to the system-apparatus combination, but is not assumed to hold `locally', i.e., for the system under investigation and apparatus separately.  We present the traditional reading of the WAY theorem.

\begin{theorem} {\rm(}Wigner-Araki-Yanase{\rm)}
Let $\mathcal{M} := \left\langle \mathcal{H}_A, U, \phi, Z, f \right\rangle$ be a measurement of a discrete-spectrum self-adjoint operator $A$ on $\mathcal{H_S}$, and let $L_{\Sy}$ and $L_{\Ap}$ be bounded self-adjoint operators on $\mathcal{H_S}$ and $\mathcal{H_A}$, respectively, such that 
$[U, L_{\Sy} +L_{\Ap}] = 0$. Assume that $\mathcal{M}$ is repeatable or $[Z,L_{\Ap}]=0$. Then $[A,L_{\Sy}]=0$.
\end{theorem}
We refer to \cite{LB2011} for a proof. Following Ozawa \cite{ozawa1}, we refer to the condition that the pointer observable $Z$ commutes with the apparatus part of the conserved quantity $L_{\Ap}$ as the
\emph{Yanase condition} \cite{yan1}. In the case that $[A,L_{\Sy}]\neq 0$, there is a positive counterpart to the impossibility result: approximate measurements of $A$, with approximate repeatability properties, are feasible, with increasingly good approximation properties the larger the variance $\left(\Delta_{\phi}L_2\right)^2$
becomes (see \cite{LB2011}, where more general measures of spread are also considered) and indeed that such large ``spread" is \emph{necessary} for good measurements of $A$.

Thus, in its usual reading, the WAY theorem does not prohibit accurate measurements of \emph{unsharp} observables 
which do not commute with $L_{\Sy}$, thus leaving room for a positive rephrasing of the theorem where a smeared, approximate version of $A$ can be measured accurately. We now address this point, arguing that, just as in the discussion following the strong version of the theorem, one should actually conclude that the measured observable in the WAY theorem must be understood as a representative of a relative observable of system and apparatus combined.

The standard interpretation of the WAY theorem states that any sharp $A$ not commuting with (the object part $L_1$ of) an additive conserved quantity $L$ for which the Yanase condition is satisfied ($[Z,L_2]=0$) cannot be measured precisely. Moreover, good approximation \emph{can} occur if there is large uncertainty with respect to the apparatus part $L_2$ of the conserved quantity in the initial state $\phi$ of the apparatus, i.e., if $\left(\Delta_{\phi}L_2 \right)^2$ is large.

In light of the theme of this paper, namely understanding the consequences of the principle that observables are invariant, we may reconsider the message of the WAY theorem. We recall that, for fixed $\phi \in \hir$, the equation
\begin{equation}
\ip{\varphi}{\Esf(X)\varphi}=\ip{\varphi \otimes \phi}{U^* \id \otimes \E^Z (X) U \varphi \otimes \phi},
\end{equation}
when stipulated to hold for all $X, \varphi$ determines the \po $\Esf$. In other words, $\Esf(X) = \Gamma_{\phi}(U^*\id \otimes \Esf^Z(X)U)$. Given the Yanase condition ($[Z,L_2]=0$)
and the conservation law ($[U,L]=0$), it follows that $[U^*ZU,L]=0$. Writing 
$U^*ZU \equiv Z(\tau)$, it is therefore evident that $Z(\tau)$ is invariant under the symmetry
generated by $L$ and that, furthermore, in the limit that $\left(\Delta_{\phi}L_2 \right)^2$
becomes large, $A$ (which is not necessarily observable) can become a good approximation of the 
observable $Z(\tau)$. 

If $L_2$ is the shift-generator in a conjugate quantity (e.g., a number operator generating phase shifts), then large $L_2$ spread in $\phi$ corresponds to high localisation
with respect to $\phi$ in the conjugate quantity, completely in line with the view that for
$A$ to be a good representative of an invariant observable, the reference system must be highly localised with respect to a phase-like quantity, {\it \`{a} la} $\Y$. This also sheds light on the reason
that $L_2$ must have large spread in the \emph{initial} state of the apparatus $\phi$. This view of the ordinary WAY theorem then arises when the strong WAY theorem is applied to system-plus-apparatus together, viewed as an isolated system.

\begin{example} \rm
{\bf Ozawa model of an unsharp position measurement: relative versus absolute position.}

The relational view just discussed may be exemplified in a position measurement model of Ozawa, introduced in 
\cite{Ozawa} and analysed further in \cite{Loveridge/Busch}, \cite{ldlthesis}. We consider the momentum--conserving position measurement scheme in which $\Sy+\R$ interacts with two apparatus
systems $\Ap+\Bp$.
This scheme  measures the absolute position $Q$ with the pointer observable $P_\Bp-P_\Ap$, a relativised momentum. Contrary to the claim in \cite{Ozawa}, a WAY-type limitation is exhibited for this model.
However, we show that the same scheme may be used to measure the relative position observable, $Q \otimes \id - \id \otimes Q_{\mathcal{R}}\equiv Q-Q_\R$; in this case there is no localisation
requirement at all for good measurements, as would be expected since $Q-Q_{\R}$ is already shift-invariant. Moreover, we demonstrate that the absolute position $Q$ well represents $Q-Q_{\R}$ precisely when $Q_{\R}$ is highly position-localised, corresponding to a large momentum spread in
the reference system $\R$ (cf.~Example \ref{ex:upo}).

The unitary measurement coupling is given by $U=e^{i \frac{\lambda}{2}(Q-Q_{\mathcal{R}})(Q_\Ap - Q_\Bp)}$, which commutes with the total momentum $P +P_{\R}+P_{\Ap}+P_{\mathcal{B}}$ (notice also that $P +P_{\R}$ is separately conserved and therefore falls under the remit of Proposition \ref{prop:ssr1}). Subsequently the pointer $Z$, given by the difference of momentum operators $P_\Bp-P_\Ap$, is measured.
We consider the initial state $\Psi_0(x,y,u,v)=\varphi(x)\phi(y)\xi_a(u)\xi_b(v)$, where $x,y,u,v$ are spectral values of $Q, Q_{\R},P_{\mathcal{B}} - P_{\Ap}, P_{\Ap}+P_{\mathcal{B}}$ respectively. The unique measured {\sc pom} $\widetilde{\mathsf{E}}: \mathcal{B} (\mathbb{R}) \to \mathcal{L}\bigl(\mathcal{H}\otimes \mathcal{H}_{\R}\bigr) \equiv \mathcal{L}(L^{2} (\mathbb{R}^2)) $
is  extracted from the condition
\begin{equation}\label{oz1}
\left\langle \Psi _{\tau }|\id\otimes \id\otimes 
\mathsf{E}^{Z}(f^{-1}(X))\otimes \id\Psi _{\tau }\right\rangle
=\left\langle \varphi \otimes \phi| \widetilde{\mathsf{E}}(X)\varphi\otimes \phi \right\rangle,
\end{equation}
required to hold for all $\varphi,~\phi$. It then follows that
\begin{equation}
\widetilde{\mathsf{E}}(X)=\chi _{X}\ast \widetilde{e}^{(\lambda )}(Q-Q_{\mathcal{R}}),
\end{equation}
where the right-hand side is the convolution of the set indicator function $\chi_X$ with the probability distribution $\widetilde{e}^{(\lambda )}(x)=\bigl\vert \xi_a^{(\lambda )}(x)\bigr\vert ^{2}$ with $\xi_a ^{(\lambda)}(s) = \sqrt{\lambda}\xi_a(\lambda s)$.

We see that $\widetilde{\mathsf{E}}$ is a smeared version of $\Esf^{Q-Q_{\mathcal{R}}}$. As such the former can be considered an approximation of the latter, and we may quantify the inaccuracy or error of that approximation by the variance of the distribution function $\widetilde{e}^{(\lambda )}$ (other measures such as overall width can also be used: see \cite{ldlthesis}). The variance of $\widetilde{e}^{(\lambda )}$ is ${\var}(\widetilde{e}^{(\lambda )})=\frac{4}{\lambda ^{2}}{\var}\left\vert \xi_a\right\vert ^{2}$. Therefore by tuning $\lambda$ to be large, arbitrarily accurate measurements of $Q-Q_{\mathcal{R}}$ can be achieved with no localisation requirement on the reference system $\R$.

The absolute position $Q$ then acts as an approximation of the observable $Q-Q_{\R}$, the approximation becoming good with good $Q_{\R}$ localisation. By fixing $\phi$ in the initial state $\Psi_0$, the measurement scheme can be viewed as ``measuring" a {\sc pom} $\Esf$ for $\Sy$:
$$
\ip{\varphi \otimes \phi}{\widetilde{\Esf}(X)\varphi \otimes \phi} =: \ip{\varphi}{\Esf (X) \varphi}.
$$ 
This is of the form $\Esf (X) =  \chi _X * e^{(\lambda)}(Q)$ with $e^{(\lambda)}$ given by 
$e^{(\lambda)}(x) =  \bigl|{\phi}\bigr|^2*\bigl|{\xi_a ^{(\lambda)}}\bigr|^2(x)$. 
The probability distribution for the relative position
has thereby been re-expressed in terms of a smeared distribution for the absolute position by considering a fixed reference state $\phi$. 
The approximation error of $\Esf$ relative to $\Esf^Q$ is given by
$\var (e^{(\lambda)}) = \var \mods{\phi} + \frac{4}{\lambda ^2} \var \mods{\xi_a}$.
The probability distributions corresponding to the relative coordinate in the states $\varphi \otimes \phi$ become indistinguishable from those of the absolute coordinate $Q$ in the limit that the localisation of the state $\phi$ with respect to $Q_{\mathcal{R}}$ 
is arbitrarily good (provided also that $\lambda$ is tuned to be large.)

This model, therefore, highlights how relative observables such as $Q-Q_{\R}$ may be measured
whilst preserving overall symmetry imposed by the conservation of total momentum,
with a pointer observable that also respects symmetry. In this case, the apparatus is not required to function as a reference system, which is internal to the measurement
device and whose localisation controls the quality of the approximation by the absolute quantity.
\end{example}
We have therefore seen that the picture of observables as relative quantities may be well maintained in the presence of dynamics. It was shown that the WAY theorem has a relative interpretation, and the model of Ozawa provided a measurement scheme for the relative position
observable, which could be re-expressed as an accurate measurement of an absolute position
precisely when the reference system was well position-localised. Absolute quantities
were seen to be good representatives of observables again in the high reference localisation limit,
the interpretation therefore not differing from the ``static" case. We now consider the impact of this enquiry on the status of superpositions, interference and superselection rules.

\section{Interference Phenomena}
We begin this section with a typical analysis of interference phenomena. Absolute coherence is the usual requirement for interference effects to manifest. We show that, from our relational perspective, 
mutual coherence replaces absolute coherence in regard to interference. We provide several models
which serve to illustrate the problem of observing coherence, and then turn to the role played by high phase localisation of the reference.

The stipulation that observables are phase-shift invariant implies that relative phase 
factors\footnote{It may be of help to note here that there are two distinct uses of the term {\em relative phase}: up to this point, we have used the term to designate observables that are defined as relative to a reference system; in the context of interference experiments, one speaks of {\em relative phase factors} to indicate that these are actually phase differences between two states appearing in superposition, contrasting these observable quantities with the unobservable overall phase factors one may attach to a state vector.} between states in a superposition 
of number eigenstates (with differing eigenvalue) ought to be unobservable. We have seen, for example, that the state $P[\varphi]$ of $\Sy$ with $\varphi =  \bigl( \ket{0} + e^{i \theta} \ket{1}\bigr)/{\sqrt{2}}$ cannot be distinguished from $\tau_{\Sy_*}(P[\varphi]) = 1/2(\ket{0}\bra{0} + \ket{1}\bra{1})$ (which has no dependence on $\theta$)
by any quantity commuting with $N_{\Sy}$. 

The usual reading of $\theta$-dependent expectation values (appearing in states akin to that discussed above) arising in measurements is that a coherent (in our language, absolutely coherent) superposition has been prepared and measured (to be coherent). Since this is equivalent to 
measuring an absolute quantity, this conclusion warrants further scrutiny. There is a history of debate and controversy surrounding the meaning 
of $\theta$-dependent expectation values in superpositions of number states (also understood as charge eigenstates) in the subject of superselection rules \cite{www, as, www70}, and of photon number states in the so-called optical coherence controversy \cite{molmer, dia}, regarding the reality of coherent states in describing the output field of a laser.

In the forthcoming subsections we motivate the question of relative phase-factor sensitivity more formally and in a dynamical context, by first discussing a generic interference experiment, followed by three model considerations. We finish with a discussion of the role of high phase localisation and the accompanying interpretation of the measurement statistics. We then use
our findings to analyse points of agreement and points of friction between our interpretation
and those appearing in the literature.

\subsection{Interferometry}

Ramsey interferometry exemplifies the typical form of interference experiments (see, e.g., \cite{dbrs}). Here, an atom enters a cavity in its ground state $\state{g}$, interacts with the cavity, and exits in a superposition of ground and excited states. At the level of the atom the following sequence (or similar) of (unitary) state evolutions is often given:
\begin{align}\label{eq:intf}
\psi _i \equiv \state{g} &\to \frac{1}{\sqrt{2}}(\state{g} - i  \state{e}) \to  \frac{1}{\sqrt{2}}(\state{g} - i e^{-i \theta} \state{e})\\ &\to \sin \left( \frac{\theta}{2} \right)\state{g} - \cos \left( \frac{\theta}{2} \right) \state{e} \equiv \psi_f,\label{eq:int2}
\end{align}
 where $\state{e}$ represents an excited state of the atom.
If the observable $P_g \equiv \state{g}\dstate{g}$ is measured in the final state, we see 
$\ip{\psi _f}{P_g \psi_f}= \sin ^2 \left( {\theta}/{2} \right)$. The orthodox reading 
of such a measurement is that this $\theta$--dependent
probability distribution for the observable $P_g$ in the state $\psi_f$ validates the 
coherence of the superposition state $\frac{1}{\sqrt{2}}(\state{g} - i e^{-i \theta} \state{e})$.

However, the Hamiltonian generating such
an evolution certainly does not commute with $P_g,~P_e $ (i.e., $N_{\Sy}$) and is, itself, therefore not phase-shift invariant (and thus not (an) observable). Equations \eqref{eq:intf} and \eqref{eq:int2} must, if applicable at all, therefore be viewed as approximate, reduced descriptions of the true, energy-conserving dynamics of system-plus-cavity. 

We will obtain a consistent description of measurements which at first sight appear
sensitive to relative phase factors between number superpositions. Keeping in mind
that observables are invariant and that states are class representatives, we may 
obtain statistics which look \emph{as if} absolute quantities have been measured
or, alternatively (and equivalently), that relative phase factors across number eigenspaces have been observed.
Again, this is a reduced, approximate description and not a true representation of the state 
of affairs. The models to be presented have strong formal similarities to the case of observability of phase factors between states of different charge and different baryon number, allowing for comparison to the issue of whether superselection rules 
may be obviated in practice (cf. \cite{as, brs, dbrs}). We show that all such attempts may be phrased purely in terms of measurements of relative quantities (i.e., observables), highlighting the fact the absolute quantities are never measured.

\subsection{Model 1: Two--level System }
We first consider a model in Hilbert space dimension $4$ to show
how to dynamically introduce a relative phase factor between number states (of the same total number eigenvalue), whilst  respecting symmetry. The restriction to low dimensions highlights the relational nature of the relative phase factor. The generic structure of this model can then be applied to the scenario
where the reference system's Hilbert space has infinite dimension, which resembles the situation for which there have been claims purporting to ``lift" \cite{brs} or evade superselection rules. However, we argue that there is no reason for the interpretation of measurement statistics in the infinite dimensional setting to be different from the model discussed below, except for the observation that with infinite dimensional reference systems, expectation values of absolute quantities (can be made to) agree arbitrarily well with those of the relative ones (contingent on a choice of reference state).

Let $N_{\Sy} \in {\lhs} \equiv \mathcal{L}(\mathbb{C}^2)$ be a number operator so that $N_{\Sy}\state{0}=0,~N_{\Sy} \state{1}=\state{1}$, and let $N_{\R} \in \lhr$ have the same definition. 
Any self-adjoint operator $A \in \lht$ must commute with $N :=N_{\Sy} \otimes \id+ \id\otimes N_{\R}$ if it is to be deemed observable.

We introduce two unitary operators $U_1$ and $U_2$ which represent two stages of time evolution,  defined as
\begin{align*}
\state{0}\state{0} & \overset{U_1}{\longrightarrow} \state{0} \state{0} \overset{U_2}{\longrightarrow} \state{0}\state{0};\\
\state{0}\state{1} & \overset{U_1}{\longrightarrow}\frac{e^{-i \frac{\theta} {2}}}{\sqrt{2}} \left( \state{0}\state{1} + e^{i \theta} \state{1} \state{0} \right) \\ & \overset{U_2}\longrightarrow  \left( \cos \left( \frac{\theta}{2} \right) \state{0}\state{1} -i \sin \left( \frac{\theta}{2} \right) \state{1} \state{0} \right);\\
\state{1}\state{0} &\overset{U_1}\longrightarrow \frac{e^{-i \frac{\theta} {2}}}{\sqrt{2}} \left( \state{0}\state{1} - e^{i \theta}\state{1} \state{0} \right)\\  & \overset{U_2}\longrightarrow \left(-  i \sin \left( \frac{\theta}{2} \right)\state{0}\state{1} + \cos \left( \frac{\theta}{2} \right)\state{1} \state{0} \right);\\
\state{1}\state{1}& \overset{U_1}\longrightarrow \state{1} \state{1} \overset{U_2}\longrightarrow \state{1} \state{1};
\end{align*}
and it can be seen that $[U_1,N]=[U_2,N]=0$. Furthermore, 
it is important to note that $U_2$ does not depend on $\theta$, which can be seen by the action of $U_2$ 
on the initial product states given by
$U_2 \state{0}\state{1} = \frac{1}{\sqrt{2}} \bigl(\state{0}\state{1} + \state{1}\state{0}\bigr)$ 
and $U_2\state{1}\state{0} = \frac{1}{\sqrt{2}} (\state{0}\state{1} - \state{1}\state{0})$.
The purpose of applying $U_2$ is to allow a measurement of an invariant quantity of $\Sy$,  
which gives rise to $\theta$--sensitive measurement statistics.

In other words, $U_1$ introduces the factor $\theta$, $U_2$ redistributes the $\theta$-dependence, so that the measurement of an invariant $\Sy$-quantity depends on $\theta$, which then
validates the superposition present after $U_1$.

Writing $P_0 :=\state{0} \dstate{0}$,  $\psi = \state{0}\state{1}$, and noting that 
$\tau{_{\Sy}} (P_0) =P_0$, 
we compute post-$U_2$ statistics: 
\begin{equation}\label{tls}
{\rm tr} \bigl[ P_0 \otimes \id \tau_{\T *}( P_{U_2U_1 \psi})  \bigr] ={\rm tr} \bigl[ P_0 {\rm tr} _{\mathcal{K}} P_{U_2U_1\psi} \bigr] .
\end{equation} 
This yields the probability $p_{U_2U_1\psi}^{P_0}(0) = \cos^2 \left( \frac{\theta}{2} \right)$, which depends explicitly on the phase $\theta$. Applying $\tau{_{\T_*}}$ at every stage does not alter the probabilities; we have, for example
\begin{align}
\tau{_{\T_*}} (P_{\psi}) &\to U_1 \tau{_{\T_*}} (P_{\psi})U_1 ^* = \tau{_{\T_*}} (U_1 P_{\psi} U_1 ^*)\\ \nonumber
& \to U_2\bigl(  \tau{_{\T_*}} (U_1 P_{\psi} U_1 ^*)\bigr) U_2 ^*  = \tau{_{\T_*}} (U_2 U_1 P_{\psi} U_1 ^* U_2 ^*) = \tau{_{\T_*}} (P_{U_2U_1 \psi}).
\end{align}
Then ${\rm tr} \bigl[P_0 \otimes \id\tau_{\T_*}(P_{U_2U_1 \psi})\bigr]$ coincides with the expression in \eqref{tls}.
The unitary maps $U_1$ followed by $U_2$ mimic what might occur in a realistic interference experiment in which the reference system is confined to a low dimensional Hilbert space. The interference fringes dictated by $\theta$ may be observed through the measurement of an invariant system-apparatus quantity. This does not require absolute coherence, i.e., does not imply the coherence of superpositions \emph{across} $N$-eigenspaces. It does, however, require mutual coherence.

Considering the states arising after application of $U_1$, it is immediately clear that the reduced states ${\rm tr}_{\mathcal{H}_R} [P_{U_1 \state{i}\state{j}}] $ and ${\rm tr}_{\his}{P_{U_1\state{i} \state{j}}}$ have no dependence on $\theta$, indicating that $\theta$ relates to
both $\Sy$ and $\R$. Since the post-$U_1$ states are entangled, the only means by which we may identify a system and a reference is via the partial trace.

We may define a restriction map 
$\Gamma_{\rho_{\R}}$ with $\rho_{\R}={\rm tr}_{\his}{P_{U_1\state{i} \state{j}}}$, which, since $\rho_{\R}$ is invariant, yields only invariant restricted quantities for $\Sy$ (if assumed invariant for $\Sy + \R$). The conclusion, then, is that $\theta$-dependent expectation values do not correspond to the observation of 
states with absolute coherence.

We now analyse a variety of infinite dimensional examples, and argue in subsection \ref{subsec:hpld} that we must
draw the same conclusion: relative phase factor sensitive measurement statistics can be achieved by measuring observables (and only observables), i.e., the relevant relative phase factors occur within an $N$-eigenspace. 
Only in the high reference phase localisation do these appear {\em as though} they pertain to
the system alone. 

\subsection{Model 2: Angular Momentum and Angle}\label{am}
We now adapt the previous model, replacing the space $\mathbb{C}^2$ of the reference system
with an infinite dimensional space, and construct a new unitary mapping (still calling it $U_1$ and restricting to the subspace spanned by $\{\ket{0},\ket{1}\}$ for the first system):
\begin{align}\label{eq:un1}
\state{0} \state{n} &\overset{U_1}\longrightarrow e^{-i \frac{\theta} {2}}\frac{1}{\sqrt{2}} \left(\state{0} \state{n} + e^{i \theta}\state{1} \state{n-1} \right),\\
 \state{1} \state{n-1} &\overset{U_1}\longrightarrow e^{-i \frac{\theta} {2}}\frac{1}{\sqrt{2}} \left(\state{0} \state{n} - e^{i \theta}\state{1} \state{n-1} \right) .\label{eq:un2}
\end{align}
Here the basis vectors are the eigenvectors of $N_i = \sum _{n=- \infty} ^{\infty}n ^{(i)} P_n ^{(i)}$. We observe that the partial trace over system or reference yields reduced states which do not depend on $\theta$.
Linearity and continuity entail
\begin{equation} \label{u1}
U_1: \Psi_0 \equiv \state{0} \state{\xi} \equiv \state{0} \sum_{n=-\infty}^\infty c_n \state{n} \longrightarrow  e^{-i \frac{\theta} {2}}\frac{1}{\sqrt{2}}\sum_{n=-\infty}^\infty c_n \left(\state{0}\state{n} + e^{i \theta} \state{1}\state{n-1}\right)\equiv\Psi_f.
\end{equation}
The initial state $\Psi_0$ under $\tau{_{\T_*}}$ takes the form (sums taken for $n$ running from $- \infty$ to $\infty$)
$$
\tau{_{\T_*}} (P_{\Psi_0}) = \sum_n P_n P_{\Psi_0} P_n = \state{0}\dstate{0} \sum_n \mods{c_n}\state{n}\dstate{n},
$$
 where the $P_n$ are the infinite-rank projectors onto the eigenspaces of $N = N_1 + N_2$ given as
\begin{equation}
P_n = \sum_{l+m=n}P^{(1)}_l \otimes P^{(2)}_m = \sum_{l} P^{(1)}_l \otimes P^{(2)}_{n-l}.
\end{equation}
We consider what observation may reveal about $\theta$ in the state $\tau_{\T *} (P_{\Psi_f})$ for 
$\Psi_f$ as given in \eqref{u1}. 
We have 
\begin{align}
\tau{_{\T_*}} (P_{\psi_f}) &=\sum_n \mods{c_n} \frac{1}{2}\Bigl\{ \state{0,n} \dstate{0,n} + \state{1,n-1} \dstate{1,n-1} \Bigr.
\nonumber\\ 
&\qquad\qquad +\Bigl. \state{0,n} \dstate{1,n-1} e^{-i \theta}  + \state{1,n-1} \dstate{0,n} e^{i \theta} \Bigr\}
\nonumber\\ 
&=\sum_n\mods{c_n}\,P_{\frac1{\sqrt2}\left(\state{0,n}+e^{i\theta}\state{1,n-1}\right)}.
\label{eq:70}
\end{align}
Here it is manifest that $\tau{_{\T_*}} (P_{\psi_f})$ is a mixture of states of different $N$-eigenvalues, and within each eigenspace labelled by $n$ there is a relative phase
factor between the states of the same $N$-eigenvalue. There exists an invariant quantity
of $\Sy + \R$ which is sensitive to $\theta$ in the state $\tau_{\T *} (P_{\Psi_f})$. For example, we may choose $A = \ket{0,n}\bra{1,n-1} + {\rm h.c.}$ and invoke relation \eqref{eq:spi} (replacing the spectral projections with the self-adjoint operators they define). 

We may extend the analysis and, in the spirit of the finite dimensional example, introduce a second unitary $U_2$ (which is independent of $\theta$), which with $U\equiv U_2 U_1$ yields on the number basis states
\begin{align}
\state{0} \state{n} &\overset{U}\longrightarrow \cos \left(\frac{\theta}{2}\right)\state{0} \state{n} - i \sin \left(\frac{\theta}{2} \right) \state{1} \state{n-1},\\
 \state{1} \state{n-1} &\overset{U}\longrightarrow -i \sin \left( \frac{\theta}{2} \right) \state{0} \state{n} + \cos \left (\frac{\theta}{2} \right)\state{1} \state{n-1} .
\end{align}

In analogy to the $2\times2$ case, we see that for example ${\rm tr} \bigl [\state{0}\dstate{0} {\rm tr} _{\mathcal{K}}[P_{U\state{0} \state{n}}] \bigr] = \cos^2 \left( \frac{\theta}{2} \right)$, and again,
since we have measured an observable (i.e., $\state{0}\dstate{0} \otimes \id$), applying $\tau{_{\T_*}}$ at all stages does not alter the result. The purpose of $U_2$ is then to bring the final states into a form wherein the measured observable is non-trivial 
only for $\Sy$ and measurement of an invariant quantity for $\Sy$ gives $\theta$-dependent expectation values. This validates the presence of mutual coherence, and does not indicate 
the existence of absolute coherence at any stage. As shall be shown, this becomes crucial in the question of whether superselection rules can be effectively overcome through a judicious choice of unitary mappings and measurements.

\subsection{Model 3: Number and Phase}

We now consider the number-phase case. Here, we have
$N_{\Sy} \otimes \id$ and $\id \otimes N_{\R}$ (each with spectrum given by $\mathbb{N}\cup\{0\}$) acting on 
$\mathcal{H}_S \otimes \mathcal{H}_R$ and $N=N_{\Sy} \otimes \id+\id\otimes N_{\R} = \sum_{n=0}^\infty n P_n$
with $P_n=\sum_{i+j=n}P_i ^{(1)} \otimes P_j ^{(2)}$.
A simple $N$-preserving unitary mapping is given by:
\begin{equation}\label{eq:evo1}
U_1: \state{0} \state{n} \to 
\left\{
\begin{array}{cl}
e^{-i \frac{\theta}{2}}\frac{1}{\sqrt{2}} \left( \state{0} \state{n}+ e^{i \theta}\state{1} \state{n-1}\right) & n > 0 \\
 \state{0} \state{0} & n=0
\end{array}
\right. 
\end{equation}
$$
U_1:\state{1} \state{n-1} \to e^{-i \frac{\theta}{2}}\tfrac{1}{\sqrt{2}} ( -e^{-i \theta}\state{0} \state{n}+ \state{1} \state{n-1})~n>0.
$$
Following the now familiar approach, we introduce a second unitary map $U_2$, under which $U \equiv U_2 U_1$ implements
\begin{equation}\label{eq:npd}
U: \state{0} \state{n} \to 
\left\{
\begin{array}{cl}
 \cos \left(\frac\theta2\right) \state{0} \state{n} - i \sin \left(\frac\theta2\right)\state{1} \state{n-1}) & n > 0 \\
 \state{0} \state{0} & n=0
\end{array}
\right. 
\end{equation}
$$
U:\state{1} \state{n-1} \to -i \sin \left(\tfrac\theta2\right)\state{0} \state{n}+ \cos\left(\tfrac\theta2\right) \state{1} \state{n-1}~n>0.
$$
Then ${\rm tr} \left[\state{0}\dstate{0} {\rm tr} _{\mathcal{K}} P_{U \state{0}\state{n}} \right] = \cos^2 \left(\frac{\theta}{2} \right)$
and once again we have a $\theta$-dependent probability distribution for an observable in the state $U\ket{0}\ket{n}$. Moreover, this does not differ from the distribution in the state
$\tau(P_{U\ket{0}\ket{n}})$. Of course, the $\theta$-dependence only corroborates the ``reality"
of the relative phase factor, within the eigenspace of $N$ with eigenvalue $n$ in the state on the top line of equation \eqref{eq:evo1}.

\subsection{Coherence and Mutual Coherence: Brief Discussion}
In each of the models we have discussed, the crucial component for witnessing interference
effects, in the form of $\theta$-dependent expectation values, is the presence of non-zero mutual coherence for states of $\Sy + \R$ (which is possible even in the absence of absolute coherence for $\Sy + \R$). Mutual coherence allows for (and is necessary for) the appearance of absolute coherence, even without a limit being taken.

For instance, $\varphi = \alpha \ket{0} + \beta \ket{1}$ gives, for $A = \ket{0}\bra{1} + \ket{1}\bra{0}$ the expectation value $2 Re (\alpha \bar{\beta})$. One can define the invariant (entangled) state $\tilde{\varphi} = \alpha \ket{01} + \beta \ket{10}$
and $\tilde{A} = \ket{01}\bra{10} + \ket{10}\bra{01}$ (which does not commute with $N_{\T}$)
so that $\ip{\tilde{\varphi}\tilde{A}}{\tilde{\varphi}} = \ip{\varphi}{ A \varphi} = 2 Re(\alpha \bar{\beta})$. This can also be done with an invariant $\tilde{A^{\prime}} = \tau_{\T}(\tilde{A})$. Thus, in this
case, asymmetric statistics of $\Sy$ can be given by symmetric ones of $\Sy + \R$ without the need for localisation. However, the physical interpretation is unclear due to the non-separability of $\Psi$. The important observation is that mutual coherence of $\tilde{\varphi}$  is required for the possibility of the appearance of absolute coherence of states of $\Sy$.

Next we examine the role of high reference phase localisation
in the interpretation of the measurement statistics.

\subsection{High Phase Localisation}\label{subsec:hpld}

We turn now to the behaviour of model 2 in the regime that the initial state of the reference system
is highly phase-localised. Let $c_n =\frac{ e^{in \theta'}}{\sqrt{2j+1}}$ for $-j \leq n \leq j$ and $0$ otherwise, and let $\state{\theta' _j} = \sum_{-j}^j c_n \state{n}$. This state is approximately
localised around the value $\theta'$, i.e., is an approximate eigenstate of the self-adjoint angle $\Theta_\R$ with eigenvalue $\theta'$, with the quality of approximation becoming increasingly good as $j$ becomes large. Indeed,
the sequence $(\state{\theta'_j})$ is an approximate eigenstate of $\Theta_{\R}$, in the sense that $\bigl\langle{\theta'_j}\big|{\Theta_{\R} \theta'_j}\bigr\rangle = \theta'$ and $\var (\Theta_\R)_{\theta' _j} \to 0$ as $j \to \infty$. (The sequence also describes a an approximately localised state in terms of concentration of probabilities, as described for the similar sequence $(\phi_n)$ of Example \ref{qinf}.)
Using the form of $U_1$ from section \ref{am}, we find
\begin{equation}\label{eq:err1}
\Psi_f \equiv U_1 \state{0} \state{\theta' _j} = e^{-i\frac{\theta}{2}}\frac{1}{\sqrt{2}}\left(\state{0}+ e^{i(\theta + \theta')} \state{1} \right)\state{\theta^{\prime} _j} + \state{\text{error}}_j 
\end{equation}
where the state 
\begin{equation}
\state{\text{error}}_j = e^{i\frac{\theta}{2}}e^{i \theta'}\frac{1}{\sqrt{2 (2j+1)}}\left( - e^{-ij\theta'}\state{1} \state{-j} +e^{i(j+1)\theta'}\state{1} \state{j+1}\right).
\end{equation}
 Clearly $\bigl\|{\state{\text{error}}_j}\bigr\|{^2}=(2j+1)^{-1}$ and therefore $\bigl\|{\state{\text{error}}_j}\bigr\| \to 0$ as $j \to \infty$. As this error term becomes arbitrarily small, $\Psi_f$ is arbitrarily norm--close (modulo an overall phase) to the product state $\bigl(\state{0}+ e^{i(\theta + \theta')} \state{1}\state{\theta' _j} \bigr)/\sqrt2$. 
 
Let $R \in \lht$ be invariant and self-adjoint. By continuity, $\lim_{j \to \infty}\no{R \state{\text{error}}_j} = 0$. Suppose we fix $\theta^{\prime}=0$ and $R = \Y(A)$ for some self-adjoint $A$ which does not commute with $N_{\Sy}$. Then, 

\begin{equation}
\lim_{j \to \infty} \ip{\Psi_f }{\Y(A) \Psi_f} = \ip{\varphi}{ A \varphi},
\end{equation}
with $\varphi:=\bigl(\state{0}+ e^{i\theta} \state{1}\state{\theta' _j} \bigr)/\sqrt2$.
Therefore, in this model, the expectation of the absolute quantity $A$ in the absolutely coherent state $\varphi$ approximates arbitrarily well the relational $\Y(A)$ in the invariant (absolutely incoherent) state $\tau_{\T *}(P[\Psi_f])$.

We may also consider the action of $U_2$, leading to an overall evolution $U$:

\begin{equation}\label{eq:err2}
U \ket{0}\ket{\theta ^{\prime} _j} = \left( \cos \left( \frac{\theta}{2} \right) \ket{0} -e^{i \theta ^{\prime}}i \sin \left( \frac{\theta}{2} \right) \ket{1}\right)  \ket{\theta_j ^{\prime}} + \ket{{\rm error}}_j,
\end{equation}
with
 \begin{equation}
\lim_{j \to \infty} \bigl\| \ket{{\rm error}}_j \bigr\| = \lim_{j \to \infty}\bigl\| \frac{1}{\sqrt{{2}(2j+1)}} \left( e^{i(j+1)\theta^{\prime}} \ket{1} \ket{j} -e^{ij \theta^{\prime}} \ket{1} \ket{-j-1} \right) \bigr\| =0,   
\end{equation}
leading to the evolution up to a term of arbitrarily small norm of 
\begin{equation}
 U \ket{0}\ket{\theta ^{\prime} _j} \approx \left( \cos \left( \frac{\theta}{2} \right) \ket{0} -ie^{i \theta ^\prime} \sin \left( \frac{\theta}{2} \right) \ket{1}\right)  \ket{\theta_j ^{\prime}}.
 \end{equation}
Therefore, if the error term $\state{{\rm error}}_j$, which can be made to have arbitrarily small
norm by choosing $j$ large enough is ignored, the state of the system alone is 
given by the first factor in the tensor product, achieved by partial tracing over $\R$. For simplicity
we set $\theta^{\prime} = 0$. Then measurement sensitivity of the observable $\ket{0}\bra{0}$ to 
$\theta$ (which is still present after operating with $\tau_{\T *}$) seems to validate the existence and measurability of (the relative phase factor $\theta$ in) the superposition $\frac{1}{\sqrt{2}}\left(\state{0}+ e^{i\theta} \state{1} \right)$, since the latter state is given as the state of the system again by ignoring the error term in equation \eqref{eq:err1}. It looks as though coherence across $N_1$ eigenspaces has been prepared 
and confirmed. We now critically analyse this conclusion.

\subsection{Interpretation} \label{subsec:int}
 
Analysis of the post-$U_1$ and post-$U$ states in the above high reference localisation regime
highlights several key points, variants of which will reappear throughout the rest of this paper under various guises. We first recapitulate:
\begin{itemize}
\item[---] Any reasonable measure of entanglement capable of capturing this situation would show
that the state $\Psi_f$ becomes arbitrarily close to an unentangled state for suitably large $j$.

\item[---] Continuity (of $R$) dictates that the statistics of absolute $A$ in absolutely coherent 
$\varphi$ can be approximated arbitrarily well, for suitably large $j$, by $\Y(A)$ in the state $\tau_{\T *}(P[\Psi_f])$. In particular, $\theta$-dependent expectation values are present before
the limit is taken.
\end{itemize}
The limit $j \to \infty$ itself must be treated with extreme caution---the rigorous
existence of such limits must be questioned, and the meaning of physical conclusions drawn
from the limit may not be clear. The main dangers of taking the large amplitude limits in the example we have discussed
are summarised below.
\begin{itemize}
\item[---] The limit $j \to \infty$ in the state $\ket{\theta ^{\prime}_j}$ does not yield a normalisable Hilbert space vector.

\item[---] $N_2$ (and thus $N$) is not a bounded/continuous operator and therefore 
$\no{N \ket{\rm{error}}_j}$ need not vanish even as $\no{\ket{\rm{error}}_j}$ does in the large $j$ limit.

\item[---] If the error term is ignored, the dynamics no longer conserve number (this is due to the unboundedness). This is most acutely observed by noting that in \eqref{eq:evo1}, $\theta$ may take any real value. Choosing $\theta= \pi$ 
and ignoring the error state, the evolution takes the form 
$U\ket{0}\ket{\theta^{\prime}}_j = \ket{1} \ket{\theta^{\prime}}_j$. It appears as though
the state $\ket{1}$ has been manufactured from $\ket{0}$ with no energy cost.

\item[---] Ignoring the error term leads to a ``reduced" unitary 
$U_{\rm eff} = U_{\Sy}\otimes \id$ and it is clear that $[U_{\rm eff},N] \neq 0$ and 
$[U_{\Sy},N_1] \neq 0$. Therefore, no matter how small 
$\no{\ket{{\rm error}}_j}$ may become, in order to properly account for energy/$N$ conservation, 
it must not be taken to be zero.

\item[---] The partial trace ${\rm tr}_{\R}[\tau_{\T *}(P[\Psi_f])]$, for any finite $j$, 
yields an invariant/absolutely incoherent state of $\Sy$ (by Proposition \ref{prop:pti}.)
Only in the limit does absolute coherence for $\Sy$ appear.

\end{itemize}

We now discuss the large amplitude limit in more detail.

\subsubsection{Meaning of the Limit}

In analysing the physical interpretation of the high amplitude limit, we will be guided
by two principles, referred to by Landsman \cite{LandBohr} as {\it Earman's principle} \cite{ear1} and 
{\it Butterfield's principle} \cite{butt1}. Earman's principle states that 
\begin{quote}
While idealisations are useful and, perhaps, even essential to progress in
physics, a sound principle of interpretation would seem to be that no effect can
be counted as a genuine physical effect if it disappears when the idealisations
are removed. (\cite[p.~191]{ear1}).
\end{quote}
Butterfield's principle then addresses the question of idealisations given by infinite limits, and describes in more detail the type of behaviour that must exist prior 
to a limit being taken:
\begin{quote}
There is a weaker, yet still vivid, novel and robust behaviour that occurs
before we get to the limit, i.e. for finite $N$. And it is this weaker behaviour
which is physically real. (\cite[p.~1065]{butt1}).
\end{quote}
Here, $N$ refers to particle number, but this principle is readily adapted to our situation. Taking this all into account, the following appears to be a consistent interpretation. 

The overall ($\Sy + \R$) dynamics are number-conserving,
and the observables which may be measured are invariant under phase shifts generated by $N$ (hence, commute with $N$). 
There is a reduced description, applicable to $\Sy$ on its own in which, in direct analogy to 
the discussion of high localisation in the kinematical case, absolute quantities, absolute coherence, and non-$N$-conserving dynamical maps approximately capture the observed statistics.

This reduced description is suitable in its convenience and usefulness in certain situations,
and provides an adequate tool for computing, to arbitrary approximation, empirically verifiable measurement statistics. For instance, the descriptions afforded by $A$ and $\Y(A)$ may be observed to be arbitrarily close given arbitrarily high reference phase localisation, with the limit
then featuring as an idealisation in which $A$ and $\Y(A)$ (or more correctly, $(\Gamma_{\phi} \circ \Y)(A)$) may be taken to be equal.

However, in addressing fundamental issues, the use of the idealisation (high localisation limit) betrays the essence of the phenomena under investigation. Guided by Earman's and Butterfield's principles, we may therefore discard as being artefacts of the idealisation those 
phenomena present in the limit but which disappear prior to the high localisation limit actually being taken. The attribution to $\Sy$ of
a state which is a superposition of eigenstates (and which is physically different from its corresponding mixture) of different $N_{\Sy}$ eigenvalue is such an example: taking the partial trace (over $\R$, with $\theta ^{\prime} = 0$) (in \eqref{eq:err1}) with the error term included (finite $j$) yields the state 
$\rho_{\Sy} = \frac{1}{2}(\ket{0}\bra{0}+ \ket{1}\bra{1})$, whereas ignoring the error term (infinite $j$) we find the state ${\rm tr}_{\R}\left[{P_{\Psi_f}}\right] = P_{\frac{1}{\sqrt{2}} (\state{0}+e^{i\theta } \ket{1})}$, i.e., the projection onto the vector unit vector 
$\varphi_{\Sy}=\frac{1}{\sqrt{2}} (\state{0}+e^{i\theta } \ket{1})$.\footnote{We may again analyse the states after the second stage of evolution. Writing $U=U_2U_1$ and considering the situation post-$U_2$,
ignoring the error term and partial tracing over $\R$ yields the state 
$\varphi ^{\prime}_{\Sy} = \cos \left( \frac{\theta}{2} \right) \ket{0} -i \sin \left( \frac{\theta}{2} \right) \ket{1}$, and therefore, in contrast to the post-$U_1$ situation, the observable $P_0=\ket{0}\bra{0} \otimes \id$
gives probabilities 
\begin{equation}
\ip{\varphi ^{\prime}_{\Sy} \otimes \theta^{\prime}_j}{P_0 \otimes \id \varphi ^{\prime}_{\Sy} \otimes\theta^{\prime}_j} = \ip{\varphi ^{\prime}_{\Sy}}{P_0\varphi ^{\prime}_{\Sy}}= \cos ^2\left( \frac{\theta}{2} \right)
\end{equation}
which are dependent on $\theta$.  There is, however, no reason to believe that such sensitivity to $\theta$ here
entails anything about observable relative phase factors between states of different number; we have already argued that $\theta$, properly interpreted, pertains to system and reference combined,
and only takes on the appearance of a relative phase factor between number states when the error
term is not taken into account---an illegitimate manoeuvre as far as the conservation law is concerned. We also note that the effective $U_2$ does not commute with $N_1$.
}

Another is the violation of energy conservation.  At any finite $j$, energy is manifestly conserved, whereas in the limit, with the error term ignored, energy conservation is violated. These two instances must therefore be viewed as pertaining to not physically real effects in the sense of Earman. The physically real effects are the statistics arising from the measurement of invariant quantities. Approximating these statistics in a convenient manner by non-invariant
quantities is legitimate, but attributing the measurement statistics to such quantities {\it as observables} is not. The measurement statistics containing $\theta$-dependent terms, close in approximation to the absolutely coherent superposition, are physically real, but the state description of $\Sy$ as absolutely coherent is not. 

Working with the idealised limit is legitimate when it comes to computing certain expectation
values which may arise in experiments. For example, using $A$ rather than $\Y(A)$ is unproblematic, provided the reference frame is prepared in a highly localised state. However, given the nature of our enterprise, that is, to understand the fundamental role played by symmetry  upon the definability and measurability of quantum mechanical quantities, it is illegitimate to move to an idealisation in which the symmetry is no longer manifest, {\it a fortiori} when the symmetry is present at every finite value of $j$ prior to the limit being taken. 

In other words, since we are interested in symmetry, we should not have recourse to a theoretical description in which, even though valid insofar as certain calculations are concerned, the symmetry 
in question is no longer present. Thinking of the description of a ball bouncing against a wall (cf. \cite{brs}), there is no problem, as far as the modelling of the ball is concerned, in taking the wall to be of infinite mass. But if one is performing an investigation of the limitations on dynamics imposed by momentum conservation, then taking the large mass limit of the wall---the limit in which momentum conservation is violated---cannot be viewed as fundamentally valid and completely obscures the issue at hand, namely the role played by symmetry and conservation.

\if It is crucial for the application of our framework that high phase localisation in the initial state of $\Sy + \R$ is precisely what is needed to make the error term in the final state arbitrarily small. When this happens, we can approximately separate system from reference, wherein we may apply all of the machinery of relativisation and restriction. The ``coincidence" that high localisation results in both approximate separability of the final states and the good approximation of relative quantities by absolute ones is worthy of a thorough separate investigation.
\fi

We now address controversies surrounding superselection rules and the reality of optical coherence, by critically analysing a number of models in the literature aimed at, in essence, obviating superselection rules. We will observe the use of dynamics and limits very similar to those
discussed above, with identical interpretation.

\if
Computing the partial trace over the reference system after the action of $U_1$, whilst ignoring the error term and setting $\theta ^{\prime} = 0$
yields ${\rm tr}_{\R}\left[{P_{\Psi_f}}\right] = P_{\frac{1}{\sqrt{2}} (\state{0}+e^{i\theta } \ket{1})}$, i.e., the projection onto the vector unit vector 
$\varphi_{\Sy}=\frac{1}{\sqrt{2}} (\state{0}+e^{i\theta } \ket{1})$.
It therefore looks \emph{as if} coherence has been created between states of different number at the level of $\Sy$, with relative phase factor $e^{i\theta}$. However, taking the same partial trace (over $\R$) with the error term included yields the state 
$\rho_{\Sy} = \frac{1}{2}(\ket{0}\bra{0}+ \ket{1}\bra{1})$. Therefore, the relative phase factor
remains a quantity relating $\Sy$ to $\R$. Also, no observable of $\Sy$ is sensitive to 
the relative phase factor at the level of $\Sy$ even if the error term is ignored.
\fi

\section{Controversies}\label{sec:cont}
The final part of this paper addresses a number of controversies which have appeared
in the literature over the last 65 years. The first relates to the fundamental status of superselection rules and the role played by reference frames there, and the second, appearing much later but strongly connected to the superselection rule debate, the question of the reality of optical coherence of laser beams.

We critique two opposing standpoints on the meaning and validity of superselection rules. Wick, Wightman and Wigner's (WWW's) seminal 1952 paper \cite{www} was met with objection from Aharonov and Susskind (AS) 15 years later \cite{as}, which was then obliquely criticised again by WWW. Subsequent efforts have been devoted on the one hand to rigorous work on superselection rules
in quantum field theory (see, e.g., \cite{Haag}), whilst on the other towards more practical questions on the role of superselection in information and communication theoretic tasks (e.g., \cite{brs, pres}).

After briefly introducing Wick, Wightman and Wigner's original argument, we focus on Aharonov and Susskind's contribution, highlighting points of agreement and disagreement between our perspective and theirs. For instance, the meaning of coherence/superpositions as requiring 
a relational understanding \cite{as} we view as ground-breaking, and this point of view has inspired much of the work in this paper. However, we do not support their conclusion (e.g., in the abstract of \cite{as}) that ``contrary to a widespread belief, interference may be possible between states with different charges"; nor do we agree that this conclusion follows from their argument. The paper suffers from mathematical flaws and a lack of conceptual clarity; what is at stake is nothing more than the appearance of measurability of absolute quantities/coherence in the presence of symmetry, and therefore the explicitly relational framework presented is well-suited to bringing a consistent and clear explanation of the issue of whether superselection rule ``forbidden" states
can be superposed to give a physically different state from its corresponding mixture.

We also critique more recent contributions \cite{brs, dbrs} along similar lines, focussing on the latter. The former \cite{brs} suffers from serious mathematical defects, some of which we have already pointed out and some of which are irreparable, which severely limit the conclusions that can be drawn from
the work. The scope of \cite{brs} is also catered heavily towards the role of reference frames in information-theoretic tasks and agent-based scenarios, e.g., entanglement theory, quantum key distribution, communication tasks, all when the given agents have no knowledge of each other's reference frame. The ensuing practical limitations gradually morph through the paper into fundamental ones, with far reaching conclusions that we contend are not warranted. We again give points of agreement (e.g., that ``all observable quantities ought to be relational") and disagreement (``superselection rules cannot provide any fundamental restrictions on quantum theory"), and again clear up dubious arguments by consistently applying
the principle that observable quantities are invariant. This also applies to \cite{dbrs}, which shares many mathematical problems with \cite{brs}. We find the language vague and occasionally conceptually unclear, and we will critique this work in detail, drawing upon ideas thus far presented.

\subsection{Brief Overview}

The notion of a superselection rule
was introduced by Wick, Wightman and Wigner \cite{www}, who proposed that superpositions
of states of bosons and fermions should be considered as equivalent to the associated mixture (i.e., that relative phase factors in superpositions of bosons and fermions are unobservable in principle), 
and a similar position was advocated for states of differing electric charge. Aharonov and Susskind  \cite{as} disagreed with the latter claim and offered a concrete experimental arrangement, very similar those we have considered in this paper (for the express purpose of critique), to demonstrate the possibility of preparing and 
observing coherent superpositions of states of different electric charge, via a formal analogy to the case of angular momentum. WWW then replied \cite{www70} with a theorem demonstrating that coherence is required 
in the initial state of one system in order to observe it in another, pointing to a circularity
in Aharonov and Susskind's argument and similar to the objection raised here in subsection \ref{subsec:sapo}.

Subsequently the issue of the ``reality" of quantum optical coherence was raised by M{\o}lmer (\cite{molmer}),
who suggested that if the gain medium of the laser is properly accounted for, the actual laser field is described by a mixture of number states, and that therefore the coherence
is merely ``convenient fiction".

Bartlett, Spekkens and Rudolph \cite{brs} (BRS), also in collaboration with Dowling
\cite{dbrs} (DBRS), have shed light on aspects of the superselection rule debate, particularly in clarifying the position of Aharonov and Susskind \cite{as}, and on the ``optical coherence controversy" \cite{dia}, highlighting the relative nature of states (and also, therefore, of coherence) and the accompanying role of reference frames.

We now present the form that these controversies take from the perspective of the relational formalism presented here. We believe that the framework we have developed for dealing with relative quantities
clarifies the seemingly opposing viewpoints of AS and WWW, and in a certain sense unifies them.
We will see that the attempts to overcome or ``lift" superselection rules (as they arise through the lack of a reference frame - see \cite{brs}) correspond to model considerations that take the same
form as the dynamical models already considered (many of which are modelled on \cite{as} and \cite{dbrs}). The framework afforded by observables-as-invariants allows for a circumvention of the ``relative" and ``global" decompositions 
of the system-apparatus Hilbert space described in \cite{brs, dbrs} (see also \cite{ldlthesis}) which have mathematical flaws, and allows for a direct assessment of the status of claims to, in essence, obviate superselection rules. 

\subsection{The Exchange between Aharonov-Susskind and Wick-Wightman-Wigner}

\subsubsection{Wick, Wightman, Wigner: The First Superselection Rule}

In 1952, Wick, Wightman and Wigner \cite{www} made a simple argument to demonstrate the existence of a dichotomy
between the assumption that all self-adjoint operators represent observables on one hand (a working assumption since von Neumann's book \cite[p.~313]{vonN}), 
and relativistic invariance on the other. Since double time reversal, 
$T^2:\his \to {\his}$ (with $\his \equiv \his{_b} \oplus \his{_f}$, the decomposition into bosonic and fermionic subspaces defining the projections $P_b$ and $P_f$ respectively), they argue, cannot be observed, and since $T^2$ has the effect of leaving bosonic
states invariant and introducing a minus sign on fermionic states:
\begin{equation}
\Psi_+ \equiv \frac{1}{\sqrt{2}}\left(\varphi_b + \varphi_f \right) 
\overset{T^2}\longrightarrow \left( \varphi_{b} - \varphi_f \right) \equiv \Psi_- ,
\end{equation}
it must be that any observable leaves the bosonic and fermionic sectors invariant, with the sign difference then unobservable. This follows
since for any self-adjoint $A$, for the consequences of a double time reversal to be unobservable,
it must be that $\ip{\Psi_+}{A \Psi_+}=\ip{\Psi_-}{A \Psi_-}$, from which it follows that 
any observable $A$ must commute with $P_b$
and $P_f$ and thus any observable $W$ with $P_b$ and $P_f$
as spectral projections. $W$ is then a \emph{superselection observable}. 

We observe that the stipulation that observables of $\Sy$ commute
with $W$ leads to the equivalence of states $\rho$ and $\tau_{\Sy *} (\rho)$, with
$\tau_{\Sy *}(\rho) := P_b\rho P_b + P_f \rho P_f$ in this case. WWW also conjectured (subsequently proven in quantum field theory by Strocchi and Wightman \cite{strowi}) that the relative phase factors
in superpositions of states of different electric charge have the same status, namely
cannot be determined by experiment, even in principle, and therefore that the states
$\rho$ and $\sum P_n \rho P_n \equiv \tau_{\Sy *}(\rho)$ are equivalent, with the sum running over all possible 
values of electric charge. Any observable must commute with charge, and must thus be invariant under shifts in phase/angle conjugate to charge.

The stipulation of such a superselection rule is formally identical to the limitation imposed by 
the a priori assumption that observables are invariant under symmetry (phase shifts in the charge case). Therefore, it must be understood whether the statement of a (say, electric charge) superselection rule amounts to anything more than the restriction thus far discussed. First, we discuss the reply of AS to
the WWW paper, along with another model (due to Dowling et al., \cite{dbrs}) purporting to 
prepare and measure (absolutely) coherent superpositions of atoms and molecules (against baryon number superselection).

As has been shown, the requirement that observables be phase shift invariant allows 
for the relative phase of system and reference to be observed, with the absolute phase 
representing the relative phase given high reference phase localisation. The reply
by Aharonov and Susskind to the WWW paper advocating the possibility of measuring relative phase factors in charge superpositions paper makes explicit use of such phase references. We now review their reply and hope to clarify their position by employing the methods and language introduced in this paper.

\subsubsection{Reply of Aharonov and Susskind: Proton-Neutron Superpositions}

In favour of observability of relative phase factors between superselection rule ``forbidden" superpositions, we sketch two thought experiments; the first is due to Aharonov and Susskind \cite{as}, conceived so as to demonstrate a realistic scenario in which coherent superpositions of states of different electric charge can be prepared and measured. The second, due to Dowling {\it et al.} \cite{dbrs}, is similar in spirit, and purports to prove that atoms and (diatomic) molecules
can be (absolutely) coherently superposed.\footnote{At least, it is claimed at one point, that ``The experiment we present aims to exhibit quantum coherence between states corresponding
to a single atom and a diatomic molecule..."}

It will be shown that in both of these examples there is an implicit relativisation of the operators to be measured, thereby constructing an invariant operator (observable) not unlike the ones we have discussed. Furthermore, in both cases a crucial role is played by the limit of high localisation of a reference state (in both cases provided by a coherent state)
with respect to a covariant phase-like operator conjugate to the symmetry generator, in direct  analogy
to the models and general results that have been presented. To our knowledge, it has not been explicitly stated anywhere that such localisation is the key property.

The Hilbert space $\his \otimes \hi _{\R _1} \otimes \hi_{\R_2}$ of Aharonov and Susskind's thought experiment is to correspond to a proton-neutron system $\Sy$ and two cavities ($\R_1$, $\R_2$) capable of containing any integer number of negatively charged mesons.
Aharonov and Susskind imagine preparing $\R_1$ and $\R_2$ in charge-coherent states 
(we include normalisation factors that were omitted in the original treatment)
\begin{equation}\
|q_1, \theta \rangle = e^{-q_1/2}\sum_n \frac{q_1 ^{n/2}}{\sqrt{n!}}\exp{(i n \theta)} |n\rangle
 \equiv \sum c_n (\theta) \ket{n}
\end{equation}
and 
\begin{equation}
|q_2, \theta^{\prime} \rangle = e^{-q_2/2}\sum_n \frac{q_2 ^{n/2}}{\sqrt{n!}}\exp{(i n  \theta^{\prime})}|n \rangle
 \equiv  \sum c'_n(\theta ^{\prime}) \ket{n}
\end{equation}
 respectively, where $\ket{n}$ denotes a charge eigenstate corresponding to $n$ negatively charged mesons. The parameters $q_1$ and $q_2$  represent the respective mean charge values in the coherent states, corresponding to the observables $Q_1, \ Q_2$, 
which are structurally identical to the number operators we have encountered thus far, except that $n$ takes (only) non-positive values.

The initial state of the nucleon is a proton $\ket{P}$, and we will use $\ket{N}$ to represent a neutron.  $\ket{P}$ and $\ket{N}$ are thus  eigenstates of the charge observable $Q_{\Sy}$  of $\Sy$
with eigenvalues $1$, and $0$, respectively. The dynamics, which take place in two stages, are governed by a Jaynes-Cummings-type Hamiltonian (which commutes with charge) $H=g(t)(\sigma ^+ a^- + \sigma ^- a^+)$ where
$\sigma^{+} = |N \rangle \langle P| $, $\sigma^{-} = |P \rangle \langle N|$ (sometimes referred to as the \emph{isospin} operators), and $a^{\pm}$ are meson creation and annihilation operators which act on the states of the cavities. The function $g(t)$ describes the interaction strength and fixes the duration of the interaction (given physically by the passage time of the nucleon travelling through the cavity). Explicitly, the dynamics are governed by $H_1 = g_1(t)(\sigma ^+ \otimes a^- \otimes \id + \sigma ^-  \otimes a^+ \otimes \id )$ with $g_1(t)=g\chi_{[0,T]}(t)$, followed by $H_2 = g_2(t)(\sigma ^+ \otimes \id \otimes a^- + \sigma ^-  \otimes \id \otimes a^+)$ with $g_2(t)=g\chi_{[T,2T]}(t)$.
The unitary $U_1$ effects the following transitions on charge eigenstates (omitting the second cavity):

\begin{align}
|N\rangle |n\rangle &\longrightarrow i \sin{\left(Tg\sqrt{n}\right)}|P\rangle |n - 1 \rangle +   \cos{\left(Tg\sqrt{n}\right)}|N\rangle |n\rangle,\\
\ket{P}\ket{n} &\longrightarrow  \cos{\left(g\sqrt{n+1}\right)} |P \rangle |n\rangle + i \sin{\left(Tg\sqrt{n+1}\right)}|N \rangle |n+1 \rangle.
\end{align}
Referring back to equation \eqref{eq:npd}, these are of an almost identical form. Analogous to what we saw there, 
we find here that we may measure the observable $\ket{P}\bra{P}\otimes \id$ in the state $U_1 \ket{P}\ket{n}$ to find the proton probability $\cos^2{\left(Tg\sqrt{n+1}\right)}$. 

Starting with the initial state $\Psi_0=\ket{P}\ket{q_1,\theta}\ket{q_2,\theta'}$, the state after the first cavity is
\begin{equation} \label{eqas}
U_1 \Psi_0 =\sum_n c_n\left[ \cos{\left(Tg\sqrt{n+1}\right)}|P\rangle |n\rangle  +  i \sin{\left(Tg\sqrt{n+1}\right)}|N\rangle |n+1 \rangle \right]\,\ket{q_2,\theta'}.
\end{equation}
One must then consider the limit of large $q_1$, which yields 
\begin{equation}
U_1 \Psi_0 \approx \left(i e^{i \theta}\sin{\left(gT\sqrt{q_1}\right)}\ket{N} + \cos{\left(gT\sqrt{q_1}\right)}\ket{P} \right)\, \ket{q_1, \theta} \, \ket{q_2, \theta ^{\prime}}.
\end{equation}

The nucleon is then approximately ``separated" from the cavities; it enters the second cavity and exits, this time in the large $q_2$ limit, as 
\begin{align}
& \Big[  \left( \cos{(gT\sqrt{q_1})}\cos{(gT\sqrt{q_2}}- e^{i(\theta - \theta^{\prime})} 
\sin{(gT\sqrt{q_1})}\sin{(gT\sqrt{q_2})}  \right) \ket{P}\\ 
& + i\left( e^{i \theta^{\prime}} \cos{(gT\sqrt{q_1})}\sin{(gT\sqrt{q_2})} + 
 e^{i \theta}\sin{(gT\sqrt{q_1})} \cos{(gT\sqrt{q_2})} \right ) \ket{N} \Big] \, \ket{q_1, \theta} \, 
\ket{q_2, \theta ^{\prime}}.\nonumber
\end{align}
As observed, the proton probability (i.e., $\tr{\ket{0}\bra{0}U \Psi_0}$) now depends on $\theta - \theta^{\prime}$, the \emph{relative phase} between $\R_1$ and $\R_2$. 

Therefore, as argued by Aharonov and Susskind, the nucleon is in a coherent superposition of proton and neutron
with relative phase $(\theta - \theta ^{\prime})$ ``\emph{when referred to the frame provided b}y $\R _2$". The idea is that the absolutely coherent superposition is created by the first cavity
(cf. \eqref{eqas}) and then confirmed by measuring an invariant quantity of $\Sy$ after passage through the second cavity. However, the model presented by AS suffers from the same kind of difficulties as discussed in subsection \ref{subsec:int}.

From the perspective developed in the present paper, we would instead say that in the limit of high reference system localisation,
we are faced with the \emph{appearance} of measuring an absolute quantity (namely a phase-like quantity sensitive to
relative phase in nucleon superpositions) in an absolutely coherent state, but this is appropriately understood as pertaining to a relative phase-like observable
between the nucleon and the cavities (and a mutually coherent state). The analogy to the angular momentum/angle case, as employed by Aharonov and Susskind to compel one to believe in the observability of proton-neutron superpositions, is indeed a good one. However, we argue for the opposite conclusion: it is \emph{not} that since absolute coherence of states of different angular momentum is observable, therefore so is the relative phase factor in superpositions of charge states, but rather, absolute coherence for angular momentum is not possible, and nor is it in the charge case.

Indeed, it is stated quite explicitly in \cite{as}, that ``the coherence of states of different angular momentum is measured relative to a frame of reference". Thus coherence itself is viewed as a relative feature; from this point of view, there is no absolute coherence of states of angular
momentum. Once again, we see the importance of the mutual coherence concept. 

The Aharonov-Susskind paper was understood by many as proving the possibility of coherent superpositions of states
of different electric charge---a situation conjectured impossible by Wick, Wightman and Wigner (WWW) \cite{www} 15 years previously. Three years after Aharonov and Susskind's contribution, WWW demonstrated the \emph{necessity} of using
superpositions of states of different charge (i.e., the absolutely coherent  cavity states) in order to demonstrate their existence (i.e., for the nucleon); see subsection \ref{subsec:WWW}. Aharonov and Susskind were alert to such a circularity and tried to avoid any possible objection in the final part of their paper, where they attempted to construct a charge eigenstate out of the two charge coherent states, providing a manifestly (phase-shift) invariant state.
This takes the form of an integral (\cite{as}, final page)
\begin{equation}\label{ASfinalpage}
\ket{i} = \int \ket{q \theta _1} \ket{q ^{\prime} \theta _2} \delta \left(\theta _1 - \theta_2 -(\theta ^{\prime} - \theta)\right)e^{-i(q + q^{\prime}) \theta _1}d \theta _1 d \theta_2,
\end{equation}
where the initial state $\ket{i}$ is then 
an eigenstate of charge $q+q'$ and fairly well-defined phase $\theta ^{\prime} - \theta$.\footnote{The numbers $q$ and $q^{\prime}$ pertain to the amplitude of the coherent states $\ket{q, \theta_1}$ and $\ket{q^{\prime}, \theta_2}$ and, as such, are continuous. However, for the expression in equation \eqref{ASfinalpage} to represent an eigenstate of the total charge, $q$ and $q^{\prime}$ must be restricted to taking integer values.}
They claim that the proton probability distribution is unchanged even when the cavities are prepared in a charge eigenstate. 
The following calculation demonstrates that their proposal is flawed: if the  two-cavity system is prepared in a charge eigenstate, $\ket{i}$, under charge-conserving evolution the approximation they give can never be valid. Suppose under evolution $U$ which conserves total charge we have 
$$
\ket{P}\otimes \ket{i} \stackrel{U}{\longrightarrow}\phi_{i+1}
$$ 
where $Q^c\phi_{i+1}=(i+1) \phi_{i+1}$. Then for
arbitrary $\psi = \alpha \ket{P} + \gamma \ket{N}$, $\left \vert \alpha \right \vert ^2+\left \vert \gamma \right \vert ^2 = 1$, we note the following trivial observation: 
$$
\no{\phi_{i+1} - \psi \otimes \ket{j}} = 0\quad\text{if and only if $\gamma = 0$ and $i=j$}.
$$
 By contrast, in the example where $\alpha = \gamma = 1/ \sqrt{2}$, we have  $\nos{\phi_{i+1} - \psi \otimes \ket{i}} \geq 2-\sqrt{2}$. Thus the resulting state is a finite (norm) distance from an eigenstate, independent of the ``size'' of the reference system. This ``fix" by Aharonov and Susskind is therefore untenable, and their conclusion that interference effects may be observed 
between states of different electric charge, given the restriction of not assuming its possibility from the outset, does not follow from their argument.
 
The approximation based on high amplitude coherent states was mathematically valid and results in states close to a product state containing proton--neutron superpositions in the system Hilbert space. High amplitude coherent states, however, already exhibit absolute coherence if the 
observables are not restricted to invariants (such a constraint on observables is barely mentioned in AS's paper.) The above result demonstrates that if the coherent states are replaced with a charge eigenstate, no such approximation can occur. WWW responded to the AS paper, also implicitly criticising the error, which we now discuss.

\subsubsection{Response of Wick, Wightman, Wigner}\label{subsec:WWW}

Wick, Wightman and Wigner \cite{www70} responded to Aharonov and Susskind's challenge to the superselection
rule for charge, making three key points which we now summarise. We note that WWW's argument was not to offer a proof of charge superselection, but rather to argue that superpositions of states of different charge cannot arise from (composition of) invariant states, charge-conserving dynamics, and subsystem separation. This therefore take place 
in the Schr\"{o}dinger picture.

 We assume that charge
($Q_{\Sy}$ for $\Sy$ and $Q_{\R}$ for $\R$) may take positive and negative values, and recall that $\tau_{\T *}(\rho) = \sum_{-\infty}^{\infty}P_n \rho P_n$ (with appropriate indices for subsystems, $\Sy$ and $\R$). 
\begin{enumerate}
\item The composition $\rho_{\Sy} \otimes \rho_{\R}$ for $\rho_{\Sy}$ and $\rho_{\R}$ invariant yields a state which commutes with total charge (i.e., is invariant) and no interference of states of different charge of $\Sy + \R$ is possible.
\item The time evolution $U:\his \otimes \hir \to \his \otimes \hir$, 
which commutes with $Q=Q_{\Sy} + Q_{\R}$, gives $[Q,U\rho U^*]=0$ for any $\rho$ for which $[\rho,Q]=0$. Equivalently, with $\mathcal{U}(\cdot) = U (\cdot) U^*$, 
$\mathcal{U}\circ \tau = \tau \circ \mathcal{U}$.
\item Given $\tau_{\T *} (\rho)$, the reduced states of $\Sy$ and $\R$ commute with $Q_{\Sy}$ and $Q_{\R}$,
respectively. Equivalently, $\tau_{\Sy *}({\rm tr}_{\R}\left[\tau_{\T *} (\rho)\right]) = {\rm tr}_{\R}\left[\tau_{\T *} (\rho)\right]$.
\end{enumerate}
The three steps outlined above correspond to composition, evolution and separation, respectively.
Regarding observing absolutely coherent superpositions of states of different charge of $\Sy$, WWW
showed that it ``takes one to know one"; specifically, a coherent superposition of states 
of different charge for $\R$ are required in order to observe them at the level of $\Sy$, showing that AS's argument, as presented, is circular (in using coherent states for the cavities) or flawed (in using a charge eigenstate for the combined cavities).

Dowling, Bartlett, Spekkens and Rudolph presented, in 2006 \cite{dbrs}, an argument in favour of superpositions 
of states of different baryon number, correcting some flaws in Aharonov-Susskind's argument. We now present this model, before comparing the viewpoints of the two ``camps" (those who believe superselection can be obviated in practice, and those who don't), and discussing the wider context of superselection rules and their obviation.

\subsection{Atom-Molecule Superpositions according to Dowling \textit{et al.}}\label{dowling}

In the spirit of the 1967 contribution by Aharonov and Susskind, Dowling \textit{et al.} \cite{dbrs} attempt to model the observation of a coherent superposition of an atom and a (diatomic) molecule, as a possible demonstration of coherent superpositions of states of differing baryon number. In order to avoid the error of Aharonov and
Susskind in preparing the cavities in an eigenstate of the conserved quantity, they instead utilise
the coherent state, but acknowledge that appropriate ``sectorising'' (i.e., application of the $\tau_{\T *}$/twirling map) is necessary in order to
respect the symmetry for the composite system. 

The reference system is provided
by a Bose--Einstein condensate (BEC), coherent states of which are written $\state{\beta} = \sum_{n=0} ^{\infty} c_n \state{n}$ ($\state{n}$ representing a state of $n$ atoms) with $c_n = \exp{(-\mods{\beta}/2}) \beta ^n / \sqrt{n!}$. We write $\beta = \sqrt{m} e^{i \theta}$, and have that $\langle N \rangle_{\beta} = \mods{\beta} = m$ and $(\Delta N )_{\beta} = \sqrt{m}$, and
as $m$ becomes large, coherent states become arbitrarily highly localised in phase. Therefore the coherent
state looks increasingly like a phase ``eigenstate''.
It is also useful to note that $\tau_{\R *} (P_{\beta}) = \sum_{n=0}^{\infty} P_n \state{\beta} \dstate{\beta}P_n = \sum_{n=0} ^{\infty} \mods{c_n} \state{n} \dstate{n}$. 

Dowling et al. describe an experiment, again with a multistage unitary
along the lines of the models we have outlined, which goes as follows: The initial state is $P_{\state{A} \otimes \state{\beta}}$ ($\sim \state{A}\dstate{A} \otimes \tau_{\R*} (P_{\beta})$), where the state $\state{A}$ is to represent an atom; accordingly molecule states are written $\state{M}$ (both of these are to be understood as shorthand: $\state{A} \equiv \state{0}_M\state{1}_A$ and $\state{M} \equiv \state{1}_M\state{0}_A$). 
Defining the cavity states
\begin{equation}
| \beta _{A}^1 \rangle = \sum_{n=0}^\infty c_n\cos{\left(\frac{\pi}{4}\sqrt{\frac{n}{m}}\right)} |n\rangle  = \sum_{n=0}^{\infty} \frac{e^{-m/2}m^{n/2}}{\sqrt{n!}}e^{in \theta}\cos{\left(\frac{\pi}{4}\sqrt{\frac{n}{m}}\right)} |n \rangle
\end{equation}
and 
\begin{equation}
\state{\beta_{M} ^1} = -i \sum_{n=0}^\infty c_n \sin \left( \frac{\pi}{4} \sqrt{\frac{n}{m}} \right) \state{n-1},
\end{equation}
they give the following sequence of unitary maps (for details on the specific form of the Hamiltonians, see \cite{dbrs}):
\begin{equation}\label{d1}
\Psi ^{\prime} \equiv U_1 |A \rangle \otimes \state{\beta} = \state{A} \otimes \state{\beta _{A}^1} + \state{M} \otimes \state {\beta _{M}^1}
\end{equation}
followed by free evolution under a Hamiltonian of the form $K \state{M} \dstate{M}$ (with $K$ a constant)
\begin{equation}
\Psi ^{\prime} \to \Psi ^{\prime \prime} \equiv U_2 \Psi ^{\prime} = \state{A} \otimes \state{\beta _{A}^1} + e^{i \phi}\state{M} \otimes \state {\beta _{M}^1},
\end{equation}
where $\phi = T K$ and $T$ is the duration of free evolution. Thus $U_2$ explicitly depends on $\phi$. 
Finally,
\begin{equation}
U_3 \Psi^{\prime \prime} = \state{A} \otimes \state{\beta _{A}^3} + \state{M} \otimes \state {\beta _{M}^3},
\end{equation}
with 
\[
\state{\beta _{A}^3} = \sin\left(\frac{\phi}2\right) \state{\beta} - i \cos \left(\frac{\phi}2\right)\sum c_n \cos \sqrt{\frac{n}{m}\frac{\pi}{2}} \state{n}
\]
 and 
\[
\state{\beta _{M}^3}=- \cos \left(\frac{\phi}2\right) \sum c_n \sin\sqrt{\frac{n}{m}\frac{\pi}{2}} \state{n-1}
\]
again representing cavity states.
The purpose of $U_2U_1$ is to introduce the relative phase factor $\phi$; $U_3$ then allows a measurement of a convenient quantity (i.e. $\state{M}\dstate{M}, \state{A}\dstate{A}$) for realistic experiments, but also to
measure an invariant quantity of $\Sy$.
For the purposes of discussing relative phase factor observability it is sufficient to consider the state following the application of $U_1$ or $U_2$, along with the asymptotic behaviour outlined in \cite{dbrs}.

Since discussions pertaining to the type of convergence thus far encountered here have been somewhat informal in the existing
work, we provide a proof in the appendix that, for example, 
\begin{equation}\label{eq:lim-m}
\bigl\Vert \state{\beta _{A}^1} - \frac{1}{\sqrt{2}} \left\vert \beta \right\rangle \bigr\Vert  \to 0~ \text{as}~ m \to \infty .
\end{equation}
We may write
\begin{equation}\label{d1}
U_1 |A \rangle \otimes \state{\beta} = \left (\tfrac{1}{\sqrt{2}}\state{A}   - ie^{i \theta}\tfrac{1}{\sqrt{2}} \state{M} \right) \otimes \state {\beta}+ \state{\text{error}}_m
\end{equation}
where
\begin{equation}
 \state{\text{error}}_m = \state{A} \otimes \left(\tfrac{1}{\sqrt{2}} \state{ \beta} - \state{\beta _{A}^1}\right) + \state{M} \otimes \left(i e^{i \theta}\tfrac{1}{\sqrt{2}} \state{\beta} + \state{\beta _{M}^1}\right)
\end{equation}
with $\theta \equiv \arg{\beta}$. It is clear that $\left \Vert \state{\text{error}}_m \right \Vert \to 0$ as $m\to\infty$ if and only if 
\[
\left \Vert\tfrac{1}{\sqrt{2}} \state{ \beta} - \state{\beta _{A}^1}\right \Vert  \to 0
\quad \text{and}\quad
\left \Vert i e^{i \theta}\tfrac{1}{\sqrt{2}} \state{\beta} + \state{\beta _{M}^1} \right \Vert  \to 0
\]
 individually, using the fact that $\ip{A}{M}=0$ and $\bigl \Vert \state{A} \bigr \Vert = 
\bigl \Vert \state{M} \bigr \Vert = 1$. 

However, one can also consider the post $U_3$ state;
again, asymptotically and ignoring the error term we have (as given in \cite{dbrs})
\begin{equation}\label{d2}
U_3U_2U_1 \state{A} \otimes \state{\beta} \cong \left[ \sin \left( \frac{\phi}{2} \right)\state{A}  - e^{i \theta} \cos \left(\frac{\phi}{2}\right) \state{M} \right] \otimes \state{\beta}.
\end{equation}

The interpretation given in \cite{dbrs} is that since one can apply $\tau_{\T *}$ at every stage (under the approximation) and still
achieve atom/molecule probabilities of $\sin ^2 (\phi/2)$ and $\cos ^2 (\phi/2)$ respectively,
a coherent superposition of an atom and a molecule has been observed.\footnote{E.g., ``...it is possible, in principle, to perform a Ramsey-type interference experiment
to exhibit a coherent superposition of a single atom and a diatomic molecule". We note, however, that DBRS do acknowledge that in states like that appearing in Eq. \eqref{d1} with the error term taken to be zero, one would be inclined to again ``twirl" ($\tau_{\Sy *}$) the resulting $\Sy$-state $\frac{1}{\sqrt{2}}\left (\state{A}   - ie^{i \theta} \state{M} \right)$, going ``full-circle" and returning from whence we came: to an equivalence between coherent and incoherent descriptions. From this ensues a discussion of alternative tensor product decompositions of the system-reference Hilbert space---relative and global---as a way of interpreting the new state. This procedure cannot work in general (or even in the case they give), and is unnecessary anyway.} 

In view of the work we have presented, along with the argument of WWW \cite{www70}, we do not agree with this view. Given the problems with taking the limit (violation of the conservation law, non-existence of limit for states, unphysical nature of such a limit), we believe that the limit should not be taken in considering the fundamental status of these experiments. As such, the analysis of WWW holds, and absolute coherence cannot be observed for atom-molecule ``superpositions". What is instead observed is mutual coherence, and the observability of the interference effects as given
by (for example) $\sin ^2 (\phi/2)$ only demonstrates the feasibility of measuring relative phase factors
within a sector, and the phase $\phi /2$ should be viewed as precisely this. The large reference system,
which provides high reference phase localisation, again provides the appearance of a relative phase factor
at the level of the system only. 

Therefore, we return once more to the main point: absolute quantities are not measurable, but represent measurable, relative quantities, with good approximation coming with good localisation (suitable, relationally, interpreted). We conclude this section with a final analysis of the two views concerning the observability of ``forbidden" superpositions.

\subsubsection{Analysis of the Opposing Standpoints}

Following, for example, the prescription given in section \ref{am}, it is possible to follow WWW's three-step
sequence to the letter:
\begin{enumerate}
\item Compose: $\ket{\Psi_0}\bra{\Psi_0} = \ket{0}\bra{0}\otimes \sum |c_n|^2 \ket{n}\bra{n}$;
\item Evolve: $\ket{\Psi_0}\bra{\Psi_0}$ evolves according to the charge-conserving unitary defined in \eqref{eq:un1} and \eqref{eq:un2} yielding $\tau_{\T*}(\ket{\Psi_f}\bra{\Psi_f}) =\sum_n\mods{c_n}\,P_{\frac1{\sqrt2}\left(\state{0,n}+e^{i\theta}\state{1,n-1}\right)}$ (eqn.\eqref{eq:70});
\item Separate: ${\rm tr}_{\R}[\tau(\ket{\Psi_f}\bra{\Psi_f})] = 1/2 \id$ (on the two-dimensional subspace spanned by $\{\ket{0}, \ket{1}$).
\end{enumerate}
On this basis, it is clear that there can never be interference observed
between $\ket{0}$ and $\ket{1}$ under the processes outlined by WWW.

On the other hand, as described in subsection \ref{subsec:hpld}, we may prepare
the state $\ket{0}\bra{0}\otimes \tau_{\R *}P[\Psi_0]$, with $\Psi_0 = \sum_{n}c_n \ket{n}$, choosing $c_n = \frac{e^{in\theta ^{\prime}}}{\sqrt{2j+1}}$ for $|n|\leq j$ and $0$ otherwise.
Then, for finite $j$, there exists invariant $A \in \lht$ so that $\tr{A \tau_{\T *}P[\Psi_f]}$ depends on $\theta$. This $\Psi_f$, as $j$ becomes arbitrarily large, becomes arbitrarily close to the product state
\begin{equation}\label{eq:finsr}
\frac{1}{\sqrt{2}}\bigl(\state{0}+ e^{i(\theta + \theta')} \state{1}\bigr) \state{\theta' _j} .
\end{equation}
Then employing relation \eqref{eq:spi}, the statistics of an invariant quantity in $\tau_{\T*}(P_{\Psi_f})$ are identical to the statistics in $\Psi_f$. One finds that, for example,
$\ip{\Psi_f}{(\Theta - \Theta_{\R}) \Psi_f}$
gives rise to statistics which are sensitive to the relative phase $e^{i(\theta + \theta^{\prime})}$. With $\theta ^{\prime} = 0$, one finds that 
$\ip{\Psi_f}{(\Theta - \Theta_{\R}) \Psi_f} =\ip{\varphi_{\ell}}{\Theta \varphi _{\ell}}$
with $\varphi_{\ell}:=\bigl(\ket{0} + e^{i \theta}\ket{1}\bigr)/\sqrt{2}$. Thus it appears as though one has measured an absolute observable in a superposition state.

In order to attempt to avoid the appearance of measuring an absolute quantity, the second unitary (e.g., that introduced in \eqref{eq:npd}) allows, on the system level
and ``once the limit has been taken", for something like this to occur:
\begin{equation}
\frac{1}{\sqrt{2}}\bigl(\state{0}+ e^{i\theta}\ket{1}\bigl) \mapsto \cos{\left(\frac{\theta}{2}\right)}\ket{0}-i\sin{\left(\frac{\theta}{2}\right)}\ket{1}.
\end{equation}
Then the observable (e.g.) $\ket{0}\bra{0}$ can be measured and a $\theta$-dependent probability
distribution achieved.

The upshot is that both WWW and AS/DBRS make arguments which bear out (once the errors have been remedied). The former show, quite correctly, that strictly speaking, only (absolute) coherence begets (absolute) coherence, and
if you don't have it, you'll never get it, as one would expect. The latter ``camp", in their attempt to show the positive possibility of creating absolute coherence from states without it,
actually show the possibility of well-approximating absolute quantities and states with absolute coherence by relative quantities and states without absolute coherence. The crucial ingredients for such an approximation are mutual coherence and high localisation.

\subsection{Further analysis: Superselection Reconsidered}

Bartlett, Spekkens and Rudolph \cite{brs} argue a superselection rule may be ``lifted", that is (we think), the following holds: a superselection rule applies to some system $\Sy$. A reference frame $\R$ may be included, the superselection rule applied to $\Sy + \R$, whose statistics then exactly give those of 
$\Sy$ as if there weren't a superselection rule for $\Sy$. This is taken as proof that ``superselection rules cannot provide any fundamental restrictions on quantum theory"  since, they argue, a SSR is simply a lack of an appropriate frame, which can always be introduced.\footnote{Such a view appears to be favoured also by Lubkin, 1970 \cite{lub1}.} 

We do not endorse this view. First, the analysis preceding the above quote in \cite{brs} is mathematically flawed. Second, the reason given for the (e.g.,photon number) superselection rule
is a practical one: agents may not share a classical phase reference. Finally, as we have noted, if the analysis is done rigorously, one sees that the ``superselection-violating statistics" of $\Sy$ can be achieved only when there is a localised/absolutely coherent state for $\R$, which
just shifts the problem of absolute coherence from $\Sy$ to $\R$. Only through the mutual coherence concept can this circularity be avoided. The question, then, is whether mutually coherent states exist in all given situations, i.e., for all phase-like quantities.

In more concrete terms, we have seen that, through the $\Y$ construction, absolute quantities and absolutely coherent states can arbitrarily well approximate the statistics of a relational quantity in an invariant state, contingent on a highly localised reference state. We view this statistical equivalence not as ``lifting" in order to show that it can be violated for $\Sy$, but rather as an expression of the fact that the ordinary 
usage of quantum mechanics, with its absolute quantities and absolutely coherent states, captures to a very good degree the true, physical situation represented by invariant quantities of system plus reference, in line with fundamental symmetry requirements.

The situation of ``lacking a phase reference", in our conception, pertains not to the lack of shared knowledge of physicists, but to the physical scenario in which the physical system being used as a reference is completely delocalised with respect to phase, for instance, if it is a number state. This gives rise to a ``reduced" description in which the structure of a superselection rule must be enforced. Whether such a reduced description afforded by absolute quantities and absolutely coherent states
does yield what is observed in any given situation is an empirical question. It seems, to us, that there may be situations in which they do not, in which case a ``superselection rule" stronger than that mooted for photon number could be in force. For example, there may be physical situations in which it is impossible for mutually coherent states to arise from unitary evolution of absolutely incoherent product states, making the approximation of relative quantities by absolute ones impossible. A ``strong" conservation law, as 
presented in subsection \ref{subsec:sww} for instance, would have this effect.

Finally, superselection rules, as they arise in quantum field theory, correspond to 
inequivalent representations of the algebra of observables (possible only for systems with infinitely many degrees of 
freedom---also suspicious according to Earman and Butterfield) and entirely different in nature, it would seem, from the kind of constraint arising from the non-observability of absolute quantities. The connection of these with the superselection rules we have discussed in this manuscript remains a task for the future. 

We conclude this section with a note of caution about the possibility of ``lifting" a superselection rule arising from the indistinguishability of quantum particles.

\subsubsection{A Cautionary Note}

In order to urge a degree of circumspection regarding the idea that reference frames can be used to
overcome superselection rules in general, we discuss now an example based on the 
indistinguishable particle superselection rule in which the physical meaning of a reference frame
is unclear.

Consider a tensor product space $L^2(\mathbb{R})\otimes L^2(\mathbb{R})$ with the action of 
$\mathbb{Z}_2$ which exchanges particle numbering, i.e., $U(a)\Psi(x_1,x_2)=\Psi(x_2,x_1)$ ($a$ is the non-identity element). Indistinguishability requires
that any observable $A$ satisfies $[A,U(a)]=0$ (cf. \cite{frede}). Addend another Hilbert space $\mathbb{C}^2$ with projectors $P{\colvec{1}{0}}$ and 
$P{\colvec{0}{1}}$ with $\mathbb{Z}_2$ action $U^{\prime}(a)P{\colvec{1}{0}} = P{\colvec{0}{1}}$.

Then by demanding invariance of observables only at the level of $\mathcal{H}_1 \otimes \mathcal{H}_2 \otimes \mathbb{C}^2$ one can take an arbitrary 
$A \otimes B \in \mathcal{L}(\mathcal{H}_1 \otimes \mathcal{H}_2)$ and see that 
\begin{equation}
A \otimes B \otimes P{\colvec{1}{0}}+ B \otimes A \otimes P{\colvec{0}{1}}
\end{equation}
defines an invariant quantity (observable). Indeed, this is $\Y(A \otimes B)$ for this (finite) group.
Then, 
\begin{equation}
\ip{\varphi \otimes \phi}{\left( A \otimes B \otimes P{\colvec{1}{0}}+B \otimes A \otimes P{\colvec{0}{1}} \right) \varphi \otimes \phi} = \ip{\varphi}{A \otimes B \varphi}
\end{equation}
for all $\varphi$ and $\phi$ the `phase-localised' state $\phi=\colvec{1}{0}$. Therefore one can introduce a reference system in order to ``measure" particle labelling. In the BRS language,
the corresponding SSR has been ``lifted". However, such a ``reference frame" provided by the $\mathbb{C}^2$ system appears highly artificial and there is a question of whether it makes any physical sense.

\subsection{Reality of Optical Coherence}

In \cite{molmer}, M\o lmer claimed that the representation of laser light using coherent states,
i.e., states of the form 
\begin{equation}
\ket{\beta} := e^{\frac{-\mods{\beta}}{2}}\sum_{n=0}^{\infty} \frac{\beta ^n}{n!}\ket{n},
\end{equation}
while being legitimate for the purposes of calculation, does not reflect the true state of affairs. Actually, he claimed, that, after analysing the internal workings of laser light production in
a physical system, the ``actual" state is (in our notation) $\tau{_{\Sy}}_*(P[\ket{\beta}])$, and (the coherence of)
$\ket{\beta}$ is nothing more than a `convenient fiction'.

The ensuing controversy is well described in \cite{dia} (see also references therein), where a fictional dialogue is presented
between hypothetical physicists representing two groups with contrasting views: those who believe in the ``fact" of optical coherence, and those who view it as fictional. Given the nature of the problem (of the reality of laser coherence), we may re-visit
the controversy and provide a perspective based on the formal framework developed here  (see also \cite{lbm}). 

The issue is whether $\ket{\beta}$ and $\tau{_{\Sy}}_*(P[\ket{\beta}])$ of some laser system $\Sy$ can be empirically distinguished, given that no invariant quantity of $\Sy$ can tell $\ket{\beta}$ from $\tau{_{\Sy}}_*(P[\ket{\beta}])$. As we have seen, however, non-invariant
quantities of $\Sy$ can be used to represent invariant quantities of $\Sy + \R$, contingent on a suitable state of $\R$. The question then is whether there is a feasible physical experiment
in which $\ket{\beta}$ and $\tau{_{\Sy}}_*(P[\ket{\beta}])$, in their role as representing invariant states of $\Sy + \R$, give rise to differing physical predictions.

An absolute phase observable $\F^{\Sy}$ of $\Sy$ (in particular, the canonical phase) is mathematically suitable for separating $\ket{\beta}$ from $\tau{_{\Sy}}_*(P[\ket{\beta}])$.
We may choose also a canonical phase for $\R$, and use $\Y$ to construct the relative phase
observable $\F^{\T}=\Y\circ \F^{\Sy}$. Fixing a sequence $(\beta^{\R}_i) \subset \hir$ of coherent states with the property of becoming increasingly well localised at $0$ as $i$ becomes large,
we then find that 
\begin{align}\label{eq:lalc}
\ip{\beta}{\Fsf^{\Sy}(X)\beta}
 &= \lim_{i \to \infty}\ip{\beta \otimes \beta^{\R}_i}{(\Y\circ \Fsf ^{\Sy})(X)\beta \otimes \beta^{\R}_i} \\
 & \nonumber = \lim_{i \to \infty} \ip{\beta}{ \Gamma_{\beta^{\R}_i}\circ\Y\circ\F^{\Sy}(X) \beta}\\
 & \nonumber = \lim_{i \to \infty}\tr{\Fsf^T(X)\tau_{\T *}(P[\beta \otimes \beta^{\R}_i)}
\end{align}
for each $X \in \mathcal{B}(S^1)$.

From an absolute point of view, absolute coherence (of $\beta^{\R}_i$ for large $i$) is required
to witness absolute coherence of $\ket{\beta}$. From a relational point of view, all that is required (for good approximation of the right hand side by the left) is mutual coherence of the pair $(\ket{\beta}, \ket{\beta_i^{\R}})$. The final line of equation \eqref{eq:lalc} shows that
the limit can be taken using only invariant states of $\Sy + \R$, and that an absolute phase
with an absolutely coherent (coherent) state captures the statistics to arbitrarily good approximation. 

Given that absolute phase observables $\F^{\Sy}$ can be reconstructed in homodyne detection experiments (e.g. \cite{psp}), with the reference state/local oscillator given as a high-amplitude coherent state, we conclude that laser light is mutually coherent. In the high amplitude limit,
the mutual coherence takes on the appearance of absolute coherence for $\ket{\beta}$. We therefore
have a resolution of the puzzle of optical coherence through the application of the `observables are invariants' principle and the concept of mutual coherence.

\section{Summary and Conclusion}
The thesis of this paper is that observable quantities
are invariant under symmetry and that, in quantum mechanical laboratory experiments,
the measured statistics pertain not to some absolute quantity, but rather to an observable, relative quantity,
corresponding to the system and apparatus combined, along with the 
appropriate high localisation limit on the side of the apparatus. 
This is quite general, and not specific
to any particular absolute quantity, though in this paper special attention
has been given to phase, angle and position.

Through our relativisation procedure, we have shown that absolute quantities with absolutely coherent states provide a good account of the observable, relative quantities (with absolutely incoherent states) under high reference localisation. In this sense, the incorporation of a reference frame into the physical description makes it look ``as though" symmetry-violating statistics exist for a subsystem. However, since we argue that the description afforded by
subsystem quantities is theoretical shorthand for the relative description, we do not believe it 
is consistent to argue that symmetry may be violated by the introduction of a reference frame. Indeed, it is the introduction of such a frame that makes symmetry explicit; some quantities simply require two systems for definition, and one of these may me called a reference frame.

Therefore, we agree with prominent physicists (Aharonov/Susskind, Bartlett/Spekkens/Rudolph) that quantum states refer not only to systems 
to which they symbolically refer (i.e., the system under investigation), but also to external physical objects which are not explicitly part of the theoretical description. We have shown that 
complete reference phase delocalisation gives rise to a reduced description formally identical to one in which
a superselection rule is present, giving a new interpretation of the phrase ``lack of a phase reference implies a photon number superselection rule".  The idea that such a rule may be ``lifted" \cite{brs}, as we understand it, corresponds to the observation that a superselection rule may be applied to system-plus-reference, in which case, under reference localisation, it appears as though a superselection rule is not applicable to the system. We believe that, since the ``reduced" description is not a full account of the state of affairs, it is not correct to conclude that 
superselection-rule-``violating" superpositions can be produced or measured. This would indicate that absolute quantities can be measured.

An important question, however, is whether, in all mooted instances of superselection rules, a reference frame may exist which makes it look like the superselection rule can be lifted or overcome. It is empirically the case that for photon number, such a frame does exist. Mutually coherent pairs of systems exist in this case, making absolute phases and coherent states a suitable shorthand description for the true, relative description, with the associated relative phase observable. On the other hand, a reference frame for lifting a superselection rule corresponding to indistiguishability appears highly suspect. As far as we know, it has yet to be settled in a laboratory whether absolute phases conjugate to charge provide an empirically adequate account.

\noindent{\bf Acknowledgements}
Thanks are due to Stephen Bartlett, Rob Spekkens, Terry Rudolph, Dennis Dieks and Guido Bacciagaluppi for helpful conversations, and to Rebecca Ronke for valuable feedback on earlier drafts of this manuscript.

\section*{Appendix}
\noindent
We prove Eq.~\eqref{eq:lim-m}, 
$\bigl\Vert \state{\beta _{A}^1} - \frac{1}{\sqrt{2}} \left\vert \beta \right\rangle \bigr\Vert  \to 0~ \text{as}~ m:=|\beta|^2 \to \infty $.


Let:
\begin{align}
& w_m (n) = |c_n|^2 = \frac{m^n }{n!}e^{-m};\\
& f_m (n)=\left[ \cos{\left(\sqrt{\frac{n}{m}}\frac{\pi}{4}\right)} - \frac{1}{\sqrt{2}}\right]^2;\\
& a_m = \sum_n  w _m (n) f_m (n) = \bigl\||\beta _{A}^1\rangle - \tfrac{1}{\sqrt{2}} |\beta\rangle \bigr\|^2.
\end{align}
 Firstly note that $\bigl|f_m (n)\bigr| \leq 3$. Let $I_{m,k} := \bigl[m-k\sqrt{m}, m + k \sqrt{m}\bigr]~ \text{with }k \in \mathbb{N}$:
\begin{equation}\label{eqconv}
\sum_n w_m (n)f_m (n)  = \sum_{n \in I_k} w_m (n)f_m (n) +  \sum_{n \notin I_k} w_m (n)f_m (n)
\end{equation}

An application of Chebyshev's inequality gives that $p(|n-m| \geq k \sigma) \leq \frac{1}{k^2}$ (where $p$ denotes the  probability distribution $n\mapsto w_m(n)$, $\sigma = \sqrt{m}, k \in \mathbb{N}$); therefore
\begin{equation}
\sum_{n \notin I_k}w_m (n)f_m (n) \leq 3 \sum_{n \notin I_k}w_m (n) \leq \frac{3}{k^2}.
\end{equation}
Exploiting the continuity of cosine, for each $k \in \mathbb{N}$ define $\delta_k$ such that $|\frac{n}{m}-1| < \delta _k$ implies $\bigl|\cos{\bigl(\sqrt{\frac{n}{m}}\frac{\pi}{4}\bigr)} -\cos{\bigl(\frac{\pi}{4}\bigr)}\bigr| < \frac{1}{k}$ (and therefore $f_m (n)< \frac{1}{k^2}$). For each $k \in \mathbb{N}$, let $M=\frac{k^2}{\delta_k ^2}$, and so for $m>M$, $\delta_k > \frac{k}{\sqrt{m}}$. In \eqref{eqconv}, we therefore have that

\begin{equation}
\sum_n w_m (n)f_m (n)  < \Bigl( \sum_{n \in I_k}w_m (n)+3 \Bigr) \frac{1}{k^2} < \frac{4}{k^2}.
\end{equation}
Since $k$ is arbitrary, this proves the result. \qed

\end{document}